\documentclass[letterpaper,11pt]{article}
\usepackage{graphicx}
\usepackage{latexsym}
\usepackage{fullpage}
\usepackage{float}
\usepackage{amssymb,amsmath,amsthm,amsfonts}
\usepackage{xspace}
\usepackage[usenames]{color}

\allowdisplaybreaks[3]      

\makeatletter
\floatstyle{ruled}
\newfloat{fragment}{H}{lop}
\floatname{fragment}{Algorithm}
\renewcommand{\floatc@ruled}[2]{\vspace{2pt}{\@fs@cfont #1.\:} #2 \par
  \vspace{1pt}} 
\makeatother

\newcommand{\mypseudocodelabel}[1]{\hfil}

\DeclareMathSymbol{\qedsymb} {\mathord}{AMSa}{"04}
\newcommand{\qedbox}{\hspace*{0pt}\hfill$\qedsymb$}
\newcommand{\eps}{\varepsilon}
\renewcommand{\epsilon}{\varepsilon}

\newcommand{\ceil}[1]{\left\lceil #1 \right\rceil}
\newcommand{\floor}[1]{\left\lfloor #1 \right\rfloor}

\newcommand{\oct}{\quad\quad}                                   

\newcommand{\LE}{\hbox{LE}}

\newcommand{\RE}{\hbox{RE}}
\newcommand{\minnm}{N}

\newcommand{\R}{\mathbb{R}}
\newcommand{\C}{\mathbb{C}}

\newcommand{\E}{\mathbf{E}}
\newcommand{\U}{\mathbf{U}}
\newcommand{\Var}{\mathbf{Var}}
\renewcommand{\Pr}{\mathbf{Pr}}
\renewcommand{\mod}{\hbox{ mod }}

\newcommand{\tO}{\tilde{O}}

\newcommand{\CorollaryName}[1]{\label{cor:#1}}
\newcommand{\DefinitionName}[1]{\label{def:#1}}
\newcommand{\EquationName}[1]{\label{eq:#1}}
\newcommand{\FactName}[1]{\label{fact:#1}}
\newcommand{\LemmaName}[1]{\label{lem:#1}}
\newcommand{\ObservationName}[1]{\label{obs:#1}}
\newcommand{\SectionName}[1]{\label{sec:#1}}
\newcommand{\TheoremName}[1]{\label{thm:#1}}
\newcommand{\FigureName}[1]{\label{fig:#1}}

\newcommand{\Corollary}[1]{Corollary~\ref{cor:#1}}
\newcommand{\Definition}[1]{Definition~\ref{def:#1}}
\newcommand{\Equation}[1]{Eq.\:\eqref{eq:#1}}
\newcommand{\Fact}[1]{Fact~\ref{fact:#1}}
\newcommand{\Lemma}[1]{Lemma~\ref{lem:#1}}
\newcommand{\Observation}[1]{Observation~\ref{obs:#1}}
\newcommand{\Section}[1]{Section~\ref{sec:#1}}
\newcommand{\Theorem}[1]{Theorem~\ref{thm:#1}}
\newcommand{\Figure}[1]{Figure~\ref{fig:#1}}

\newcommand{\thmabove}{3pt}
\newcommand{\thmbelow}{3pt}

    \newtheoremstyle{mythmstyle}
      {\thmabove}   
      {\thmbelow}   
      {}            
      {}            
      {\bfseries}   
      {. }          
      {2.5pt}       
      {\thmname{#1}\thmnumber{ #2}\thmnote{ \normalfont (#3)}}   
    \theoremstyle{mythmstyle}
    \newtheorem{theorem}{Theorem}[section]\numberwithin{equation}{section}
    \newtheorem{corollary}[theorem]{Corollary}
    \newtheorem{definition}[theorem]{Definition}
    \newtheorem{observation}[theorem]{Observation}
    
    \newtheorem{fact}[theorem]{Fact}

    \newtheorem{lemma}[theorem]{Lemma}

\newcommand{\proofbelow}{3pt}
\newcommand{\afterproof}{\hfill $\blacksquare$ \par \vspace{\proofbelow}}
\newcommand{\aftersubproof}{\hfill $\Box$ \par \vspace{\proofbelow}}
\renewenvironment{proof}{\noindent\textbf{Proof.}\,}{\afterproof}

\newenvironment{proofof}[1]{\noindent\textbf{Proof} \,(of #1).\,}{\afterproof}

\newcommand{\poly}{\mathop{{\rm poly}}}
\newcommand{\lsb}{\mathop{{\rm lsb}}}
\renewcommand{\th}{\ifmmode{^{\textrm{th}}}\else{\textsuperscript{th}\ }\fi}

\newcommand{\comment}[1]{}

\newcommand{\TODO}[1]{}

\makeatletter
\renewcommand{\maketitle}{
    \begin{center}
    \begin{minipage}[t]{6.5in}
    \begin{center}
    \vspace*{15pt}
    {\LARGE\bf \@title \par}
    \vspace*{10pt}
    \vspace*{20pt}
    {\Large \@author}
    \vspace*{3pt}
    \end{center}
    \end{minipage}
    \end{center}
}
\makeatother

\renewcommand{\thefootnote}{\fnsymbol{footnote}}
\title{Revisiting Norm Estimation in Data Streams}

\author{Daniel M. Kane\footnotemark[2]\oct Jelani
  Nelson\footnotemark[3]\oct
  David P. Woodruff\footnotemark[4]}

\date{}

\newcommand{\eqdef}{\mathbin{\stackrel{\rm def}{=}}}

\begin{document}

\footnotetext[1]{Harvard University, Department of
  Mathematics. \texttt{dankane@math.harvard.edu}. Supported by a
  National Defense Science and Engineering Graduate (NDSEG)
  Fellowship.}
\footnotetext[2]{MIT Computer Science and Artificial Intelligence
  Laboratory. \texttt{minilek@mit.edu}. Supported by a National
  Defense Science and Engineering Graduate (NDSEG) Fellowship. Much of this
  work was done while the author was at the IBM Almaden
  Research Center.}
\footnotetext[3]{IBM Almaden Research Center, 650 Harry Road, San
  Jose, CA, USA. \texttt{dpwoodru@us.ibm.com}.}

\maketitle

\renewcommand{\thefootnote}{\arabic{footnote}}
\setcounter{page}{0}

\begin{abstract}
\thispagestyle{empty}
\noindent
We revisit the problem of $(1\pm\eps)$-approximating the $L_p$
norm, for real $p$ with $0 \le p \le 2$, of
a length-$n$ vector updated in a length-$m$ stream with 
updates to its coordinates. We assume the updates are integers
in the range $[-M, M]$. 
We prove new bounds on the space and time complexity
of this problem. In many cases our results are optimal.
\begin{enumerate}
\item We give a $1$-pass space-optimal algorithm for $L_p$-estimation
  for constant $p$, $0 < p < 2$. Namely, we give an algorithm using
  $O(\eps^{-2}\log(mM) + \log\log n)$ bits
of space to estimate $L_p$ within relative error $\epsilon$ with 
constant probability. Unlike previous
algorithms which achieved
optimal dependence on $1/\eps$, but suboptimal 
dependence on $n$ and $m$, our algorithm
{\it does not use a generic pseudorandom generator (PRG)}. 
\item We improve the $1$-pass lower bound on the space to 
$\Omega(\eps^{-2}\log(\eps^2 N))$ bits for real constant $p \geq
0$ and $1/\sqrt{N} \le \eps \le 1$, where $N = \min\{n,m\}$.
If $p > 0$, the bound improves to $\Omega(\min\{N, \eps^{-2}\log(\eps^2mM)\})$.
Our bound is based on showing a direct sum property for the $1$-way 
communication of the gap-Hamming problem.
\item For $p=0$, we give
an algorithm which matches our space lower bound up to an
$O(\log(1/\eps) +
\log\log(mM))$ factor. Our algorithm is the first space-efficient
algorithm to achieve $O(1)$ update and reporting time. 
Our techniques also 
yield a $1$-pass $O((\eps^{-2}+\log N) \log \log N + \log\log
n)$-space algorithm for estimating $F_0$, the
number of distinct elements in the update-only model, with $O(1)$
update and reporting time. This
significantly improves upon previous algorithms achieving this amount of space,
which suffered from $\tO(\eps^{-2})$ worst-case update time.
\item We reduce the space complexity of dimensionality reduction in a stream
with respect to the $L_2$ norm by replacing the use of Nisan's PRG in Indyk's 
algorithm with an improved PRG built by efficiently combining an extractor of 
Guruswami, Umans,
and Vadhan with a PRG construction of Armoni. 
The new PRG stretches a seed of
$O((S/(\log(S) - \log\log(R) + O(1)))\log R)$ bits to $R$ bits fooling 
space-$S$ algorithms for any $R = 2^{O(S)}$, improving the $O(S\log R)$ 
seed length of Nisan's PRG. Many existing algorithms rely
on Nisan's PRG, and this new PRG reduces the space complexity
of these algorithms. 
\end{enumerate}
Our results immediately imply various separations between the complexity of
$L_p$-estimation in different update models, one versus multiple passes, and
$p = 0$ versus $p > 0$. 
\end{abstract}

\newpage

\section{Introduction}
Computing over massive data streams is increasingly important. 
Large data sets, such as sensor
networks, transaction data, the web, and network traffic, have grown at
a tremendous pace. It is impractical for most devices to store
even a small fraction of the data, and this necessitates the design of
extremely efficient algorithms. Such
algorithms are often only given a single pass over the
data, e.g., it may be expensive to read the contents of
an external disk multiple times, and in the case of an internet
router, it may be impossible to make multiple passes.

Even very basic statistics of a data set cannot be computed
exactly or deterministically in this model, and
so algorithms must be both approximate and probabilistic. This model
is known as the streaming model and has become
popular in the theory community, dating back to the works
of Munro and Paterson \cite{munro} and Flajolet and Martin
\cite{FM83}, and resurging with the
work of Alon, Matias, and Szegedy \cite{AMS99}. For a survey
of results, see the book by Muthukrishnan
\cite{Muthu}, or notes from Indyk's course \cite{IndykCourse}. 

A fundamental problem in this area is that of norm
estimation \cite{AMS99}. Formally, we have a vector $a =
(a_1,\ldots,a_n)$ initialized as
$a = \vec{0}$, and a stream of $m$ updates, where an update $(i,v)\in
[n] \times \{-M,\ldots,M\}$
causes the change $a_i \leftarrow a_i + v$. If the $a_i$ are
guaranteed to be non-negative at all times, this is called the {\em
  strict
  turnstile model}; else it is called the {\em turnstile model}.
Our goal is to output a $(1\pm \eps)$-approximation to the value
$L_p(a) = \left (\sum_{i=1}^n |a_i|^p \right )^{1/p}$.
Sometimes this problem is posed as estimating $F_p(a) = L_p^p(a)$,
which is called the $p$-th frequency moment of $a$. A large body of
work has been done in this area, see, e.g., the references in
\cite{IndykCourse,Muthu}.

When $p = 0$, $L_0\eqdef |\{i \mid
a_i \neq 0\}|$, and it is called the ``Hamming norm''. 
In an update-only stream,
i.e., where updates $(i,v)$ always have $v = 1$,
this coincides with the well-studied problem of 
estimating the number
of distinct elements, which is
useful for query
optimizers in the context of databases, internet routing, and
detecting Denial of Service attacks \cite{abrs}.
The Hamming norm is
also useful in
streams with deletions, for which it can be used to measure
the dissimilarity of two streams, which is useful for packet
tracing and database auditing \cite{CDIM03}.

\subsection{Results and Techniques}
We prove new upper and lower bounds on the space and time complexity
of $L_p$-estimation
for $0 \leq p \leq 2$
\footnote{When $0 < p < 1$, $L_p$ is not a norm since it does not
satisfy the triangle inequality, though it is still well-defined.}.
In many cases our results are optimal. We shall use the term update time
to refer to the per item processing time in the stream, while we use
the term reporting time to refer to the time to output the estimate at
any given point in the stream. In what follows in this section, and
throughout the rest of the paper, we
omit an implicit additive $\log\log n$ which exists in all the $L_p$
space upper and lower
bounds. In strict turnstile and turnstile streams, the additive term
increases to $\log\log(nmM)$.
Each following subsection describes an overview of our
techniques for a problem we consider, and
a discussion of previous work.  A table listing all our bounds is also
given in \Figure{our-bounds}.

\begin{figure}
\begin{center}
\small
\begin{tabular}{|l|l|l|l|l|}
\hline
Problem & upper bound & lower bound & update & reporting\\
\hline
$L_p$ & $O(\eps^{-2}\log(mM))$ &
$\Omega(\eps^{-2}\log(mM))$ & $\tO(\eps^{-2})$ &
$O(1)$\\
\hline
$L_0$ (1-pass) & $O(\eps^{-2}(\log(1/\eps) + \log\log(mM))\log N)$ &
$\Omega(\eps^{-2}\log N)$ & $O(1)$ & $O(1)$ \\
\hline
$L_0$ (2-pass) & $O(\eps^{-2}(\log(1/\eps) + \log\log(mM))+\log N)$ &
$\Omega(\eps^{-2}+\log N)^*$ & $O(1)$ & $O(1)$ \\
\hline
$F_0$ & $O(\eps^{-2}\log\log N + \log(1/\eps)\log N)$ & $\Omega(\eps^{-2} + \log
N)^{**}$ & $O(1)$ & $O(1)$ \\
\hline
$L_2\rightarrow L_2$ & $O(\eps^{-2}\log (nM/(\eps\delta)) \log (n/(\eps \delta)) \log
(1/\delta)/\log(1/\eps))$ &
$\Omega(\eps^{-2}\log(nM))$ & *** & $O(1)$\\
\hline
\end{tabular}
\caption{\small Table of our results. The 2nd and 3rd columns are
  space bounds, in bits, and the 1st row is for $0<p<2$. The last two
  columns are time.
All bounds above are ours,
  except for * \cite{AMS99,BC09} and ** \cite{AMS99, BC09,IW03,
    W04,jks07,WoodruffThesis}. $N$ denotes $\min\{n,m\}$. All lower
 bounds hold for $\eps$ larger than some threshold (e.g., they
never go above $\Omega(N)$), and all bounds are stated for a
 desired constant probability of success, except for the last row.  In
 the last row,
 $1-\delta$ success probability is desired for $\delta =
 O(1/t^2)$, where we want to do $L_2\rightarrow L_2$ dimensionality
 reduction of $t$ points in a stream, and thus need $\delta =
 O(1/t^2)$ to union bound for all pairwise distances to be
 preserved (the space shown is for one of the $t$ points). $F_0$
 denotes $L_0$ in
 update-only streams.  For ***, the time is polynomial in the
 space. Note for rows 1 and 5, the
 reporting times are $O(1)$ since we can recompute the estimator
 during updates.}\FigureName{our-bounds}
\end{center}
\end{figure}

\subsubsection{New algorithms for $L_p$-estimation, $0 < p < 2$}\SectionName{lp-intro}
Our first result is the first $1$-pass space-optimal algorithm for
$L_p$-estimation, $0 < p < 2$. Namely, we give an algorithm using $O(\eps^{-2}\log(mM))$ bits
of space to estimate $L_p$ within relative error $\epsilon$ with 
constant probability. Unlike the previous algorithms of Indyk and Li which
achieved optimal dependence on $1/\eps$, but suboptimal dependence on
$n$ and $m$ \cite{Indyk06,Li08b}, our algorithm uses only $k$-wise
independence and
{\it does not use a generic pseudorandom generator (PRG)}. 
In fact, the previous algorithms failed to achieve space-optimality
precisely because of the use of a PRG \cite{Nisan92}.
Our main technical lemma shows that $k$-wise independence
preserves the properties of sums of $p$-stable random variables in a
useful way.  This is the
first example of such a
statement outside the case $p = 2$.  PRGs
are a central tool in the design of streaming
algorithms, and Indyk's algorithm has become the canonical
example of a streaming algorithm for which
no derandomization more efficient than via a generic PRG was
known. We believe that removing this heavy hammer 
from norm estimation is an important step forward in the
derandomization of streaming algorithms, and that our
techniques may spur improved derandomizations of other streaming
algorithms. 

To see where our improvement comes from, let us recall Indyk's algorithm
\cite{Indyk06}. That algorithm maintains $r =
\Theta(1/\eps^2)$ counters $X_j = \sum_{i=1}^n a_iX_{i,j}$, where the
$X_{i,j}$ are
i.i.d. from a discretized {\em $p$-stable distribution}.  A
$p$-stable
distribution $\mathcal{D}$ is a distribution with the property that,
for all vectors $a\in\R^n$ and i.i.d. random variables
$\{X_i\}_{i=1}^n$ from $\mathcal{D}$, it holds that
$\sum_{i=1}^n a_iX_i \sim ||a||_p X$, where $X \sim \mathcal{D}$.
His algorithm then 
returns the median of the $|X_j|$.  The main issue with Indyk's
algorithm, and also a later algorithm of Li \cite{Li08b}, is that the
amount of randomness
needed to generate the $X_{i,j}$ is $\Omega(N/\eps^2)$. A
polylogarithmic-space algorithm thus cannot afford to store all the
$X_{i,j}$.
Indyk remedied this problem by using Nisan's PRG \cite{Nisan92}, but
at the cost of multiplying his space by a $\log(N/\eps)$ factor.

Our algorithm, like those of Indyk and Li, is also based on $p$-stable
distributions.  However, we do not use the median estimator of Indyk, or
the geometric mean or harmonic mean estimators of Li.  Rather, we give
a new estimator which we show can be derandomized using only $k$-wise
independence for small $k$ (specifically, $k =
O(\log(1/\eps)/\log\log\log(1/\eps))$ --- any $k = O(1/\eps^2)$ would
have given us a space-optimal algorithm, but smaller $k$ gives smaller
update time).  We first show that the median
estimator of Indyk gives a constant-factor approximation of $L_p$ with
arbitrarily large constant probability as long as $k,r$ are chosen
larger than some constant.  Even this was previously not known.
Once we have a value $A$ such that
$||a||_p/A = \Theta(1)$, we then give an estimator that can
$(1\pm\eps)$-approximate $||a||_p/A$ using only $k$-wise
independence.  Despite the two-stage nature of our algorithm (first
obtain a constant-factor approximation to $||a||_p$, then
refine to a $(1\pm\eps)$-approximation), our
algorithm is naturally implementable in one pass.

Other work on $L_p$-estimation includes \cite{CG07},
though their scheme uses $\Omega(\eps^{-2-p}\poly\log(mM))$ space.
For $p>2$, space polynomial
in $n$ is necessary and sufficient \cite{AMS99, bgks06, bjks02b, cks03, IW}.

\subsubsection{Tight space lower bounds for $L_p$-estimation}
To show optimality of our $L_p$-estimation algorithm, for $p>0$ we
improve the
space lower bound to $\Omega(\min\{N, \eps^{-2}\log(\eps^2mM)\})$
bits. For $p=0$, we show a lower bound of
$\Omega(\eps^{-2}\log(\eps^2 N))$.  Here, $1/\sqrt{N} \le \eps \le 1$,
with $N = \min\{n,m\}$.
The previous lower bound in both cases is $\Omega(\eps^{-2}+ \log
N)$,
 and is the result of a sequence of work \cite{AMS99, IW03, W04,
   BC09}.
See \cite{jks07, WoodruffThesis} for simpler proofs.
Since Thorup and Zhang \cite{ThorupZhang04} give a
time-optimal variant of the $L_2$-estimation sketch of 
Alon, Matias, and Szegedy \cite{AMS99}, our work closes the problem of
$L_2$-estimation, up to
constant factors.
Our bound holds even 
when each coordinate is updated twice, implying that the space of
Feigenbaum et al.\ \cite{FKSV02} for
$L_1$-difference estimation is optimal.  Our
lower bound is also the first to give a logarithmic dependence on $mM$
(previously only an $\Omega(\log\log(mM))$ bound was
known by a
reduction from the communication complexity of \textsc{Equality}).

Our lower bounds are based upon embedding multiple
geometrically-growing hard instances for estimating $L_p$ in an
insertion-only stream into a stream, and using the deletion property
together with the geometrically-growing property to reduce the problem
to solving a single hard instance.
More precisely, a hard instance for $L_p$ is
based on a reduction from a two-party communication game in which the
first party, Alice, receives a string $x \in \{0,1\}^{\eps^{-2}}$, and
Bob an index $i \in [\eps^{-2}]$, and Alice sends a single message to
Bob who must output $x_i$ with constant probability. This problem,
known as indexing, requires $\Omega(\eps^{-2})$ bits of space. To
reduce it to estimating $L_p$ in an insertion-only stream, there is a
reduction \cite{IW03, W04, WoodruffThesis} through the gap-Hamming
problem for which Alice creates a
stream $\mathcal{S}_x$ and Bob a stream $\mathcal{S}_i$, with the
property that either $L_p(\mathcal{S}_x \circ \mathcal{S}_i) \geq
\eps^{-2}/2 + \eps^{-1}/2$, or $L_p(\mathcal{S}_x \circ \mathcal{S}_i)
\leq \eps^{-2}/2 - \eps^{-1}/2$. 
Here, ``$\circ$'' denotes
concatenation of two streams.
Thus, any $1$-pass streaming
algorithm which $(1 \pm \eps)$-approximates $L_p$ requires space which
is at least the communication cost of indexing, namely,
$\Omega(\eps^{-2})$.

We instead consider the augmented-indexing problem. Set $t = 
\Theta(\eps^{-2}\log (\eps^2 N))$. 
We give Alice a string $x \in \{0,1\}^t$ and Bob
both an index $i \in [t]$ together with a subset of the bits $x_{i+1},
\ldots, x_t$. This
problem requires $\Omega(t)$ bits of communication if Alice
sends only a single message to Bob \cite{BJKK04, MNSW98}. Alice splits
$x$
into $b = \eps^2 t$ equal-sized blocks $X_0, \ldots, X_{b-1}$. In the
$j$-th block
she uses the $\eps^{-2}$ bits assigned to it to create a stream
$\mathcal{S}_{X_j}$ that is similar to what she would have created in
the insertion-only case, but each non-zero item is duplicated
$2^j$ times. Given $i$, Bob finds the block $j$ for which it
belongs, and creates a stream $\mathcal{S}_i$ as in the insertion-only
case, but where each non-zero item is duplicated $2^j$
times. Moreover, Bob can create all the streams $\mathcal{S}_{X_{j'}}$
for blocks $j'$ above block $j$. Bob inserts all of these latter
stream items as deletions, while Alice inserts them as
insertions. Thus, when running an $L_p$ algorithm on Alice's list of
streams followed by Bob's, all items in streams $\mathcal{S}_{X_{j'}}$
vanish. Due to the duplication of non-zero coordinates, approximating
$L_p$ well on the entire stream corresponds to approximating $L_p$
well on $\mathcal{S}_{X_{j}} \circ \mathcal{S}_i$, and thus a
$(1\pm\eps)$-approximation algorithm to $L_p$ can be used to solve
augmented-indexing. For $p > 0$, we can do better by using the universe
size to our advantage. Instead of duplicating each coordinate $2^j$
times in the $j$-th block, we scale each coordinate's frequency by
$2^{j/p}$ in
the $j$-th block. For constant $p > 0$, this has a similar effect as
duplicating coordinates.  Our technique can be viewed as showing
a direct sum property for the one-way communication complexity of the
gap-Hamming problem.

For $p \neq 1$, our lower bound holds even in
the strict turnstile model.
The assumption that $p \neq 1$ in the strict turnstile 
model is necessary, since one can easily compute $L_1$ 
exactly in this model 
by maintaining a counter. Also, as it is known
that $L_0$ can be estimated
in $\tilde{O}(\eps^{-2} + \log N)$ bits of
space\footnote{We say $f
  = \tilde{O}(g)$ if $f = O(g\cdot \mathrm{polylog}(g))$.} in the
update-only model, our lower
bound establishes the first {\em separation} of estimating $L_0$ in
these two well-studied models.  Our technique also gives the best
known lower bound for additive
approximation of the entropy in the strict turnstile
model, improving the $\Omega(\eps^{-2})$ bound
that follows\footnote{Their lower bound is
  stated against multiplicative approximation, but the
  additive lower bound easily follows from their proof.} from the
work of \cite{CGM07} to
$\Omega(\eps^{-2}\log(N)/\log(1/\eps))$.  Their lower bound though
also holds in the
update-only model.
Additive estimation of entropy can be used to additively
approximate conditional entropy
and mutual information, each of which cannot be multiplicatively
approximated in small space \cite{IndMcGreg08}. Variants of our techniques were also applied to establish tight bounds for linear algebra problems in a stream \cite{CW08}.

\subsubsection{Near-optimal algorithms for $L_0$ in turnstile and
  update-only models}\SectionName{l0-proofsketch}
In the case of $L_0$, we give a $1$-pass 
algorithm which is nearly optimal in the most general turnstile model.
Our algorithm needs only
$O(\eps^{-2}\log(\eps^2 N)(\log(1/\eps) + \log\log(mM)))$
bits of space, and has optimal $O(1)$ update and reporting time. 
Given our lower bound and a folklore $\Omega(\log\log(nmM))$ lower
bound, our space upper bound is tight up to
potentially the $\log(1/\eps)$ term, and the $\log\log(mM)$ term being
multiplicative instead of additive. Note our algorithm implies a
separation between $L_0$ estimation and $L_p$ estimation, $p>0$, since
we show a logarithmic dependence on $mM$ is necessary for the
latter. Our algorithm
improves on prior work
which either (1) both assumes the weaker strict turnstile model and uses an
extra $\log(mM)$ factor in space
\cite{Ganguly07}, or (2) has space complexity which is worse by at least a 
$\min((\log^2 N \log^2 m)/\log (mM)), 1/\eps)$ factor
\cite{CDIM03, CG07}. Also, all previous algorithms had at least a
logarithmic dependence on $mM$, and none had $O(1)$ update time.
Here we assume the word RAM model (as did previous work, except
\cite{BJKST02}, for which we later translate their update times to the
word RAM model), where standard arithmetic and bit
operations on $\Omega(\log(nmM))$-bit
words take constant time. Furthermore, we show that our algorithm has
a natural $2$-pass implementation using
$O(\eps^{-2}(\log(1/\eps)+\log\log(mM)) + \log N)$ space. Given our
$1$-pass lower bound, this implies the first known {\em separation}
for $L_0$ between $1$ and $2$ passes.  Furthermore, due to a recent
breakthrough
of Brody and Chakrabarti \cite{BC09}, our $2$-pass algorithm is
optimal up to $O(\log(1/\eps)+\log\log(mM))$ for any constant number
of passes.
Finally, we give an algorithm for estimating 
$L_0$ in the update-only model, i.e., the number of distinct elements,
with $O((\eps^{-2}+ \log N)\log \log N)$ bits of space and $O(1)$
update and reporting time. Our space is optimal up to the $\log \log
N$ \footnote{Our gap to optimality is even smaller for $\eps$ small.
  See \Figure{our-bounds}.}, while our time is optimal. This greatly improves the
time complexity of the only previous algorithms (the 2nd and 3rd
algorithms\footnote{Their 3rd algorithm has $O(\log(1/\eps) + \log\log
  N)$ amortized
  update time, but $\tO(\eps^{-2})$ worst-case update time.} of
\cite{BJKST02}) with this space complexity, from $\tO(\eps^{-2})$
to $O(1)$.

We sketch some of our techniques, and the differences with
previous work.  In both our $1$-pass $L_0$ algorithms (update-only and
turnstile),
we run in parallel a
a subroutine to obtain a value $R = \Theta(L_0)$. We also in parallel
pairwise independently subsample the universe at a rate of $1/2^j$ for
$j=1,\ldots,\log(\eps^2N)$ (note that $L_0 \le N$) to create
$\log(\eps^2 N)$ substreams. This subsampling can be done by hashing
into $[N]$ then sending item $i$ to level $\lsb(h(i))$, where $\lsb$
is the least significant bit.  At each level $j$ we feed the $j$th
substream into a subroutine which approximates $L_0$ well when promised
$L_0$ is small.  We then base our estimator on the level $j$ with
$R/2^j = \Theta(1/\eps^2)$, since the $L_0$ of that substream will
be $(1\pm\eps)L_0/2^j$ with good probability, so that we can scale
back up to get $(1\pm\eps)L_0$. The idea of subsampling the stream 
and using an estimate from some appropriate level is
not new, see, e.g., 
\cite{BJKST02,Ganguly07,GT01,PavanTir07}. For example, the
best known algorithm for $L_0$ estimation in the strict turnstile
model, due to Ganguly \cite{Ganguly07}, follows
this high-level approach. We now explain where our techniques differ.

First we discuss the turnstile model.  We develop a subroutine
using only $O(\log(N)\log\log(mM))$ space to obtain $R$.
Previously, no subroutine using $o(\log(N)\log(mM))$ space was
known.  Next, at level $j$ we play a balls-and-bins game where we
throw $A$ balls into $1/\eps^2$ bins $k$-wise independently for
$k=O(\log(1/\eps)/\log\log(1/\eps))$, then base our estimator on the
outcome of this random process. This is similar to Algorithm II of
Ganguly \cite{Ganguly07}, which itself was based on the second
algorithm of
\cite{BJKST02}. The $A$ balls are the
$L_0$-contributors mapped to level $j$, and the $1/\eps^2$ bins are
counters.
 In Ganguly's algorithm, he bases his estimator
on the number of bins receiving exactly one ball, and develops a
subroutine to use inside each bin which detects this. However, this
subroutine requires $O(\log(mM))$ bits and only works only in the strict
turnstile model. We overcome both issues by basing our estimator on the
number of bins receiving {\em at least} one ball. To detect
if a bin is hit, we cannot simply keep frequency sums
since colliding balls could have frequencies of opposite sign and
cancel each other.  Instead, each bin
maintains the dot product of frequencies with
a random vector over a suitably large finite field. This allows us to
both reduce the $mM$ dependence to doubly logarithmic, and work in the
turnstile model. Also, one time bottleneck is  
evaluating the $k$-wise independent hash
function, but we observe that this can be done in $O(1)$ time using a scheme
of Siegel \cite{Siegel04} after perfectly hashing the universe down to
$[1/\eps^4]$.
Furthermore, we
non-trivially extend the analysis
of \cite{BJKST02} to analyze throwing $A$ balls into $1/\eps^2$ bins
with $k$-wise independence when potentially $A\ll 1/\eps^2$, to deal with the
case when $L_0\ll 1/\eps^2$ since then there is no $j$ with $L_0/2^j =
\Theta(1/\eps^2)$. The algorithm of \cite{BJKST02} worked by
estimating the probability that a single bin, say bin 1, is hit.
Since their random variable had constant expectation, the variance was
constant for free.
In our case, the number of non-empty bins is non-constant (it
grows with $A$), so we need to prove a sharp bound on
the variance. Ganguly deals with small $L_0$ via a separate
subroutine, which itself requires $\Omega(\log(1/\eps))$ update time,
and uses space suboptimal
by a $\log(mM)$ factor.

Now we discuss update-only streams. By convention, $L_0$
in the update-only case is typically referred to as $F_0$.
As in our $L_0$ algorithm, we use
a balls-and-bins approach, though with a major difference. Our
key to saving space is that all
$\log(\eps^2N)$ levels {\em share the same bins}, and each bin only
records the deepest level $j$ in which it was hit. Thus, we
can maintain all bins in the algorithm using
$O(\eps^{-2}\log\log(\eps^2N))$
space as opposed to $O(\eps^{-2}\log(\eps^2 N))$.
 An obvious obstacle in our algorithm is that when counting the number
 of bins
 hit at level $j$, our count is obscured by bins that were
hit both at level $j$ and at some deeper level.  Since each bin only
keeps track of the deepest level it was hit in, we lose information
about shallow levels. Our analysis then leads us to a
more general random process, where there are $A$ ``good'' balls and
$B$ ``bad'' balls, and we want to understand the number of ``good
bins'', i.e. bins  hit by
at least one good ball and no bad balls.  We show that
the truly random
process is well-approximated even when all balls are thrown
$k$-wise
independently. The good balls are the distinct items at
level $j$, and the bad ones are those at deeper levels.  As long as
$R/2^j = \Theta(1/\eps^2)$, we have both that (1) $A/B = 1\pm O(\eps)$
with good probability (by Chebyshev's inequality), and (2)
$A = (1\pm O(\eps))F_0/2^r$ (also by Chebyshev's inequality).  Item
(1) allows us to approximate the expected number of good bins as a
function of just $A$, then invert to get $A$.
Item (2) allows us to scale
our estimate for $A$ to recover an estimate for $F_0$.
Our scheme is different from \cite{BJKST02}, which did not
subsample the universe, and based its estimator on the
fraction of hash functions in a $k$-wise independent family which
map at least one ball to bin $1$ (out of $R$ bins). To estimate this
fraction well,
\cite{BJKST02} required $\tO(1/\eps^2)$ update time.  Our update time,
however, is constant.

\subsubsection{Other results: embedding into a normed space and an
  improved PRG}\SectionName{intro-prg}
Dimensionality reduction is a useful technique for mapping a set of
high-dimensional points to a set of low-dimensional points with similar
distance properties. This technique has numerous applications in theoretical
computer science, especially the Johnson-Lindenstrauss embedding \cite{JL84} for
the $L_2$ norm. Viewing the underlying vector of the data stream as a point
in $n$-dimensional space, given two points $a, b \in [M]^n$ in two different
streams, 
one can view our sketches $S_a$, $S_b$ 
as a type of dimensionality reduction, so that $||a-b||_p$ can be
estimated from the sketches $S_a$ and $S_b$. Unfortunately, our sketches (as well
as previous sketches for estimating $L_p$), are not in a normed space, 
and this could restrict the applications of it
as a dimensionality reduction technique. This is because there are many algorithms,
such as nearest-neighbor algorithms, designed for normed spaces. Indyk \cite{Indyk06}
overcomes this for the important case of $L_2$ by doing the following. His
streaming algorithm maintains $Ta$, where $a$ is the vector in the stream, 
and $T$ is an implicitly defined sketching matrix whose entries
are pseudorandomly generated normal random variables. From
$Ta$ and $Tb$, $||Ta-Tb||_2$ gives a $(1\pm \eps)$-approximation to
$||a-b||_2$, and
this gives an embedding into a normed space. The space is
$O(\eps^{-2}\log (nM/(\eps\delta)) \log (n/(\eps \delta)) \log
(1/\delta))$ bits, where
$\delta$ is the desired failure probability.

We reduce the space complexity of this scheme by a $\log(1/\eps)$
factor by replacing
the use of Nisan's PRG \cite{Nisan92} in
Indyk's algorithm
with an improved version of Armoni's PRG
\cite{Armoni98}. 
When writing his original PRG construction, time- and space-efficient
optimal extractors were not known, so his PRG would only improve
Indyk's use of Nisan's PRG when $\eps$ was sufficiently small.  
We show that a recent optimal extractor construction of Guruswami, Umans,
and Vadhan \cite{GUV07} can be modified to be computable in
linear space and thus fed into Armoni's construction to improve his
PRG.  Specifically, the improved Armoni PRG stretches a seed of
$O((S/(\log(S) - \log\log(R) + O(1)))\log R)$ bits to $R$ bits fooling 
space-$S$ algorithms for any $R = 2^{O(S)}$, improving the $O(S\log R)$ 
seed length of Nisan's PRG. As many existing streaming algorithms rely
on Nisan's PRG, using this PRG instead reduces the space complexity
of these algorithms. 

Much of the reason the GUV extractor implementation
described in \cite{GUV07} does not use linear space is its reliance
on Shoup's algorithm \cite{Shoup90} for finding irreducible
polynomials over small finite fields, and in fact most of the
implementation modifications we make are so that the GUV
extractor can avoid all calls to Shoup's algorithm.  

\subsection{Other Previous Work}\label{sec:relatedWork}
Here we discuss other previous work not mentioned above.
$L_0$-estimation in the update-only model 
was first considered by Flajolet and Martin \cite{FM83}, who assumed
the existence of hash functions with properties that are unknown to
exist to obtain a constant-factor approximation. The ideal hash
function assumption was later removed in \cite{AMS99}.
Bar-Yossef et al.\ \cite{BJKST02} provide the best previous
algorithms, described above in \Section{l0-proofsketch}. 
Estan, Varghese, and Fisk \cite{EVF06} give an algorithm which assumes
a random oracle and a $O(1)$-approximation
to $L_0$, and seems to achieve $O(\eps^{-2}
\log N)$ space with $O(\log N)$ update time, though a formal
analysis is not given. There is a previous algorithm for $L_0$-estimation in
the turnstile model due to Cormode
et al.\ \cite{CDIM03}  which
needs to store $O(\eps^{-2})$ random
variables from a $p$-stable distribution for $p=O(\eps/\log(mM))$ and
has $O(\eps^{-2})$ update time, though the precision needed to hold
$p$-stable samples for such small $p$ is $\Omega(\eps^{-1}\log N)$,
making their overall space dependence on $1/\eps$ cubic. Work of
Cormode and Ganguly \cite{CG07} implies an algorithm
with $O(\eps^{-2} \log^2N \log^2(mM))$ space and $O(\log^2 N\log(mM))$
worst-case update time in the turnstile model.

\subsection{Notation}

For integer $z>0$, $[z]$ denotes the set $\{1,\ldots,z\}$.
For our upper bounds we let $[U]$ denote the universe.
That is, upon
receiving an update $(i,v)$ in the stream, we assume
$i\in[U]$.  We can assume $U =
\min\{n,O(m^2)\}$ with at most an additive $O(\log\log n)$ in all
our $L_p$ space upper bounds.  Though this is somewhat standard,
achieving an additive $O(\log\log n)$ as opposed to $O(\log n)$ is
perhaps less well-known, so we include justification
in \Section{small-universe}.  All our space upper and lower bounds are
measured in bits.

We also use $\lsb(x)$ to denote the least significant bit of an
integer $x$ when written in binary.  We note when $x$ fits in a
machine word,
$\lsb(x)$ can be computed in $O(1)$ time
\cite{Brodnik93,FredmanWillard93}.

\section{$L_p$ Estimation $(0 < p < 2)$}\SectionName{lp-est}
Here we describe our space-optimal $L_p$ estimation algorithm
mentioned in \Section{lp-intro}, as well as the approach mentioned in
\Section{intro-prg} of using an improved PRG.


\subsection{An Optimal Algorithm}

We assume $p$ is a fixed constant. 
Some constants in our
asymptotic notation are functions of $p$.  We also assume
$||a||_p > 0$; $||a||_p = 0$ is detected when $A=0$ in
\Figure{lpalg}.  Finally, we assume $\eps \ge 1/\sqrt{m}$.  Otherwise,
the trivial solution of keeping the entire stream in memory requires
$O(m\log(UM)) = O(\eps^{-2}\log(NM)) = O(\eps^{-2}\log(mM))$ space.
The main theorem of this section is the following.

\begin{theorem}\TheoremName{main-optimallp}
Let $0<p<2$ be a fixed real constant. The algorithm of \Figure{lpalg}
uses space $O(\eps^{-2}\log(mM))$ and outputs
$(1\pm\eps)||a||_p$ with probability at
least $2/3$.
\end{theorem}

\begin{figure*}
\begin{center}
\fbox{
\parbox{6.375in} {
\small
\begin{enumerate}
\addtolength{\itemsep}{-1mm}
\item Maintain $A_j = \sum_{i=1}^n a_iX_{i,j}$ for $j\in[r]$, $r =
  \Theta(1/\eps^2)$. Each $X_{i,j}$ is distributed according to
$\mathcal{D}_p$.  For fixed $j$, the $X_{i,j}$ are $k$-wise
  independent with $k = \Theta(\log(1/\eps)/\log\log\log(1/\eps))$.
  For $j\neq j'$, the seeds used to generate the $\{X_{i,j}\}_{i=1}^n$
  and $\{X_{i,j'}\}_{i=1}^n$ are pairwise independent.
\item Let $A = \mathrm{median}\{|A_j|\}_{j=1}^r$. Output
$A\cdot \left(-\ln\left(\frac{1}{r}\sum_{j=1}^r
    \cos\left(\frac{A_j}{A}\right) \right) \right)^{1/p}$.
\end{enumerate}
}}
\end{center}
\vspace{-.2in}
\caption{$L_p$ estimation algorithm pseudocode, $0<p<2$}\FigureName{lpalg}
\end{figure*}

To understand the first step of \Figure{lpalg}, we recall
the definition of a $p$-stable distribution.

\begin{definition}[Zolotarev \cite{Zolotarev86}]
For $0<p<2$, there exists a probability distribution $\mathcal{D}_p$
called the {\em $p$-stable distribution} with $\E[e^{itX}] = e^{-|t|^p}$
for $X\sim \mathcal{D}_p$.
For any integer $n>0$ and vector $a\in\mathbb{R}^n$, if
$X_1,\ldots,X_n \sim \mathcal{D}_p$ are independent, then
$\sum_{i=1}^n a_i X_i \sim ||a||_p\mathcal{D}_p$.
\end{definition}


To prove \Lemma{mainprop}, which is at the heart of the correctness of
our algorithm, we use the following lemma.

\begin{lemma}[Nolan {\cite[Theorem 1.12]{Nolan09}}]
For fixed $0<p<2$, the probability density function of the $p$-stable
distribution is $\Theta(|x|^{-p-1})$.
\end{lemma}

Now we prove our main technical lemma.

\begin{lemma}\LemmaName{mainprop}
Let $n$ be a positive integer and $0<\eps<1$. Let $f(z)$
be a function holomorphic on the complex plane with $|f(z)| =
e^{O(1+|\Im(z)|)}$, where $\Im(z)$ denotes the imaginary part of $z$.
Let
$k=\log(1/\eps)/\log\log\log(1/\eps)$.  Let
$a_1,\ldots,a_n$ be real numbers with $||a||_p = \left(\sum_i
  |a_i|^p\right)^{1/p} = O(1)$.  Let $X_i$ be a $3Ck$-independent
family of $p$-stable random variables for $C$ a suitably large
even constant.  Let $Y_i$ be a fully independent family of $p$-stable
random variables.  Let $X=\sum_i a_i X_i$ and $Y=\sum_i a_i Y_i$.
Then $E[f(X)] = E[f(Y)]+O(\eps)$.
\end{lemma}
\begin{proof}
The basic idea of the proof will be to show that the expectation can
be computed to within $O(\eps)$ just by knowing that the $X_i$'s
are $k$-independent.  Our main idea is to approximate $f$ by a Taylor
series and use the fact that we know the moments of the $X_i$.  The
problem is that the tails of the variables $X_i$ are too wide, and
hence the moments are not defined.  In order to solve this we will
need to truncate some of them in order to get finite moments.

First, we use Cauchy's integral formula to bound the high-order
derivatives of $f$.

\begin{lemma}\LemmaName{derbound}
Let $f^{(\ell)}$ denote the $\ell$th derivative of $f$.  Then,
$|f^{(\ell)}|
= e^{O(\ell)}$ on $\R$.
\end{lemma}
\begin{proof}
For $x\in R$, let $C$ be the circle of radius $\ell$
centered at $x$ in the complex
plane. By Cauchy's integral
formula,
\begin{eqnarray*}
|f^{(\ell)}(x)| &=& \left|\frac{\ell!}{2\pi i}\oint_C
\frac{f(z)}{(z-x)^{\ell + 1}}dz\right|\\
&\le& \frac{\ell!}{2\pi}\int_0^{2\pi}\left|\frac{e^{O(1 + |\ell\cdot
      \sin(t)|)}}{(\ell e^{it})^{\ell +
      1}} \ell e^{it}dt\right|\\
&\le&\frac{\ell!e^{O(\ell)}}{2\pi
  \ell^{\ell}}\int_0^{2\pi}\frac{1}{\left|e^{i\ell t}\right|}dt\\
&\le&\frac{e^{O(\ell)}}{2\pi
  }\int_0^{2\pi}dt\\
&=& e^{O(\ell)}.
\end{eqnarray*}
\end{proof}

Now, define the random variable
$$
B_i = \begin{cases} 0 \ & \textrm{if} \ |a_i X_i| >1 \\ 1 \
  &\textrm{otherwise} \end{cases}
$$
Let
$$
U_i = 1-B_i = \begin{cases} 1 \ & \textrm{if} \ |a_i X_i| >1 \\ 0 \
  &\textrm{otherwise} \end{cases}
$$
and let
$$
X_i' = B_i X_i = \begin{cases} 0 \ & \textrm{if} \ |a_i X_i| >1 \\ X_i
  \ &\textrm{otherwise} \end{cases}
$$
Lastly, define the random variable
$$
D = \sum_i U_i.
$$
We note a couple of properties of these.  In particular
$$
\E[U_i] = O\left(\int_{|a_i|^{-1}}^\infty x^{-1-p}dx \right) =
O\left(|a_i|^p \right).
$$
We would also like to bound the moments of $X_i'$.  In particular we
note that $\E[(a_i X_i')^\ell]$ is 1 for $\ell=0$, by symmetry is 0
when $\ell$ is odd, and otherwise is
\begin{equation}\EquationName{momentcalc}
O\left(\int_0^{|a_i|^{-1}} (a_i x)^\ell x^{-p-1} \right) =
O\left(|a_i|^{\ell} |a_i|^{-\ell+p} \right) = O\left( |a_i|^p \right)
\end{equation}
where the implied constant above can be chosen to hold independently
of $\ell$ (in fact we can pick a better constant if $\ell$ is large).

We will approximate $\E[f(X)]$ as
\begin{equation}\EquationName{approx}
\E\left[\sum_{S,T}\left((-1)^{|T|} \left( \prod_{i\in
        S}U_i\right)\left( \prod_{i\in T} U_i\right)f\left(\sum_{i\in
        S} a_i X_i + \sum_{i\not\in S} a_i X_i' \right)
  \right)\right],
\end{equation}
where the outer sum is over pairs of subsets $S,T\subseteq [n]$, with
$|S|,|T| \leq Ck$, and $S$ and $T$ disjoint.  Call the function inside
the expectation in \Equation{approx}
$F\left(\overrightarrow{X}\right)$.  We would like to bound the error
in approximating $f(X)$ by $F\left(\overrightarrow{X}\right)$.  Fix
values of the $X_i$, and let $O$ be the set of $i$ so that $U_i=1$.  We
note that
$$
F\left(\overrightarrow{X}\right) = \sum_{\substack{S \subseteq O \\
    |S|\leq Ck}} \sum_{\substack{T\subseteq O \backslash S\\ |T|\leq
    Ck}} (-1)^{|T|} f\left(\sum_{i\in S} a_i X_i + \sum_{i\not\in S}
  a_i X_i' \right).
$$
Notice that other than the $(-1)^{|T|}$ term, the expression inside
the sum does not depend on $T$.  This means that if $0<|O\backslash
S|\leq Ck$ then the inner sum is 0, since $O\backslash S$ will have exactly
as many even subsets as odd ones.  Hence if $|O|\leq Ck$, we have that
$$
F\left(\overrightarrow{X}\right) = \sum_{S=O} f\left(\sum_{i\in S} a_i
  X_i + \sum_{i\not\in S} a_i X_i' \right) = f\left(\sum_{i\in O} a_i
  X_i + \sum_{i\not\in O} a_i X_i' \right) = f(X).
$$
Otherwise, after fixing $O$ and $S$, we can sum over possible values
of $t=|T|$ and obtain:
$$
\sum_{\substack{T\subseteq O \backslash S\\ |T|\leq Ck}} (-1)^{|T|} =
\sum_{t=0}^{Ck} (-1)^t \binom{|O\backslash S|}{t}.
$$
In order to bound this we use the following Lemma:
\begin{lemma}\LemmaName{incexcl}
For integers $A\geq B+1>0$ we have that $\sum_{i=0}^B (-1)^i
\binom{A}{i}$ and $\sum_{i=0}^{B+1} (-1)^i \binom{A}{i}$ have
different signs, with the latter sum being 0 if $A=B+1$.
\end{lemma}
\begin{proof}
First suppose that $B < A/2$.  We note that since the terms in each
sum are increasing in $i$, each sum has the same sign as its last
term, proving our result in this case.  For $B\geq A/2$ we note that
$\sum_{i=0}^A (-1)^i \binom{A}{i}=0$, and hence letting $j=A-i$, we
can replace the sums by $(-1)^{A+1}\sum_{j=0}^{A-B-1} (-1)^j
\binom{A}{j}$ and $(-1)^{A+1}\sum_{j=0}^{A-B-2} (-1)^j \binom{A}{j}$,
reducing to the case of $B'=A-B-1 < A/2$.
\end{proof}
Using \Lemma{incexcl}, we note that $\sum_{t=0}^{Ck} (-1)^t
\binom{|O\backslash S|}{t}$ and $\sum_{t=0}^{Ck+1} (-1)^t
\binom{|O\backslash S|}{t}$ have different signs.  Therefore we have
that
$$
\left| \sum_{t=0}^{Ck} (-1)^t \binom{|O\backslash S|}{t} \right| \leq
\binom{|O\backslash S|}{Ck+1} = \binom{D-|S|}{Ck+1}.
$$
Recalling that $|f|$ is bounded, we are now ready to bound $\left|
  F\left(\overrightarrow{X}\right)-f(X)\right|$.  Recall that if
$D\leq Ck$, this is $0$, and otherwise we have that
\begin{align*}
\left|F\left(\overrightarrow{X}\right)-f(X)\right| \leq &
O\left(1+\sum_{\substack{S \subseteq O \\ |S|\leq Ck}}
  \binom{D-|S|}{Ck+1}\right) \\
= & O\left(\sum_{s=0}^{Ck} \binom{D}{s} \binom{D-s}{Ck+1} \right)\\
= & O\left(\sum_{s=0}^{Ck} \binom{D}{Ck+s+1}\binom{Ck+s+1}{s}\right)
\\
\leq & O\left(\sum_{s=0}^{Ck} 2^{Ck+s+1} \binom{D}{Ck+s+1}\right).
\end{align*}
Therefore we can bound the error as
$$
\left|\E\left[F\left(\overrightarrow{X}\right)\right] - \E[f(X)]\right| = O\left( \sum_{s=0}^{Ck} 2^{Ck+s+1} \E\left[\binom{D}{Ck+s+1}\right]\right).
$$
We note that
$$
\binom{D}{Ck+s+1} = \sum_{\substack{I\subseteq [n] \\ |I| = Ck+s+1}}
\prod_{i\in I} U_i.
$$
Hence by linearity of expectation and $2Ck+1$-independence,
\begin{align*}
\E\left[ \binom{D}{Ck+s+1} \right] & = \sum_{\substack{I\subseteq [n]
    \\ |I| = Ck+s+1}} \E\left[ \prod_{i\in I} U_i\right]\\
& = \sum_{\substack{I\subseteq [n] \\ |I| = Ck+s+1}} \prod_{i\in I}
O(|a_i|^p)\\
& = \sum_{\substack{I\subseteq [n] \\ |I| = Ck+s+1}}\left( \prod_{i\in
    I} |a_i|^p\right) e^{O(Ck)}.
\end{align*}
We note that when this sum is multiplied by $(Ck+s+1)!$, these terms
all show up in the expansion of $\left( ||a||_p^p\right)^{Ck+s+1}$.
In fact, more generally for any integer $0\leq t \leq n$
\begin{equation}\EquationName{sympoleq}
\sum_{\substack{I\subseteq [n] \\ |I| = t}} \prod_{i\in I} |a_i|^p
\leq \frac{||a||_p^{tp}}{t!}
\end{equation}
Hence
$$
\E\left[\binom{D}{Ck+s+1}\right] = \frac{e^{O(Ck)}}{(Ck+s+1)!} = e^{O(Ck)} (Ck)^{-Ck-s}.
$$
Therefore we have that
\begin{align*}
\left|\E\left[F\left(\overrightarrow{X}\right)\right] - \E[f(X)]\right| & = O\left( \sum_{s=0}^{Ck} e^{O(Ck)}(Ck)^{-Ck-s}\right)\\
& \leq e^{O(Ck)}(k)^{-Ck}\\
& = \exp\left(-Ck\log k + O(Ck) \right)\\
& = \exp\left(\frac{-C\log(1/\eps) \log\log(1/\eps) }{\log\log\log(1/\eps)} + O(k)\right)\\
& = O(\eps).
\end{align*}
Hence it suffices to approximate
$\E\left[F\left(\overrightarrow{X}\right)\right ]$.

Let
$$
F\left(\overrightarrow{X}\right) = \sum_{\substack{S,T\subseteq [n] \\ |S|,|T| \leq Ck \\ S \cap T = \emptyset}} F_{S,T}\left(\overrightarrow{X}\right),
$$
where
$$
F_{S,T}\left(\overrightarrow{X}\right) = (-1)^{|T|} \left(\prod_{i\in S\cup T} U_i\right) f\left(\sum_{i\in S} a_i X_i + \sum_{i\not\in S} a_i X_i' \right).
$$
We will attempt to compute the conditional expectation of $F_{S,T}\left(\overrightarrow{X}\right)$, conditioned on the values of $X_i$ for $i\in S \cup T$.  It should be noted that the independence on the $X_i$'s is sufficient that the values of the $X_i$ for $i\in S\cup T$ are completely independent of one another, and that even having fixed these values, the other $X_i$ are still $Ck$-independent.

We begin by making some definitions.  Let $R=[n]\backslash (S\cup T)$.  Having fixed $S$, $T$, and the values of $X_i$ for $i\in S \cup T$, we let $c=\sum_{i\in S} a_i X_i$ and let $X'=\sum_{i\in R}a_i X_i'$. We note that unless $U_i=1$ for all $i\in S\cup T$, that $F_{S,T}\left(\overrightarrow{X}\right)=0$, and otherwise that
$$
F_{S,T}\left(\overrightarrow{X}\right) = f(c+X').
$$
This is because if $U_i=1$ for some $i\in T$, then $X_i'=0$.
Let $p_c(x)$ be the Taylor series for $f(c+x)$ about $x=0$ truncated
so that its highest degree term is degree $Ck-1$.  We will attempt to
approximate $\E[f(c+X')]$ by $p_c(X')$.  By Taylor's theorem,
\Lemma{derbound}, and the fact that $C$ is even,
\begin{equation}\EquationName{taylorerror}
|p_c(x) - f(c+x)| \leq \frac{|x|^{Ck}e^{O(Ck)}}{(Ck)!} =
\frac{x^{Ck}e^{O(Ck)}}{(Ck)!}.
\end{equation}
We note that $\E[p_c(X')]$ is determined simply by the independence
properties of the $X_i$ since it is a low-degree polynomial in
functions of the $X_i$.

We now attempt to bound the error in approximating $f(x+c)$ by
$p_c(x)$.  In order to do so we will wish to bound $\E[(X')^{Ck}]$.
Let $\ell=Ck$.  We have that $\E[(X')^\ell] = \E\left[\left(\sum_{i\in
      R} a_i X_i' \right)^\ell \right]$.  Expanding this out and using
linearity of expectation, yields a sum of terms of the form
$\E\left[\prod_{i\in R} (a_i X_i')^{\ell_i}\right],$ for some
non-negative integers $\ell_i$ summing to $\ell$.  Let $L$ be the set
of $i$ so that $\ell_i>0$.  Since $|L|\leq \ell$ which is at most the
degree of independence, \Equation{momentcalc} implies that the
above expectation is $\left(\prod_{i\in L}|a_i|^p\right) e^{O(|L|)}$.
Notice that the sum of the coefficients in front of such terms with a
given $L$ is at most $|L|^\ell$.  This is because for each term in the
product, we need to select an $i\in L$.  \Equation{sympoleq}
implies that summing $\prod_{i\in L}|a_i|^p$ over all subsets $L$ of
size, $s$, gives at most $\frac{||a||_p^{ps}}{s!}$.  Putting everything
together we find that:
\begin{align*}
\E\left[(X')^\ell\right] \leq & \sum_{s=1}^\ell \frac{s^\ell
  e^{O(s)}}{s!} = \sum_{s=1}^\ell \exp\left(\ell\log(s) - s\log s
  +O(s) \right).
\end{align*}
The summand (ignoring the $O(s)$) is maximized when $\frac{\ell}{s} =
\log(s)+1$.  This happens when $s =O\left(\frac{\ell}{\log
    \ell}\right).$  Since the sum is at most $\ell$ times the biggest
term, we get that
$$
\E\left[(X')^\ell\right] \leq \exp\left( \ell \log (\ell) - \ell
  \log\log (\ell) + O(\ell) \right).
$$
Therefore we have that
\begin{align*}
|\E[f(c+X')] - \E[p_c(X')]| & \leq \E\left[\frac{(X')^\ell
    e^{O(\ell)}}{\ell!} \right]\\
& \leq \exp\left(\ell\log(\ell) - \ell\log\log(\ell) - \ell\log(\ell)
  + O(\ell) \right)\\
& = \exp\left( -\ell\log\log(\ell) +O(\ell)\right) \\
& = \exp\left(
  \frac{-C\log(1/\eps)\log\log\log(1/\eps)}{\log\log\log(1/\eps)}
  + o(\log(\eps))\right)\\
& = \exp\left( -(C+o(1))\log(1/\eps)\right) = O(\eps).
\end{align*}

So to summarize:
$$
\E[f(X)] = \E\left[F\left(\overrightarrow{X}\right)\right] +
O(\eps).
$$
Now,
\begin{align*}
\E\left[F\left(\overrightarrow{X}\right)\right] & = \sum_{\substack{S,T\subseteq [n] \\ |S|,|T| \leq Ck \\ S \cap T = \emptyset}} \E\left[F_{S,T}\left(\overrightarrow{X}\right)\right]\\
& = \sum_{\substack{S,T\subseteq [n] \\ |S|,|T| \leq Ck \\ S \cap T = \emptyset}}(-1)^{|T|} \int_{\{x_i\}_{i\in S\cup T}}\left( \prod_{i\in S\cup T} U_i\right) \E[f(c+X')] dX_i(x_i)\\
& = \sum_{\substack{S,T\subseteq [n] \\ |S|,|T| \leq Ck \\ S \cap T = \emptyset}} (-1)^{|T|}\int_{\{x_i\}_{i\in S\cup T}}\left( \prod_{i\in S\cup T} U_i\right)\left( \E[p_c(X')]+O(\eps)\right) dX_i(x_i).
\end{align*}
We recall that the term involving $\E[p_c(X')]$ is entirely determined by the $3Ck$-independence of the $X_i$'s.  We are left with an error of magnitude
\begin{align*}
& O(\eps)\cdot\left(\sum_{\substack{S,T\subseteq [n] \\ |S|,|T| \leq Ck \\ S \cap T = \emptyset}} (-1)^{|T|}\int_{\{x_i\}_{i\in S\cup T}}\left( \prod_{i\in S\cup T} U_i\right) dX_i(x_i)\right)\\
\leq & O(\eps)\cdot\left(\sum_{\substack{S,T\subseteq [n] \\ |S|,|T| \leq Ck \\ S \cap T = \emptyset}} \E\left[\prod_{i\in S\cup T} U_i\right]\right)\\
\leq & O(\eps)\cdot\left(\sum_{\substack{S,T\subseteq [n] \\ |S|,|T| \leq Ck \\ S \cap T = \emptyset}} \left(\prod_{i\in S\cup T} |a_i|^p\right) e^{O(|S|+|T|)}\right).
\end{align*}
Letting $s=|S|+|T|$, we change this into a sum over $s$.  We use
\Equation{sympoleq} to deal with the product. We also note that given
$S\cup T$, there are at most $2^s$ ways to pick $S$ and $T$.  Putting
this together we determine that the above is at most
\begin{align*}
O(\eps)\cdot \left(\sum_{s=0}^{2Ck} 2^s
  \left(\frac{||a||_p^{ps}}{s!}\right)e^{O(s)} \right) 
= & O(\eps) \cdot \left( \sum_{s=0}^{2Ck} \frac{O(1)^s}{s!}
\right) \\
= & O(\eps).
\end{align*}

Hence the value of $\E[f(X)]$ is determined up to $O(\eps)$.

\end{proof}

The following is a corollary of \Lemma{mainprop} which is more
readily applicable.

\begin{corollary}\CorollaryName{mcor}
Let $n$ be a positive integer and $0<\eps<1$.
Let $f(z)$ be a holomorphic function on $\C$ so that $|f(z)| =
e^{O(1+|\Im(z)|)}$.
Let $k=c\log(1/\eps)/\log\log\log(1/\eps)$ for a sufficiently large
constant $c>0$.
  Let $a_i$ be
real numbers for $1\leq i \leq n$.  Let $C>0$ be a real number so that
$||a||_p = O(C)$.  Let $X_i$ be a $k$-wise independent family of $p$-stable
random variables, $Z$ be a single
$p$-stable random variable, and
$X=\sum_i a_i X_i$. Then, $\E[f(X/C)] = \E[f(||a||_p Z/C)] +
O(\eps)$.
\end{corollary}
\begin{proof}
Apply \Lemma{mainprop} with the vector whose entries are
$a_i/C$ so that $||a||_p = O(1)$ and $Y=\sum_i (a_i/C)Y_i$ has the
same distribution as $||a||_p Z/C$.
\end{proof}

We now show the implications of \Corollary{mcor}.

\begin{lemma}\LemmaName{median-lp}
Let $a_i$ be real numbers for $1\leq i \leq n$.  Let $k$ and $r$ be a
suitably large constants, and let $X_{i,j}$ a $2$-wise independent family of
$k$-wise independent $p$-stable random variables ($1\leq i \leq
n$, $1\leq j \leq r$).  Then the median value across all $j$ of
$|\sum_i a_i X_{i,j}|$ is within a constant multiple of $||a||_p$ with
probability tending to 1 as $k$ and $r$ tend to infinity, independent
of $n$.
\end{lemma}
\begin{proof}
We apply \Corollary{mcor} to a suitable function $f$ which
\begin{enumerate}
\item is strictly positive for all $x\in\R$,
\item  is an even function, and
\item  decreases strictly monotonically to $0$ as $x$ tends away from
  $0$.
\end{enumerate}
We note 
$$f(x) = -\int_{-\infty}^x \frac{\sin^4(y)}{y^3} dy$$
satisfies these properties. A thorough explanation of why $f$
satisfies the desired properties is in \Section{magic-function}.
Henceforth, for
$0<z\le f(0)$,
$f^{-1}(z)$ denotes the (unique) nonnegative inverse of $z$.
We
consider
for constants $C=\Theta(1)$ the random variable
$$ A_C = \frac 1r\left(\sum_j f\left(\left(\sum_i a_i
    X_{i,j}\right)\cdot\left(\frac{C}{||a||_p}\right)\right)\right) .$$
By \Corollary{mcor}, if $Z$ is a $p$-stable variable, then
$\E[A_C] = \E[f(CZ)] + O(1)$, where the $O(1)$ term can be made 
arbitrarily small
by choosing $k$ sufficiently large.  Furthermore since $f$ is
bounded and the terms in the sum over $j$ defining $A_C$ are
$2$-wise independent, $\Var(A_C) = O(1/r)$.  Thus by
Chebyshev's inequality, for $k,r$ sufficiently
large, $A_C$ is within any desired constant of
$\E[f(CZ)]$ with probability arbitrarily close to 1.

We apply the above for a $C>0$ large enough that $\E[f(CZ)] <
f(0)/3$, and $C'>0$ small enough that $\E[f(C'Z)] > 2f(0)/3$. 
By picking $k,r$ sufficiently large,
then with any desired constant probability we can ensure
$A_C<x<f(0)/2<y<A_{C'}$, for some
constants $x>f(0)/3$ and $y<2f(0)/3$ of our choosing --- to be
concrete, pick $x = 4f(0)/9$ and $y=5f(0)/9$. In order for
this to hold it
must be the case
that for at least half of the $j$'s that
$$f\left( \left(\sum_i a_i
    X_{i,j}\right)\cdot\left(\frac{C}{||a||_p}\right)\right) < 2x <
f(0) .$$
This bounds the median of $\left|\sum_i a_i X_i \right|$ from below by
$$\left(\frac{f^{-1}\left(8f(0)/9\right)}{C}\right)||a||_p >
\left(\frac{2}{5C}\right)||a||_p.$$
Similarly, it must also be the case that for at
least half of the $j$'s 
$$f\left(\left(\sum_i a_i
    X_{i,j}\right)\cdot\left(\frac{C'}{||a||_p}\right)\right) >
2(y-f(0)/2) > 0 .$$
This bounds the median of $\left|\sum_i a_i X_i \right|$ from above by
$$\left(\frac{f^{-1}\left(f(0)/9\right)}{C'}\right)||a||_p <
\left(\frac 2{C'}\right)||a||_p.$$
The bounds on $f^{-1}(8f(0)/9)$ and $f^{-1}(f(0)/9)$ were verified
by computer. Comments on computing $C,C'$ are
in \Section{magic-function}.
\end{proof}

\begin{lemma}\LemmaName{cosine-works}
Given $\eps>0$, $k$ as in \Corollary{mcor}, $r$ a suitably
large multiple of $\eps^{-2}$, $C=\Theta(||a||_p)$, and $X_{i,j}$
($1\leq i \leq n, 1\leq j \leq r$) a 2-independent family of
$k$-independent families of $p$-stable random variables then with
probability that can be made arbitrarily close to 1 (by increasing
$r$), it holds that 
$$\left| \left(\frac 1r\sum_{j=1}^r \cos\left(\frac{\sum_{i=1}^n a_i
        X_{i,j}}{C}\right)\right)
  - e^{-\left(\frac{||a||_p}{C}\right)^p}
\right| < \eps .$$
\end{lemma}
\begin{proof}
The Fourier transform of the
probability density function $q(x)$ of $\mathcal{D}_p$ is
$\hat{q}(\xi) = e^{-|\xi|^p}.$
Letting $B=\frac{||a||_p}{C}$, the expectation of $\cos(BZ)$ for
$Z\sim\mathcal{D}$ is
$$
\int_{-\infty}^{\infty} q(x)\frac{e^{iBx}+e^{-iBx}}{2}dx =
\frac{\hat{q}(B)+\hat{q}(-B)}{2} = e^{-|B|^p}.
$$
By \Corollary{mcor}, if $k$ is sufficiently
large, the expected value of
$(\sum_j \cos((\sum_i a_i X_i)/C))/r$
is within $\eps/2$ of $e^{-(||a||_p/C)^p}$.
Noting that each term in the sum is bounded by 1, and that they
form a 2-independent family of random variables, we have that the
variance of our estimator is upper bounded by $1/r$.
Hence by Chebyshev's inequality, if $r$ is chosen to by a suitably
large multiple of $\eps^{-2}$, then with the desired probability
our estimator is within $\eps/2$ of its expected value.
\end{proof}

Now we prove our main theorem.

\begin{proofof}{\Theorem{main-optimallp}}
In \Figure{lpalg}, as long as $k,r$ are chosen to be larger
than some constant, $A$ is a constant factor approximation to
$||a||_p$ by \Lemma{median-lp} with probability at least $7/8$.
Conditioned on this,
consider $C=(\sum_j \cos(A_j/A))/r$. By
\Lemma{cosine-works}, with probability at least $7/8$, $C$
is within $O(\eps)$ of
$e^{-(||a||_p/A)^p}$ from which a
$(1+O(\eps))$-approximation of
$||a||_p$ can be computed as $A\cdot (-\ln(C))^{1/p}$.
Note that our approximation is in fact a $(1+O(\eps))$-approximation
since  the function $f(x) = e^{-|x|^p}$
is bounded both from above and below by constants for $x$ in a
constant-sized interval (in our case, $x$ is $||a||_p/A$), and thus an
additive $O(\eps)$-approximation to $e^{-|x|^p}$ is also a
multiplicative $(1+O(\eps))$-approximation.

There are though still two basic problems with the algorithm of
\Figure{lpalg}.  The
first is that we cannot store the values of $X_j$ to unlimited
precision, and will at some point have rounding errors.  The second
problem is that we can only produce families of random variables with
finite entropy and hence cannot keep track of a family of continuous
random variables.

We deal with the precision problem first.  We will pick some number
$\delta=\Theta(\eps m^{-1})$.  We round each $X_{i,j}$ to the
nearest multiple of $\delta$.  This means that we only need to store
the $X_j$ to a precision of $\delta$.  This does produce an error in
the value of $X_j$ of size at most $||a||_1\delta \leq |i:a_i\neq 0|
\max(|a_i|)\delta \leq m||a||_p \delta = \Theta(\eps ||a||_p)$.  This
means that $C$ is going to be off by a factor of at most
$O(\eps)$, and hence still probably within a constant multiple of
$||a||_p$.  Hence the values of $X_j/C$ will be off by $O(\eps)$, so
the values of $A$ and our approximation for $||a||_p$ will be off by an
additional factor of $O(\eps)$.

Next we need to determine how to compute these continuous
distributions.  It was shown by \cite{CMS76} that a $p$-stable
random variable can be generated by taking $\theta$ uniform in
$[-\pi/2,\pi/2]$, $r$ uniform in $[0,1]$ and letting
$$
X = f(r,\theta) = \frac{\sin(p\theta)}{\cos^{1/p}(\theta)}\cdot
\left(\frac{\cos(\theta(1-p))}{\log(1/r)}\right)^{(1-p)/p}.
$$
We would like to know how much of an error is introduced by using
values of $r$ and $\theta$ only accurate to within $\delta'$.  This
error is at most $\delta'$ times the derivative of $f$.  This
derivative is not large except when $\theta$ or $(1-p)\theta$ is close
to $\pm \pi/2$, or when $r$ is close to $0$ or $1$.  Since we only
ever need $mr$ different values of $X_{i,j}$, we can assume that with
reasonable probability we never get an $r$ or $\theta$ closer to
these values than $O(m^{-1}\eps^2)$.  In such a case the
derivative will be bounded by $(m\eps^{-1})^{O(1)}$.  Therefore,
if we choose $r$ and $\theta$ with a precision of
$(m^{-1}\eps)^{O(1)}$, we can get the value of $X$ with
introducing an error of only $\delta$.

Lastly, we need to consider memory requirements.  Our family must be a
$2$-independent family containing $O(\eps^{-2})$ $k$-independent
families of $U$ random variables.  Each random variable requires
$O(\log(m\eps^{-1}))$ bits.  The amount of space needed to pick
out an element of this family is only $O(k(\log(U)+\log(m\eps^{-1}))) =
O(k\log(m/\eps)) = O(k\log m)$ (recall we can assume $\log(U) = O(\log
N)$, and $\eps\ge 1/\sqrt{m}$). More important is the
information needed to store the $X_j$.  We need to store them to a
precision of $\delta$.  Since there are only $mr$ values of $X_{i,j}$,
with reasonable probability, none of them is bigger than a polynomial
in $mr$.  If this is the case, the maximum value of any $X_j$ is at
most $(mM\eps^{-1})^{O(1)}$.  Hence each $X_j$ can be stored in
$O(\log(mM\eps^{-1})) = O(\log(mM))$ space, thus making the total
space requirements $O(\eps^{-2}\log(mM))$.
\end{proofof}

\subsection{Derandomizing $L_p$ Estimation via Armoni's
  PRG}\SectionName{prg}

Indyk \cite{Indyk06}, and later Li \cite{Li08b}, gave algorithms for
$L_p$ estimation which are also based on $p$-stable distributions.
Their algorithms differ from ours in \Figure{lpalg} in two ways.
First, both Indyk and Li made the variables $X_{i,j}$ in Step 1 truly
random as opposed to having limited independence.  Second, the
estimator they use in Step 2 differs.  Indyk uses a median estimator
on the $|A_j|$, and Li has two estimators: one based on the geometric
mean, and one on the harmonic mean.
The change in Step 1 at
first seems to make the algorithms
of Indyk and Li not implementable in small space, since there are
$n/\eps^2$ random variables $X_{i,j}$ to be stored.  Indyk though
observed that his algorithm could be derandomized by using a PRG against
small-space computation, and invoked Nisan's PRG to derandomize his
algorithm.  Doing so multiplied his space complexity by a
$\log(N/\eps)$ factor.  Li then similarly used Nisan's PRG to
derandomize his algorithm.

Nisan's PRG
\cite{Nisan92}
stretches a seed of $O(S\log R)$ random bits to $R$ ``pseudorandom'' bits
fooling any space-$S$ algorithm with one-way access to its
randomness.  We show
that a PRG construction of Armoni \cite{Armoni98} can be combined
with a more space-efficient implementation of a recent extractor of
Guruswami, Umans, and Vadhan (GUV) \cite{GUV07} to produce a PRG whose
seed length is only $O((S/(\log S - \log\log R + O(1)))\log R)$
for any $R = 2^{O(S)}$.  Due to the weaknesses of
extractor constructions at the time, Armoni's original PRG
only worked when $R < 2^{S^{1-\delta}}$ for constant $\delta > 0$.
In the cases of Indyk and Li, $S =
O(\eps^{-2}\log(mM))$ and $R = \poly(N)/\eps^2$.  The key here is that
although $N$ can be exponentially large in $\log(mM)$, the dependence
on $\eps$ in both $S$ and $R$ are polynomially related.  The result is
that using the improved Armoni PRG provides a more efficient
derandomization than Nisan's PRG
by a $\log(1/\eps)$ factor, giving the following.

\begin{theorem}
The $L_p$-estimation algorithms of Indyk and Li can be
implemented in space
$O(\eps^{-2}\log(mM)\log(N/\eps)/\log(1/\eps))$.
\end{theorem}



Most of our changes to the implementation of the GUV
extractor are parameter changes which guarantee that we always work
over a field for which a highly explicit family of
irreducible polynomials is known. For example, we change the
parameters of an expander construction of
  GUV based on Parvaresh-Vardy codes \cite{PV05} which feeds into
  their extractor construction. Doing so allows us to replace calls to
  Shoup's algorithm for finding irreducibles over
  $\mathbb{F}_2[x]$, which uses superlinear space, with using two
  explicit families of
  irreducibles over $\mathbb{F}_7[x]$ with a few properties. One property
  we need is that if we define extension fields using polynomials
  from one family, then the polynomials from the other 
  family remain irreducible over these extension fields.
Full details are in \Section{appendix-armoni}.

\section{Lower Bounds}\SectionName{lower-bounds}
In this section we prove our lower bounds for
$(1\pm\eps)$-multiplicative approximation of $F_p$ for any real
constant $p\ge 0$ when deletions are allowed.  When $p\ge 0$, we prove
a $\Omega(\eps^{-2}\log(\eps^2\minnm))$ lower
bound.
When $p$ is a constant strictly greater than $0$, the lower bound
improves to
$\Omega(\min\{N, \eps^{-2}(\log(\eps^2m M))\})$.  All our lower
bounds assume $\eps \ge 1/\sqrt{\minnm}$.  We also point out that
$\Omega(\log\log(nmM))$ is a folklore lower bound for all
problems we
consider in the strict turnstile model by a direct reduction from
\textsc{Equality}.  In the update-only model, there is a folklore
$\Omega(\log\log n)$ lower bound.  Both lower bounds assume $m\ge 2$.
Our lower bounds hold for all ranges of the
parameters $\eps,n,m,M$ varying independently. 

Our proof in part uses the fact that
\textsc{Augmented-Indexing} requires a linear amount of communication
in the one-way, one-round model \cite{BJKK04,MNSW98}.  We also use a known
reduction \cite{jks07,WoodruffThesis} from indexing to
\textsc{Gap-Hamdist}.  Henceforth all
communication games discussed will be one-round and two-player, with
the first
player to speak named ``Alice'', and the second ``Bob''.  We assume
that Alice and Bob have access to public randomness.

\begin{definition}
In the \textsc{Augmented-Indexing} problem, Alice receives a vector
$x\in\{0,1\}^n$, Bob receives some $i\in [n]$ as well as all $x_j$
for $j > i$, and Bob must output $x_i$.  The problem
\textsc{Indexing} is defined similarly, except Bob receives only $i\in
[n]$, without receiving $x_j$ for $j > i$.
\end{definition}

\begin{definition}
In the \textsc{Gap-Hamdist} problem, Alice receives $x\in\{0,1\}^n$
and Bob receives $y\in\{0,1\}^n$.  Bob is promised that either
$\Delta(x,y) \le n/2 - \sqrt{n}$ (\textsc{NO} instance), or
$\Delta(x,y) \ge n/2 + \sqrt{n}$ (\textsc{YES} instance) and must
decide which case holds.  Here $\Delta(\cdot,\cdot)$ denotes the
Hamming distance.
\end{definition}

The following two theorems are due to \cite{BJKK04,MNSW98} and
\cite{jks07,WoodruffThesis}.

\begin{theorem}[Miltersen {et al.} \cite{MNSW98}, Bar-Yossef {et al.}
  \cite{BJKK04}]\TheoremName{mnsw}
The randomized one-round, one-way communication complexity of
solving \textsc{Augmented-Indexing} with probability at least $2/3$ is
$\Omega(n)$.  Furthermore, this lower
bound holds even if Alice's and Bob's inputs are each chosen
independently, uniformly at random.  The lower bound also still holds
if Bob only receives a subset of the $x_j$ for $j > i$.\qedbox
\end{theorem}

\begin{theorem}[Jayram {et al.} \cite{jks07}, Woodruff
  {\cite[Section 4.3]{WoodruffThesis}}]\TheoremName{gaphamdist}
There is a reduction from \textsc{Indexing} to \textsc{Gap-Hamdist}
such that the uniform (i.e. hard) distribution over \textsc{Indexing}
instances is
mapped to a distribution over \textsc{Gap-Hamdist}
instances where each of Alice and Bob receive strings whose marginal
distribution is uniform, and deciding \textsc{Gap-Hamdist} over this
distribution with
probability at least $11/12$ implies a solution to \textsc{Indexing}
with probability at least $2/3$.
Also, in this reduction the vector length $n$ in \textsc{Indexing} is
the same as the
vector length in the reduced \textsc{Gap-Hamdist}
instance to within a constant factor.\qedbox
\end{theorem}


We now give our lower bounds.  We use the following observation
in the proof of \Theorem{main-lowerbound}.

\begin{observation}\ObservationName{gaptol0}
For two binary vectors $u,v$ of equal length, let $\Delta(u,v)$ denote
their Hamming distance.  Then for any $p \ge 0$, $(2^p-2)\Delta(u,v) =
2^p||u||_1+2^p||v||_1 - 2||u+v||_p^p$.
\end{observation}

\begin{theorem}\TheoremName{main-lowerbound}
For any real constant $p\ge 0$, any one-pass
streaming algorithm for
$(1\pm\eps)$-multiplicative approximation of $F_p$ with probability at
least $11/12$ in the strict turnstile model requires
$\Omega(|p-1|^2\eps^{-2}\log (\eps^2 \minnm / |p-1|^2))$ bits of space.
\end{theorem}
\begin{proof}
Given an algorithm $A$ providing a $(1\pm d|p-1|\eps)$-multiplicative
approximation of $F_p$ with probability at least $11/12$, where $d>0$
is some constant to be fixed later, we devise a protocol to decide
\textsc{Augmented-Indexing} on strings of length $\eps^{-2}\log(\eps^2
N)$.

Let Alice receive
$x\in\{0,1\}^{\eps^{-2}(\log(\eps^2 N))}$, and Bob
receive $z\in [\eps^{-2}(\log(\eps^2 N))]$.
Alice divides $x$ into
$\log(\eps^2 N)$ contiguous blocks where the $i$th block
$b_i$ is of size $1/\eps^2$.  Bob's index $z$ lies in some
$b_{i(z)}$, and Bob receives bits $x_j$ that lie in a block
$b_i$ with $i > i(z)$. Alice applies the
\textsc{Gap-Hamdist}
reduction of \Theorem{gaphamdist} to each $b_i$ separately to obtain
new vectors $y_i$ each of length at most $c/\eps^2$ for some constant
$c$ for all $0 \le i < \log(\eps^2 N)$.
Alice then creates a stream
from the set of $y_i$ by, for each $i$ and each bit $(y_i)_j$ of
$y_i$, imagining universe elements $(i,j,1),\ldots,(i,j,2^i)$
and inserting them all into
the stream if $(y_i)_j = 1$, and not inserting them otherwise.  Alice
processes this stream with $A$ then sends the state of $A$ to Bob
along with the Hamming
weight $w(y_i)$ of $y_i$ for all $i$. Note the size of the
universe in the stream is at most $c\eps^{-2}\sum_{i=0}^{\log(\eps^2
  N) - 1}2^i = O(N) = O(n)$.

Now, since Bob knows the bits in $b_i$ for $i > i(z)$ and shares
randomness with
Alice, he can run the same \textsc{Gap-Hamdist} reduction as Alice to
obtain the $y_i$ for $i > i(z)$ then delete all the insertions Alice
made for these $y_i$.
Bob then
performs his part of the reduction from \textsc{Indexing} on strings of length
$1/\eps^2$ to \textsc{Gap-Hamdist} within the block $b_{i(z)}$ to obtain a
vector $y(B)$ such that deciding whether
$\Delta(y(B),y_{i(z)}) > \eps^{-2}/2 + \eps^{-1}$ or
$\Delta(y(B),y_{i(z)}) < \eps^{-2}/2 + \eps^{-1}$ with probability at
least $11/12$ allows one to decide the \textsc{Indexing} instance with
probability at least $2/3$.  Here
$\Delta(\cdot, \cdot)$ denotes Hamming distance.
For each $j$ such that
$y(B)_j = 1$, Bob inserts universe elements
$(i(z),j,1),\ldots,(i(z),j,2^{i(z)})$ into the stream being processed
by $A$.  We have so far described all stream updates, and thus the
number of updates is at most $2c\eps^{-2}\sum_{i=0}^{\log(\eps^2
  N) - 1}2^i = O(N) = O(m)$.
By \Observation{gaptol0} with $u=y_{i(z)}$ and $v=y(B)$, the $p$th moment $L''$
of the stream now exactly satisfies
$L'' = 2^{i(z)}((1-2^{p-1})\Delta(y(B),y_{i(z)}) + 2^{p-1}w(y_{i(z)})
+ 2^{p-1}w(y(B))) + \sum_{i<i(z)}w(y_i)2^i$.  Setting $\eta =
\sum_{i<i(z)} w(y_i)2^i$ and rearranging terms,
\vspace{-.1in}
\[ \Delta(y(B),y_{i(z)}) = \frac{2^{p-1}}{2^{p-1}-1}w(y_{i(z)}) +
\frac{2^{p-1}}{2^{p-1}-1}w(y(B)) + \frac{2^{-i(z)}(\eta -
  L'')}{2^{p-1}-1}\vspace{-.1in}\]
Recall that in this \textsc{Gap-Hamdist} instance, Bob must decide
whether $\Delta(y(B),y_{i(z)}) < 1/2\eps^2 - 1/\eps$ or
$\Delta(y(B),y_{i(z)}) > 1/2\eps^2 + 1/\eps$.  Bob knows
$\eta$, $w(y_{i(z)})$, and $w(y(B))$ exactly.
To decide \textsc{Gap-Hamdist} it thus suffices to obtain
a $((2^{p-1}-1)/(4\eps))$-additive approximation to
$2^{-i(z)}L''$. Since $2^{-i(z)}L''$ is upper-bounded in
absolute value by $(1+2^p)/\eps^2$, our desired
additive approximation is guaranteed by obtaining a $(1\pm
((2^{p-1}-1)\eps/(4\cdot(1+2^p))))$-multiplicative approximation
to $L''$.
Since $p\neq 1$ is a constant and $|2^x - 1| = \Theta(|x|)$
as $x\rightarrow 0$, this is a $(1\pm
O(|p-1|\eps))$-multiplicative approximation, which we can obtain from
$A$ by setting $d$ to be a sufficiently large constant.
Recalling that $A$ provides this $(1\pm O(|p-1|\eps))$-approximation with
probability at least $11/12$, we solve
\textsc{Gap-Hamdist} in the block $i(z)$ with probability at least
$11/12$, and thus \textsc{Indexing} in block $i(z)$ with
probability
at least $2/3$ by \Theorem{gaphamdist}.  Note this is equivalent to
solving the original \textsc{Augmented-Indexing} instance.

The only
bits communicated other than the state of $A$ are the
transmissions of $w(y_i)$ for $0 \le i\le \log(\eps^2 N)$.  Since
$w(y_i) \le 1/\eps^2$, all Hamming weights can be communicated in
$O(\log(1/\eps)\log(\eps^2 N)) = o(\eps^{-2}\log(\eps^2 N))$ bits.  By
the lower bound on
\textsc{Augmented-Indexing} from \Theorem{mnsw}, we thus have that
$(1\pm d|p-1|\eps)$-approximation requires
$\Omega(\eps^{-2}\log(\eps^2 N))$ bits of space for some constant $d >
0$.  In other words,
setting $\eps' = d'|p-1|\eps$ we have that a
$(1\pm\eps')$-approximation requires
$\Omega(|p-1|^2\eps'^{-2}\log(\eps'^2 N/|p-1|^2)$ bits of space.
\end{proof}

When $p$ is strictly positive, we can improve our lower bound
by gaining a dependence on $mM$ rather than $N$, obtaining the
following lower bound.

\begin{theorem}\TheoremName{strictlowerbound-M}
For any real constant $p > 0$, any one-pass streaming algorithm for
$(1\pm\eps)$-multiplicative approximation of $F_p$ with probability at
least $11/12$ in the strict turnstile model requires
$\Omega(\min\{N, |p-1|^2\eps^{-2}(\log (\eps^2 mM / |p-1|^2))\})$
bits of space.
\end{theorem}
\begin{proof}
In the proof of \Theorem{main-lowerbound}, Alice divided her input $x$
into $\log(\eps^2 N)$ blocks each of equal size and used the $i$th
block to create an instance of \textsc{Gap-Hamdist}.  However, in
order to have the weight of each block's contribution to the stream
increase geometrically, Alice had to replicate each coordinate in the
$i$th block $2^i$ times.  Now, instead, round $M$ to the nearest power
of $2^{1/p}$ and let Alice's input be a string $x$ of length
$\eps^{-2}\min\{\log_{2^{1/p}} M,\eps^2 N\}$.  Dividing her input into
$\min\{\log_{2^{1/p}} M, \eps^2 N\}$ blocks, Alice does not replicate
any coordinate in a block $i$ but rather gives each coordinate
frequency $2^{i/p}$.  By choice of the number of blocks, no item's
frequency will be larger than $M$, and the number of universe elements
and the stream length will each be at most $N$.
These frequencies $f_1,f_2,\ldots$ are chosen
so that $f_i^p = 2^i$.  Similarly to \Observation{gaptol0}, for two
vectors $u,v$ of equal length where each coordinate is either $t$ or
$0$ (in \Observation{gaptol0} the vectors were binary),
for any $p \ge 0$ we have $t^p(2^p-2)\Delta(u,v) =
t^p2^p||u||_1+t^p2^p||v||_1 - 2||u+v||_p^p$ where $\Delta(u,v)$ is the
Hamming distance of $u,v$.  

Following the same steps as in \Theorem{main-lowerbound} with the same
notation, one arrives at
\[ \Delta(y(B),y_{i(z)}) = \frac{2^{p-1}}{2^{p-1}-1}w(y_{i(z)}) +
\frac{2^{p-1}}{2^{p-1}-1}w(y(B)) + \frac{(\eta -
  L'')}{2^{i(z)}(2^{p-1}-1)}\]
since $f_{i(z)}^p = 2^{i(z)}$.
For deciding \textsc{Gap-Hamdist} in block $i(z)$ it suffices to
obtain an additive $2^{i(z)}(2^{p-1}-1)/(4\eps)$-additive approximation
to $L''$.  Since $L'' \le 2^{i(z)}(2^p+1)/\eps^2$, the desired
additive approximation can be obtained by a
$(1\pm((2^{p-1}-1)\eps/(4\cdot(2^p+1)))$-multiplicative approximation,
just as in \Theorem{main-lowerbound}.  The rest of the proof
is identical as in \Theorem{main-lowerbound}.

The above argument yields the lower bound
$\Omega(\min\{N,\eps^2\log(M))$.  We can similarly obtain the lower
bound $\Omega(\min\{N,\eps^2\log(\eps^2 m))$ by, rather than updating an
item in the stream by $f_i = 2^{i/p}$ in one update, we
update the same item $f_i$ times by $1$.  The number of total updates
in the $i$th block is then $2^{i/p}/\eps^2$, and thus the maximum
number of blocks we can give Alice to ensure that both the stream
length and number of used universe elements is at most $N$ is
$\min\{\eps^2 N, O(\log(\eps^2 m))\}$.
\end{proof}

The decay of our lower bounds as $p\rightarrow 1$ is necessary in the
strict turnstile model since Li gave an algorithm in this model whose
dependence on $\eps$ becomes subquadratic as $p\rightarrow 1$
\cite{Li08a}.  Furthermore, when $p=1$ there is a $O(\log(mM))$-space
deterministic algorithm for computing $F_1$: maintain a
counter. In
the turnstile model, for $p>0$ we give a lower bound
matching \Theorem{strictlowerbound-M} but without any decay as
$p\rightarrow 1$.

\begin{theorem}\TheoremName{lowerbound-M}
For any real constant $p> 0$, any one-pass streaming algorithm for
$(1\pm\eps)$-multiplicative approximation of $F_p$ in the turnstile
model with probability at least $11/12$ requires
$\Omega(\min\{N, \eps^{-2}(\log(\eps^2 mM))\})$
bits of space.
\end{theorem}
\begin{proof}
 As in
\Theorem{strictlowerbound-M},
Alice receives an input string $x$ of length $\eps^{-2}\min\{\log M,
\eps^2 N\}$ as opposed to the string of length $\eps^{-2}\log(\eps^2N)$
in \Theorem{main-lowerbound}.  Also, Alice carries out her part of
the protocol just as in \Theorem{strictlowerbound-M}.
However, for each $j$ such that
$y(B)_j=1$, rather than inserting a universe element with frequency
$2^{i(z)/p}$, Bob {\em deletes} it with that frequency.
Now we have
$L''$, the $p$th moment of the stream, exactly equals
$2^{i(z)}\Delta(y(B),y_{i(z)}) + \sum_{i<i(z)}w(y_i)2^{i(z)}$, and
thus $\Delta(y(B), y_{i(z)}) = 2^{-i(z)}(\eta - L'')$.
As in \Theorem{main-lowerbound}, Bob knows $\eta$ exactly and thus
only needs a $(1/4\eps)$-additive approximation to $L''2^{-i(z)}$ to
decide the \textsc{Gap-Hamdist} instance (and thus the original
\textsc{Augmented-Indexing} instance),
which he can obtain via a $(1\pm (\eps/8))$-approximation to $L''$
since $L''2^{-i(z)} \le 2/\eps^2$.
\end{proof}

Our technique also improves the known lower
bound for additively estimating the entropy of a stream in the
strict turnstile model. The
proof combines ideas of \cite{CGM07} with our technique of
embedding geometrically-growing hard instances. By entropy of the
stream, we mean the
empirical probability distribution on $[n]$ obtained by setting $p_i =
a_i/||a||_1$.

\begin{theorem}\TheoremName{entropy-lb}
Any algorithm for $\eps$-additive approximation of $H$, the
entropy of a stream, in the strict turnstile model with probability at
least $11/12$ requires space $\Omega(\eps^{-2}\log
(\minnm)/\log(1/\eps))$.
\end{theorem}
\begin{proof}
We reduce from \textsc{Augmented-Indexing}, as in
\Theorem{main-lowerbound}. Alice receives a string of
length $s = \log N / (2\eps^2\log(1/\eps))$, and Bob receives an index
$z\in[s]$.  Alice conceptually divides her input into $b = \eps^2 s$
blocks, each of size $1/\eps^2$, and reduces each
block using the \textsc{Indexing}$\rightarrow$\textsc{Gap-Hamdist}
reduction of \Theorem{gaphamdist} to obtain $b$ \textsc{Gap-Hamdist}
instances with strings
$y_1,\ldots,y_b$, each of length $\ell = \Theta(1/\eps^2)$.  For each
$1\le i\le b$, and $1\le j\le \ell$ Alice inserts universe elements
$(i,j,1,(y_i)_j),\ldots,(i,j,\eps^{-2i},(y_i)_j)$ into the stream and
sends the state of a streaming algorithm to Bob.

Bob identifies the block $i(z)$ in which $z$ lands and deletes all
stream elements associated with blocks with index $i > i(z)$.  He then
does his part in the
\textsc{Indexing}$\rightarrow$\textsc{Gap-Hamdist} reduction to obtain a
vector $y(\hbox{Bob})$ of length $\ell$.  For all $1\le j\le \ell$, he inserts
the universe elements $(i(z),j,1,y(\hbox{Bob})_j), \ldots,
(i(z),j,\eps^{-2i(z)},y(\hbox{Bob})_j)$ into the stream.

The number of stream tokens from block indices $i < i(z)$ is $A =
\eps^{-2}\sum_{i=0}^{i(z)-1}\eps^{-2i} = \Theta(\eps^{-2i(z)})$.  The
number of tokens
in block $i(z)$ from Alice and Bob combined is $2\eps^{-(2i(z)+2)}$.
Define $B = \eps^{-2i(z)}$ and $C = \eps^{-2}$.  The
$L_1$ weight of the stream is $R = A+2BC$.  Let $\Delta$
denote the Hamming distance between $y_{i(z)}$ and $y(\hbox{Bob})$ and $H$
denote the entropy of the stream.

We have:
\begin{eqnarray*}
H &=& \frac{A}{R}\log(R) + \frac{2B(C-\Delta)}{R}\log\left(\frac{R}{2}\right) +
\frac{2B\Delta}{R}\log(R)\\
&=& \frac{A}{R}\log(R) + \frac{2BC}{R}\log(R) - \frac{2BC}{R} +
\frac{2B\Delta}{R}
\end{eqnarray*}

Rearranging terms gives

\begin{equation}\EquationName{ham-relation}
\Delta = \frac{HR}{2B} + C - C\log(R) - \frac{A}{2B}\log(R)
\end{equation}

To decide the \textsc{Gap-Hamdist} instance, we must decide whether
$\Delta < 1/2\eps^2 - 1/\eps$ or
$\Delta > 1/2\eps^2 + 1/\eps$.  By \Equation{ham-relation} and the
fact that Bob knows $A$, $B$, $C$, and $R$,
it suffices to obtain a $1/\eps$-additive approximation to $HR/(2B)$
to accomplish this goal.  In other
words, we need a $2B/(\eps R)$-additive approximation to $H$.  Since
$B/R = \Theta(\eps^2)$, it suffices to obtain an additive
$\Theta(\eps)$-approximation to $H$.  Let $\mathcal{A}$ be a streaming
algorithm which can provide an additive $\Theta(\eps)$-approximation
with probability at least $11/12$.  Recalling that correctly
deciding the \textsc{Gap-Hamdist} instance with probability $11/12$
allows one to correctly decide the original
\textsc{Augmented-Indexing} instance with probability $2/3$ by
\Theorem{gaphamdist}, and given \Theorem{mnsw}, $\mathcal{A}$ must use
at least $\log(N)/(\eps^2\log(1/\eps))$ bits of space.  As required,
the length of
the vector being updated in the stream is at most $\sum_{i=1}^s
\eps^{-2i} = O(N) = O(n)$, and the length of the stream is exactly
twice the vector length, and thus $O(N) = O(m)$.
\end{proof}

\section{$L_0$ in turnstile streams}\SectionName{l0-est}

We describe our algorithm for multiplicatively approximating
$L_0$ in the turnstile model using $O(\eps^{-2}\log(\eps^2
N)(\log(1/\eps) + \log\log(mM)))$ space with $O(1)$ update and
reporting time. Without loss of generality,
we assume (1) $N$ is a power of $2$, and (2) $\eps \ge
1/(3\cdot\minnm)$.  We can assume (2) 
since otherwise
one could compute $L_0$ exactly since $L_0 \le \minnm$ is an
integer. In both this algorithm and our $F_0$ algorithm, we make use
of a few lemmas analyzing a balls-and-bins random process where $A$
good balls and $B$ bad balls are thrown into $K$ bins with limited
independence (in the case of our $L_0$ algorithm, $B$ is $0$).  These
lemmas we occasionally refer to are in \Section{balls-and-bins}. 

\subsection{A Promise Version}

We give an algorithm \textsc{LogEstimator} for estimating $L_0$ when
promised that $L_0 \le 1/(20\eps^2)$ which works as follows. 
First, we assume that the universe size is $O(1/\eps^4)$ since we can
pairwise independently hash the universe down to
$[b/\eps^4]$ for some constant $b>0$ via some hash function $h_3$.  In
doing so
we can assure that the indices contributing to $L_0$ are perfectly
hashed with constant probability arbitrarily close to $1$ by choosing
$b$ large enough.  Henceforth in this subsection we assume updates
$(i,v)$ have $i\in[U']$ for $U' = O(1/\eps^4)$.
Let $\eps' =
\eps/\max\{200,f\}$ for a constant $f$ appearing in the analysis.
We pick hash functions
$h_1: [U']\rightarrow [1/(\eps')^2]$ from a
$c_1\log(1/\eps)/\log\log(1/\eps))$-wise
independent hash family and $h_2: [U']\rightarrow [1/(\eps')^2]$ from a
pairwise independent family. The value $c_1$ is a positive constant
to be chosen later, and $h_1$ is chosen from a hash family of Siegel
\cite{Siegel04} to have constant evaluation time. The function $h_1$
should be thought of
as the function that assigns the $L_0$ items to their appropriate bins,
while $h_2$ is chosen as part of a technical solution to prevent two 
items with non-zero frequency that hash to the same bin from 
canceling each other out.

We also choose a prime $p$ randomly in $[D, D^2]$ for $D =
\log(mM)/\eps^2$. Notice that for $mM$ larger than
some constant, by standard results on the density of primes,
there are at least $\log(mM)/(400\eps^2)$ primes in
the interval $[D, D^2]$. This implies non-zero frequencies remain
non-zero modulo $p$ with good probability. Next, we randomly pick a
vector $\mathbf{u}\in \mathbb{F}_p^{1/(\eps')^2}$.

We maintain $1/(\eps')^2$ counters $C_1,C_2,\ldots,C_{1/(\eps')^2}$
modulo $p$,
each initialized to zero.  Upon receiving an update $(i,v)$, we
do
$$C_{h_1(i)} \leftarrow (C_{h_1(i)} + v\cdot\mathbf{u}_{h_2(i)})\mod p.$$
Let $I = \{i : C_i\neq 0\}$.  If
$|I| \le 100$, our estimate of $L_0$ is $|I|$.  Else, 
our estimate is $\tilde{L}_0 = \ln(1 -
(\eps')^2|I|)/\ln(1-(\eps')^2)$.

Before we analyze our algorithm, we need a few lemmas and
facts.

\begin{lemma}\LemmaName{logging-works}
Let $\mathcal{H}$ be a family of $c \cdot
\log(1/\eps)/\log\log(1/\eps)$-wise
independent hash functions $h:[U] \rightarrow [1/\eps^2]$ for
a sufficiently large constant $c > 0$. Let $S \subset [U]$ be an
arbitrary subset of $100 \leq L_0 \leq 1/(20\eps^2)$ distinct
items. Suppose we choose a random $h \in \mathcal{H}$. For $i \in
[1/\eps^2]$, let $X_i'$ be an indicator variable which is $1$ if and
only if there
is an $x \in S$ for which $h(x) = i$. Let $X' =
\sum_{i=1}^{1/\eps^2}X_i'$ and
let $Y = \ln(1 - \eps^2 X')/\ln(1-\eps^2)$.  Then there is a constant
$f>0$ so that $\Pr_h[|Y - L_0| \ge \eps f L_0] \le 1/4$.
Moreover, for any $x = (1\pm c\eps)\mu$, 
$|\ln(1 - \eps^2 x)/\ln(1-\eps^2) - L_0|\le \eps f L_0$ for a constant
$f = f(c)$, where $\mu = \eps^{-2}(1-(1-\eps^2)^{L_0})$.
\end{lemma}
\begin{proof}
We first prove the second statement.
Recall $100\le L_0 \le 1/(20\eps^2)$, implying $\eps < 1/5$.
Supposing $|x - \mu| \le c\eps\mu$ for some constant $c > 0$, we have
\begin{eqnarray*}
\frac{\ln(1-\eps^2 x)}{\ln(1-\eps^2)} &=& \frac{\ln((1-\eps^2)^{L_0}
  \pm 8\eps^3\mu)}{\ln(1-\eps^2)}\\
&=& \frac{\ln((1-\eps^2)^{L_0})}{\ln(1-\eps^2)} \pm
\frac{O(\eps^3\mu)}{\ln(1-\eps^2)}\\
&=& L_0 \pm O\left(\frac{\eps^3\mu}{\eps^2}\right)\\
&=& L_0 \pm O(\eps\mu)\\
&=& (1\pm O(\eps))L_0
\end{eqnarray*}

The second equality holds since $\eps$ is bounded away from $1$,
implying $y = (1-\eps^2)^{L_0}$ is bounded away from $0$, so
the derivative of $\ln$ at $y$ is bounded by a constant.  The third
equality similarly holds since $1-\eps^2$ is
bounded away from $0$ so that $\ln(1-\eps^2) = \Theta(\eps^2)$.  The
final equality holds since $\mu \le L_0$.  The first part of the
theorem follows since
$|X-\mu|\le 8\eps\mu$ with probability at least $3/4$ by
\Lemma{ind-consequences}.
\end{proof}

\begin{fact}\FactName{finitefield-fact}
Let $\mathbb{F}_q$ be a finite field and $v\in\mathbb{F}_q^d$ be a
non-zero vector.  Then, picking a vector $w$ at random in
$\mathbb{F}_q^d$ gives $\Pr[v\cdot w = 0] = 1/q$, where $v\cdot w$ is
the inner product over $\mathbb{F}_q$.
\end{fact}
\begin{proof}
The set of vectors orthogonal to $v$ is a linear subspace of
$\mathbb{F}_q^d$ of dimension $d-1$ and thus has $q^{d-1}$ points.  A
random $w\in\mathbb{F}_q^d$ thus lands in this subspace with
probability $1/q$.
\end{proof}

\begin{fact}\FactName{pairwise-fact}
Let $U$,$t$ be positive integers.  Pick a function $h:[U]\rightarrow
[t]$ from a pairwise independent family.  Then for any set $S\subset
[U]$ of size $s\le t$, $\E[\sum_{i=1}^s \binom{|h^{-1}(i)\cap S|}{2}]
\le s^2/(2t)$.
\end{fact}
\begin{proof}
Assume $S=\{1,\ldots,s\}$.  Let $X_{i,j}$ indicate
$h(i)=j$.  By symmetry of the $X_{i,j}$, the desired expectation is
$$
t\sum_{i<i'}\E[X_{i,1}]\E[X_{i',1}]= t 
\binom{s}{2}\frac{1}{t^2} \le \frac{s^2}{2t}$$
\end{proof}

To evaluate the hash function $h_1$ in constant time, we use the
following a theorem of Siegel, in a form that was
stated more succinctly
by Dietzfelbinger and Woelfel \cite{DW03}.

\begin{theorem}[Siegel \cite{Siegel04}]\TheoremName{siegel}
Let $0<\mu<1$ and $k\ge 1$ with $\mu k<1$ be given.  Then if $\zeta <
1$ and $d\ge 1$ satisfy $\zeta \ge \frac{2k}{d} + \frac{1+\log d +
  \mu\log z}{\zeta\log z}\cdot k$ (for $z$ large enough), then there
is a way of randomly choosing a function $h:[z^k]\rightarrow [z]$ such
that the following hold:
(1) the description of $h$ comprises
$O(z^{\zeta})$ words in $[z]$,
(2) the function $h$ can be evaluated by XOR-ing
together $d^{k/\zeta}$ $k\log z$-bit words, and
(3) the class formed by all
these $h$'s is $z^{\mu}$-wise independent.
\end{theorem}

Finally, we need the following lemma to achieve $O(1)$ reporting time.

\begin{lemma}\LemmaName{fastlog}
Let $K=1/\eps^2$ be a positive integer with $\eps < 1/2$.  It is
possible to construct a lookup
table requiring $O(\eps^{-1}\log(1/\eps))$ bits such that 
 $\ln(1 - c/K)$ can then be computed with relative accuracy $\eps$ in
 constant time for all integers $c\in[4K/5]$.
\end{lemma}
\begin{proof}
We set $\eps' = \eps/15$ and discretize the
interval $[1/5, 1-\eps^2]$ geometrically by powers
of $(1+\eps')$.  We precompute the natural algorithm evaluated at all
discretization points, with relative error $\eps/3$, taking space
$O(\eps^{-1}\log(1/\eps))$.  We
answer a query $\ln(1-c/K)$ by outputting the
natural logarithm of the closest discretization point in the
table. Our output is then, up to $(1\pm\eps/3)$,
$$\ln(1 - (1\pm\eps')c/K) = \ln(1 - c/K \pm \eps'c/K) = \ln(1 - c/K) \pm
5\eps'c/K = \ln(1 - c/K) \pm \eps c/(3K).$$
Using the fact that $|\ln(1-z)|\ge z/(1-z)$ for $0<z<1$, we have that
$|\ln(1- c/K)| \ge c/(K-c) \ge c/K$. Thus, 
$$(1\pm\eps/3)(\ln(1 - c/K) \pm \eps c/K) =
(1\pm\eps/3)(1\pm\eps/3)\ln(1 - c/K) = (1\pm\eps)\ln(1 - c/K).$$
\end{proof}

Now we analyze \textsc{LogEstimator}.

\begin{theorem}\TheoremName{main-logestimator}
Ignoring the space to store $h_3$, \textsc{LogEstimator}
uses space
$O(\eps^{-2}(\log(1/\eps) + \log\log(mM)))$.  The
update and reporting times are $O(1)$.  If $L_0 \le 1/(20\eps^2)$ then
\textsc{LogEstimator} outputs a value
$\tilde{L}_0 = (1\pm \eps)L_0$ with probability at least $3/5$.
\end{theorem}
\begin{proof}
The vector $\mathbf{u}$ takes $O(\eps^{-2}\log p) =
O(\eps^{-2}(\log(1/\eps) + \log\log(mM)))$ bits to store. Each
counter $C_i$ takes space $O(\log p)$ and there are $O(1/\eps^2)$
counters, thus also requiring $O(\eps^{-2}(\log(1/\eps) +
\log\log(mM)))$ total space.  The hash function $h_2$ requires
$O(\log(1/\eps))$ space.

For the update time, for each stream
token we must evaluate three hash functions.
The hash functions $h_2,h_3$ each take constant time.  For $h_1$, we
can
use the hash family of \Theorem{siegel} with $z = 1/\eps^2,
k=2+o(1),\mu=1/8,\zeta=1/2,d=9$.  We then have that $h_1$ is
$1/\eps^{1/4}$-wise independent, which is
$c_1\log(1/\eps)/\log\log(1/\eps)$-wise
independent for $\eps$ smaller than some constant.  Also, $h_1$
can be evaluated in constant time,
and it requires $O(\eps^{-1}\log(1/\eps))$ bits of storage. This
storage is dominated by the amount of storage required just to hold
the counters $C_i$.
We must
also multiply by a coordinate of $\mathbf{u}$ fitting in a word,
taking constant time.

For the reporting time, we can precompute $\ln(1 - (\eps')^2)$ during
preprocessing.  To compute $\ln(1 - (\eps')^2|I|)$, first note that we
can maintain $|I|$ in constant time during updates using an
$O(\log(1/\eps))$-bit counter.  Also note that
$$\E[|I|] \le (1\pm\eps)\frac{1}{(\eps')^2}\left(1 - \left(1 -
    (\eps')^2\right)^{L_0}\right) \le \frac{2}{(\eps')^2}\left(1 -
  \left(1 - \frac 1{800000}\right)\right) \le
\frac{1}{400000(\eps')^2} .$$
Thus, by Markov's inequality, $|I|\le 1/(4(\eps')^2)$ with probability
at least $99/100$, and we can use a lookup table as in
\Lemma{fastlog} compute the natural logarithm.  The space required to
store the lookup table is dominated by the space used in other parts
of the algorithm.

We now prove correctness.  First, we handle the case $100 \le L_0 <
1/(20\eps^2)$.

Let $S$ be the set of $L_0$ indices
$j\in[U']$ with $x_j\neq 0$ at the end of the stream. 

Let $\mathcal{Q}$ be the
event that $p$ does not divide any $|x_j|$. 

Let $\mathcal{Q}'$ be
the event that $h_2(j) \neq h_2(j')$ for distinct
indices $j,j'\in S$ with $h_1(j) = h_1(j')$.  

Henceforth, we condition on
both $\mathcal{Q}$ and $\mathcal{Q}'$ occurring, which we later
show holds with good probability. Define $I
\subseteq [1/(\eps')^2]$ by $I = \{i : h_1^{-1}(i)\cap S \neq
\emptyset\}$, that is, $I$ is the image of $S$ under $h_1$.
 For each $i\in I$, $C_i$ can be viewed as
maintaining the dot product of a non-zero vector $\mathbf{v}$ in
$\mathbb{F}_p^{L_0}$, the frequency vector $x$ restricted to
coordinates in $S$, with a random vector $\mathbf{w}$, namely, 
the vector obtained
by restricting $\mathbf{u}$ to coordinates in $S$.  The
vector $\mathbf{v}$ is non-zero since we condition on 
$\mathcal{Q}$, and $\mathbf{w}$ is
random since we condition on $\mathcal{Q}'$.  

Let $\mathcal{Q}''$ be
the event that no $C_i$ is zero for $i\in I$.

Conditioned on $\mathcal{Q}$, $\mathcal{Q}'$, and
$\mathcal{Q}''$, we can apply \Lemma{logging-works}, and since
$\eps' \leq \eps/f$, our estimate $\tilde{L}_0$ of $L_0$ will satisfy
$|\tilde{L}_0 - L_0| \le \eps L_0$ with
probability at least $3/4$.  

Now we analyze the probability that
$\mathcal{Q}$, $\mathcal{Q'}$, and
$\mathcal{Q''}$ all occur.  Each $|x_j|$ is at most $mM$ and thus has
at most $\log(mM)$ prime factors.  Thus, there are at most $L_0 \log(mM)
\le
\log(mM)/(20\eps^2)$ prime divisors that divide some $|x_j|$, $j\in
S$.  By our choice of $p$, we pick such a prime with probability at
most $1/20$, and thus $\Pr[\mathcal{Q}] \ge 19/20$.  

Now, let
$X_{i,j}$ be a random variable indicating that $h_1(j) = h_1(j')$ for
distinct $j,j'\in S$. Let $X = \sum_{j< j'} X_{j,j'}$.  By
\Fact{pairwise-fact} with $U=U'$, $t = 1/(\eps')^2 \ge 1/\eps^2$, and $s = L_0 <
1/(20\eps^2)$, we have that $\E[X] \le 1/(800\eps^2)$. Let $J =
\{(j,j')\in\binom{S}{2} : h_1(j)=h_1(j')\}$.  For $(j,j')\in J$
let $Y_{j,j'}$ be a random variable indicating $h_2(j) = h_2(j')$, and
let $Y = \sum_{(j,j')\in J}Y_{j,j'}$.
Then by pairwise independence of $h_2$, $\E[Y] = \sum_{(j,j')\in
  J}\Pr[h_2(j)=h_2(j')] = |J|(\eps')^2 \leq |J|\eps^2$.  
Note $|J| = X$.  Conditioned on
$X \le 20\E[X] \le 1/(40\eps^2)$, which happens with probability at
least $19/20$ by Markov's inequality, we have that $\E[Y] \leq |J|\eps^2
\le 1/40$, so that $\Pr[Y \ge 1] \le 1/40$.  Thus, $\mathcal{Q'}$
holds with probability at least $(19/20)\cdot(39/40)>7/8$.  

Finally,
by \Fact{finitefield-fact} with $q = p$, and union bounding over all
$1/\eps^2$ counters $C_i$, $\mathcal{Q''}$
holds with probability at least $1 - 1/(\eps^2p) \ge 99/100$.  Thus,
$\Pr[\mathcal{Q}\wedge\mathcal{Q'}\wedge\mathcal{Q''}] =
\Pr[\mathcal{Q}\wedge\mathcal{Q'}]\Pr[\mathcal{Q''} |
\mathcal{Q}\wedge\mathcal{Q'}] >
(19/20)\cdot(7/8)\cdot(99/100) > 4/5$ (notice that $\mathcal{Q}$ and
$\mathcal{Q'}$
are independent).  The algorithm thus succeeds
with probability at least $(4/5)\cdot(3/4) = 3/5$ in this case.

Now we consider the case $L_0 \le 100$.  If the elements of $S$ are
perfectly hashed and $\mathcal{Q}$ holds, we output $L_0$ exactly.  By
choice of $\eps'$, $1/(\eps'^2) \ge (200)^2$.  Thus, all elements of
$S$ are perfectly hashed with probability at least $7/8$ by pairwise
independence of $h_1$.  We already saw that $\Pr[\mathcal{Q}] \ge
19/20$, so we output $L_0$ exactly with probability 
$\geq (7/8)\cdot(19/20) > 3/5$.
\end{proof}

\subsection{A Rough Estimator}
For our full algorithm to function, we need
to run in parallel a subroutine giving a constant-factor approximation
to $L_0$. We describe here a subroutine \textsc{RoughEstimator} which
does exactly this.
First, we need the following lemma which states that when $L_0$ is at
most some constant $c$, it can be computed exactly in small space.
The lemma follows by picking a random prime $p =
\Theta(\log(mM)\log\log(mM))$ and pairwise independently hashing the
universe into
$[\Theta(c^2)]$ buckets.  Each bucket is a counter which tracks of the
sum of frequencies modulo $p$ of updates to universe items landing in
that bucket.  The estimate of $L_0$ is then the total number of
non-zero counters, and the maximum estimate after $O(\log(1/\eta))$
trials is finally output. This gives the following.

\begin{lemma}\LemmaName{simple-exact-parity}
There is an algorithm which, when given the promise that $L_0 \le c$,
outputs $L_0$ exactly with probability at least
$1-\eta$ using
$O(c^2\log\log(mM))$ space, in addition to
needing to store $O(\log(1/\eta))$ independently chosen pairwise
independent hash
functions mapping $[U]$ onto $[c^2]$.  The
update and reporting times are $O(1)$.
\end{lemma}

Now we describe \textsc{RoughEstimator}.
We pick a function $h:[U]\rightarrow [N]$ at random
from a pairwise independent family.
For each $0 \le j \le \log N$ we create a substream
$\mathcal{S}^j$ consisting of those $x\in [U]$ with
$\lsb(h(x)) = j$. Let
$L_0(\mathcal{S})$
denote $L_0$ of the substream $\mathcal{S}$.
For each $\mathcal{S}^j$ we run an instantiation $B_j$ of
\Lemma{simple-exact-parity} with $c=141$ and $\eta =
1/16$.  All instantiations share the same $O(\log(1/\eta))$ hash
functions $h^1,\ldots,h^{\Theta(\log(1/\eta))}$.

To obtain our final estimate of $L_0$ for the entire stream, we find
the largest value of $j$ for which
$B^j$ declares $L_0(\mathcal{S}^j) > 8$.
Our estimate of $L_0$ is $\tilde{L}_0 = 2^j$.  If no such $j$
exists, we estimate $\tilde{L}_0 = 1$.  Finally, we run this entire
procedure $O(1)$ times and take the median estimate.

\begin{theorem}\TheoremName{rough-est}
With probability at least $99/100$
\textsc{RoughEstimator} outputs
a value $\tilde{L}_0$ satisfying $L_0 \le \tilde{L}_0 \le
110 L_0$.  The space used is $O(\log(N)\log\log(mM))$,
and the update and reporting times are $O(1)$.
\end{theorem}
\begin{proof}
We first analyze one instantiation of \textsc{RoughEstimator}.
The space to store $h$ is $O(\log N)$.
The $\Theta(\log(1/\eta))$ hash functions $h^i$ in total require
$O(\log(1/\eta)\log U) = O(\log N)$ bits to store since
$1/\eta = O(1)$. The remaining space to store a
single $B^j$ for a level is
$O(\log\log(mM))$ by \Lemma{simple-exact-parity}, and thus storing all
$B^j$ across all levels requires space $O(\log(N)\log\log(mM))$.

As for running time, upon receiving a stream update $(x,v)$, we first
hash $x$ using $h$, taking time $O(1)$.
Then, we compute $\lsb(h(x))$, also in constant time
\cite{Brodnik93,FredmanWillard93}.
Now,
given our choice of $\eta$ for $B^j$, we can update $B^j$
in $O(1)$ time by \Lemma{simple-exact-parity}.

To obtain $O(1)$ reporting time, we again use the fact that we can
compute the least significant bit of a machine word in constant time.
We maintain a
single machine word $z$ of at least $\log N$ bits and treat it as a bit
vector.  We maintain that the $j$th bit of $z$ is $1$ iff
$L_0(\mathcal{S}^j)$ is reported to be at least $8$ by $B^j$.  This
property can be maintained in
constant time during updates.  Constant reporting time then follows
since finding the deepest level $j$ with at least $8$ reported
elements is equivalent to computing $\lsb(z)$.

Now we prove correctness.  Observe that
$\E[L_0(\mathcal{S}^j)] = L_0/2^{j+1}$ when $j < \log N$ and
$\E[L_0(\mathcal{S}^j)] = L_0/2^j = L_0/N$ when $j = \log N$.
Let $j^*$ be the largest $j$ satisfying
$\E[L_0(\mathcal{S}^j)] \ge
1$ and note that $1 \le \E[L_0(\mathcal{S}^{j^*})]\le 2$.  For any
$j>j^*$, $\Pr[L_0(\mathcal{S}^j) > 8] \le 1/(8\cdot 2^{j-j^*-1})$ by
Markov's inequality.  Thus, by a union bound, the probability that any
$j>j^*$ has $L_0(\mathcal{S}^j) > 8$ is at most
$(1/8)\cdot\sum_{j-j^*=1}^{\infty} 2^{-(j-j^*-1)}
= 1/4$.  Now, let $j^{**}<j^*$ be the largest $j$ such that
$\E[L_0(\mathcal{S}^j)] \ge 55$, if such a $j$
exists.  Since we increase the $j$ by powers of $2$, we have
$55\le \E[L_0(\mathcal{S}^{j^{**}})]< 110$. Note that $h$ is pairwise
independent, so $\Var[L_0(\mathcal{S}^{j^{**}})] \le
\E[L_0(\mathcal{S}^{j^{**}})]$.
For this range of $\E[L_0(\mathcal{S}^{j^{**}})]$, we then
have by Chebyshev's inequality that
$$\Pr\left[|L_0(\mathcal{S}^{j^{**}} )- \E[L_0(\mathcal{S}^{j^{**}})]| \ge
3\sqrt{\E[L_0(\mathcal{S}^{j^{**}})]}\right] \le 1/9$$

If $|L_0(\mathcal{S}^{j^{**}})-\E[L_0(\mathcal{S}^{j^{**}})]| <
3\sqrt{\E[L_0(\mathcal{S}^{j^{**}})]}$, then 
$$32<55-3\sqrt{55}<L_0(\mathcal{S}^{j^{**}}) < 110+3\sqrt{110}<142$$ 
since $55\le \E[L_0(\mathcal{S}^{j^{**}})] < 110$.

So far we have shown that with probability at least $3/4$,
$L_0(\mathcal{S}^j)
\le 8$ for all $j > j^*$.  Thus, for these $j$ the $B^j$ will estimate
$L_0$ of the corresponding substreams to be at most $8$, and we
will not output $\tilde{L}_0 = 2^j$ for $j > j^*$.  On the other
hand, we know for $j^{**}$ (if it exists) that with probability at
least $8/9$, $\mathcal{S}^{j^{**}}$ will have $32 <
L_0(\mathcal{S}_i^{j^{**}})< 142$.  By our choice of $c =
141$ and $\eta = 1/16$ in the $B^j$, $B^{j^{**}}$ will output a value
$\tilde{L}_0(\mathcal{S}_i^{j^{**}}) \ge L_0(\mathcal{S}_i^{j^{**}})/4
> 8$ with
probability at least $1 - (1/9 + 1/16) > 13/16$ by
\Lemma{simple-exact-parity}.  Thus, with
probability at least $1 - (3/16 + 1/4) = 9/16$, we output
$\tilde{L}_0 = 2^j$ for some $j^{**} \le j \le
j^*$, which satisfies $110\cdot 2^{j} < L_0 \le 2^{j}$. If such
a $j^{**}$ does not exist, then $L_0 <
55$, and thus
$1$ serves as a $55$-approximation in this case.

Since one instantiation of \textsc{RoughEstimator} gives the desired
approximation with constant probability strictly greater than $1/2$
(i.e. $9/16$), the theorem follows by taking the median of a constant
number of independent instantiations and applying a Chernoff bound.
\end{proof}

\subsection{Putting the Final Algorithm Together}\SectionName{full-l0}
Our full algorithm \textsc{FullAlg} for estimating $L_0$ works as
follows.  Set $\eps' = \eps/420$.
Choose a
$c_1\log(1/\eps')/\log\log(1/\eps')$-wise independent hash function
$h_1$,
pairwise independent hash functions $h_2,h_3$, and random prime
$p\in[D,D^2]$ for $D = \log(mM)/\eps^2$, as is
required by \textsc{LogEstimator}.  We run
an instantiation $\LE$ of \textsc{LogEstimator} with desired error
$\eps'$,
an instantiation $\RE$ of \textsc{RoughEstimator}, and $\log N
- \log(1/(\eps')^2) = \log((\eps')^2 N)$
instantiations $\LE_{0},\ldots,\LE_{\log((\eps')^2 N)}$
of
\textsc{LogEstimator} in parallel with the promise
$L_0 \le 1/(20(\eps')^2)$ and desired error
$\eps'$. All instantiations of
\textsc{LogEstimator} share the same $h_1$, $h_2$, $h_3$, and
prime $p$. We pick a
hash function $h:[U]\rightarrow [N]$ at random from pairwise
independent family of hash functions.  For each update
$(i,v)$ in the stream, we feed the update to both $\LE$ and $\RE$.
Also, if the length $j$ of the longest suffix of zeroes in $h(i)$ is
at most $\log(1/(\eps')^2)$, we feed the update $(i,v)$ to
$\LE_j$.

Let $R$ be the estimate of $L_0$ provided by $\RE$.  If
$R<1/(20(\eps')^2)$, we output the estimate provided by
\textsc{LE}. Otherwise, we output the estimate of $\tilde{L}_0$
provided by $\LE_{\ceil{\log (R/(4400(\eps')^2))}}$. To analyze our
algorithm, we first prove the following lemma.

\begin{lemma}\LemmaName{subsamplingworks}
Let $j$ be a level such that $20/\eps^2 \le
\E[L_0(\mathcal{S}_i^j)]$.
Then
$|2^{j'} L_0(\mathcal{S}_i^j) - L_0| \le 2\eps
L_0/3$
with probability at least $7/8$ for $j' = j$ when $j = \log N$, and
$j' = j+1$ otherwise.
\end{lemma}
\begin{proof}
Let $S = \{i:x_i\neq 0\hbox{ at the end of the stream}\}$ and for
$i\in S$ let $X_{i,j}$ be a random variable indicating that $i$ is
hashed to the substream at level $j$, and let $X_j = \sum_{i\in S}
X_{i,j}$.  We assume here $j<\log N$ since
the proof is nearly identical for $j=\log N$.  Then we have
$\E[X_j] = L_0/2^{j+1}$, and by pairwise independence of $H$,
$\Var[X_j] \le \E[X_j]$.  Thus by Chebyshev's inequality, 
$$\Pr[|2^{j+1}X - L_0| \ge 2\eps L_0/3] \le \frac{9\E[X]}{4\eps^2 \E^2[X]}
< \frac{1}{8}$$
\end{proof}

Now we prove our main theorem for $L_0$ estimation.

\begin{theorem}\TheoremName{main-thm}
\textsc{FullAlg} uses space $O(\eps^{-2}\log(\eps^2
N)(\log(1/\eps) + \log\log(mM)))$,
has $O(1)$ update and reporting times, and
$(1\pm\eps)$-approximates $L_0$ with probability at least $3/4$.
\end{theorem}
\begin{proof}
We analyze one instantiation of \textsc{FullAlg}.
The space and time requirements follow from \Theorem{main-logestimator}
and \Theorem{rough-est}, and the
fact that the hash functions $h,h_3$ can be stored in $O(\log U) =
O(\log N)$ bits
and can be evaluated in constant time.

As for correctness, with probability at least $99/100$, the value $R$
returned by $\RE$ satisfies $L_0 \le R \le 110 L_0$ by
\Theorem{rough-est}.  We henceforth condition on this occurring.  If
$R<1/(20(\eps')^2)$ then $L_0<1/(20(\eps')^2)$,
so $\LE$ outputs $(1\pm\eps')L_0 = (1\pm\eps)L_0$ with probability at
least $3/5$ by \Theorem{main-logestimator}.  Otherwise, we output the
estimate of $\tilde{L}_0$
provided by $\LE_j$ for $j = \ceil{\log
  (R/(4400(\eps')^2))}$.  Let $L_0^j$ denote the expected value $L_0$
of the
substream at level $j$.  For our choice of $j$,
$L_0/(8800(\eps')^2) \le \E[L_0^j] \le L_0/(40(\eps')^2)$.  By
\Lemma{subsamplingworks} and choice of $\eps'$,
$(1\pm(2\eps/3))L_0/(8800(\eps')^2) \le L_0 \le
(1\pm(2\eps/3))L_0/(40(\eps')^2) \le L_0/(20(\eps')^2)$
with probability at least $7/8$.  By \Theorem{main-logestimator},
conditioned on $L_0^j \le L_0/(20\eps')^2$ and by choice of $\eps'$, we
have that
$\LE_j$ outputs $(1\pm \eps')L_0^j = (1\pm(\eps/420))L_0^j$ with
probability at least $3/5$.  Again by
\Lemma{subsamplingworks}, using that $ 20/\eps^2\le 1/(8800(\eps')^2)$
by choice of $\eps'\le\eps/420$, we have that $2^{j+1}L_0^j$ serves as a
$(1\pm(\eps/420))(1\pm(2\eps/3))$-approximation to $L_0$ in this case,
which is at most $(1\pm\eps)$ for $\eps$ smaller than some constant.
Thus, in the case $R\ge
1/(20(\eps')^2)$, \textsc{FullAlg} outputs a valid approximation with
probability at least $(3/5)\cdot(7/8) > 33/64$.  Thus, in total, the
algorithm outputs a valid approximation with probability at least
$(99/100)\cdot(33/64)$ (since we conditioned on $R$ being a valid
approximation), which is strictly bigger than $1/2$.  The theorem
follows by repeating a constant number of instantiations of
\textsc{FullAlg} in parallel and returning the median result.
\end{proof}

When given $2$ passes, in the first pass we can obtain $R$, then in
the second pass we need only instantiate $\LE_j$ for the appropriate
level $j$, thus avoiding the $\log(\eps^2N)$ factor blowup in space
from maintaining $\log(\eps^2N)$ different $\LE_j$. Thus we have the
following theorem.

\begin{theorem}\TheoremName{2pass-thm}
There is an algorithm $(1\pm\eps)$-approximating $L_0$ in $2$ passes
with probability $3/4$,
using space $O(\eps^{-2}(\log(1/\eps) +
\log\log(mM)) + \log N)$, with $O(1)$ update and reporting times.
\end{theorem}

Note that when combined with \Theorem{main-lowerbound},
\Theorem{2pass-thm} shows a separation between the space complexity of
$1$ and $2$ passes for
$L_0$ for a large range of settings of $\eps$ and $mM$.



\section{$L_0$ in update-only streams}\SectionName{f0}
Here we describe an algorithm for estimating $F_0$, the
number of distinct items in an update-only stream. 
Our main result is the following. The space
bound is never more than a $O(\log\log N)$ factor away from optimal,
for any $\eps$.

\begin{theorem}\TheoremName{optimal-f0}
There is an algorithm for $(1\pm\eps)$-approximating
$F_0$ with probability $2/3$ in space
$O(\eps^{-2}\log\log(\eps^2N)+ \log(1/\eps)\log(N))$.
The update and reporting times are both $O(1)$.
\end{theorem}

The algorithm works as follows.  We allocate $K = 1/\eps^2$ counters
$C_1,\ldots,C_K$ initialized to \texttt{null}, each capable of holding
an integer in $[\log(\eps^2
N)+1]$, and we
pick an $O(\log(1/\eps)/\log\log(1/\eps))$-wise independent
hash function $h_1:[1/\eps^4]\rightarrow[K]$.  We also
pick pairwise independent hash
functions $h_2:[U]\rightarrow[1/\eps^4]$ and
$h_3:[U]\rightarrow[N]$. We run Algorithm I of \cite{BJKST02} to
obtain a
value $F_0/2 \le R\le F_0$ with probability $99/100$, taking
$O(\log N + \log\log n)$ space and has constant update
and reporting time\footnote{The
  space and time bounds are not listed this way in \cite{BJKST02}
  because (1) they do not assume the word RAM model, and (2) they do
  not ensure $U = O(\log N)$ but rather just use a universe of size
  $n$.}.  Upon seeing an item $i\in[U]$ in the stream, we set
$$ C_{h_1(h_2(i))} \leftarrow \max\{C_{h_1(h_2(i))}, \min\{
\log(\eps^2 N)+1, \lsb(h_3(i))\}\} .$$
We also maintain $\log(\eps^2 N)$ counters $Y_1,\ldots Y_{\log(\eps^2
  N)}$, where $Y_r$ tracks $|\{j : C_j = r\}|$. To
estimate $F_0$ there are three cases. If $R\le 100$, we output $|\{j :
C_j\neq\texttt{null}\}|$.  Else, if $100<R\le K/40$, we output
$\ln(1 - |\{j : C_j\neq\texttt{null}\}|/K)/\ln(1-1/K)$. Otherwise,
let $r$ be the smallest positive integer such that $R/2^r \le K/40$.
We define
$f(A) = K((1 - 1/K)^A - (1 - 1/K)^{2A})$ and output
$2^rA$ for the smallest $A$ with $f(A)= Y_r$. For time efficiency, $h_1$
is chosen from a hash
family of Siegel \cite{Siegel04} to have $O(1)$ evaluation time.

We now analyze our algorithm.  First, we need the following two
lemmas, whose proofs are in \Section{f0-proofs}.

\begin{lemma}\LemmaName{bounded-deriv}
Fix $x\ge 2$. Consider the function
$$ f(y) = x\left(\left(1 - \frac 1x\right)^y - \left(1 - \frac
    1x\right)^{2y}\right).$$
If $y\le x/3$, then $f'(y)\ge 1/9$.
\end{lemma}

\begin{lemma}\LemmaName{AequalsB}
Let $0\le \eps < 1/2$ and suppose $(1 - \eps) B \le B' \le (1 +
\eps)B$ with $0\le B\le K$ for integers $B,B'$. If $A,K\ge 0$ then
$$ K\left(1 - \left(1 - \frac 1K\right)^A\right)\left(1 - \frac
  1K\right)^B = (1 \pm 2\eps)K\left(1 - \left(1 - \frac
  1K\right)^A\right)\left(1 - \frac 1K\right)^{B'}$$
\end{lemma}

We now prove correctness of our $F_0$ algorithm and analyze the
update and reporting times and space usage.

\begin{proofof}{\Theorem{optimal-f0}}
 Our use of an algorithm of \cite{BJKST02} to obtain
a $R$ requires $O(\log N)$ space and adds $O(1)$ to both
the update and reporting time.  Now we analyze the rest of our
algorithm.

First we analyze space requirements.
We maintain $1/\eps^2$ counters $C_j$, each holding an integer in
$[\log(\eps^2 N) + 1]$ (or \texttt{null}), taking
$O(\eps^{-2}\log\log(\eps^2 N))$ bits.  Storing $h_1$ from Siegel's family
takes $O(\eps^{-1}\log(1/\eps)) = o(\eps^{-2})$ bits, as in the
proof of \Theorem{main-logestimator}. The functions $h_2,h_3$ combined
require $O(\log N + \log(1/\eps))$ bits.  Finally, the last bits of
required storage come from storing the $Y_r$, which in total take
$O(\log(K)\log N) = O(\log(1/\eps)\log(N))$ bits.

Now we analyze update time. Each update requires evaluating
each of $h_1,h_2,h_3$ once, taking $O(1)$ time.  We also compute the
$\lsb$ of an integer fitting in a word, taking $O(1)$ time
using \cite{Brodnik93,FredmanWillard93}.  Finally, we have to maintain
the $Y_r$.  During an update, we change the value of at most one
$C_j$, from, say, $r$ to $r'$.  This just requires decrementing $Y_r$ and
incrementing $Y_{r'}$.

Before analyze reporting time, we prove correctness.  We condition on
the event $\mathcal{Q}$ that $F_0/2 \le R \le F_0$.  If $R\le
100$, then $F_0\le 200$ and
the distinct elements are perfectly hashed with $7/8$ probability (for
$\eps$
sufficiently small), and we estimate $F_0$ exactly in this case.  If
$100 < R\le K/20$, correctness follows from
\Lemma{logging-works}.  We now consider $R > K/20$.
We consider the level $r$ with 
$$\frac{1}{80\eps^2} < \frac{R}{2^r} \le \frac{1}{40\eps^2}$$
and thus
$$\frac{1}{80\eps^2} < \frac{F_0}{2^r} \le \frac{1}{20\eps^2}$$
Letting $F_0'$ be the number of distinct elements mapped to level $r$,
we condition on the event $\mathcal{Q}'$ that $F_0' = (1\pm
50\eps)F_0/2^r$. For $\eps$ sufficiently small, this implies
$$\frac 1{160\eps^2} < \frac{F_0}{2^{r+1}} \le F_0' \le
\frac{F_0}{2^{r-1}} \le \frac 1{10\eps^2}.$$
We also let $F_0''$ be the number of distinct elements mapped to
levels $r'>r$ and condition on the event $\mathcal{Q}''$ that $F_0'' =
(1\pm 50\eps)F_0/2^r$.  This similarly implies
$$\frac 1{160\eps^2} \le F_0'' \le \frac 1{10\eps^2}.$$

Next, we condition on the event $\mathcal{Q}'''$ that the $F_0'+F_0''
\le 1/(5\eps^2)$
items at levels $r$ and greater are perfectly hashed under $h_2$.
Now we use our analysis of  the balls and bins random process
described 
in \Section{balls-and-bins} with $A=F_0'$ ``good balls'' and $B=F_0''$
``bad balls''. Let $X'$ be the random variable counting the number of
bins $C_j$ hit by good balls under $h_1$.  By
\Lemma{exactgoodbadballs} and \Lemma{ind}, $\E[X] = (1\pm\eps)\mu$
with
$$\mu = K\left(1 - \left(1 - \frac
    1K\right)^A\right)\left(1 -
  \frac 1K\right)^B$$
We define the event
$\mathcal{Q}''''$ that $ X' = (1\pm 4002\eps)\mu$.

Recall the definition of the function
$$ f(A') = K\left(1 - \left(1 - \frac 1K\right)^{A'}\right)\left(1 -
  \frac 1K\right)^{A'} = K\left(\left(1 - \frac 1K\right)^{A'} -
  \left(1 - \frac 1K\right)^{2A'}\right)$$
Conditioned on $\mathcal{Q},\mathcal{Q}'$, $A = (1\pm 100\eps)B$, and
thus by \Lemma{AequalsB}, $\mu = (1\pm 200\eps)f(A)$.

Conditioned on $\mathcal{Q}''''$, 
$$|X' - f(A)| \le |X' - \mu| + |\mu - f(A)| \le 4002\eps\mu + 200\eps
f(A) \le 4202\eps K,$$
in which case also $|X' - f(A)|\le K/1000$ for $\eps$ sufficiently small,
implying $K/1000\le X' \le K/9$. The lower bound holds since $f(A)\ge
K/500$ by
\Lemma{badballsvar}, and the upper bound holds since $f(A)\le A\le
K/10$.  We also note $f(K/3)
\ge K(e^{-1/3} - e^{-2/3} - 1/K)$ by \Lemma{mr95}, which is at least
$K/9$ for $K$
sufficiently large (i.e. $\eps$ sufficiently small).
Thus, there exists $A'\le K/3$ with  $f(A') = X'$. Furthermore,
by \Lemma{bounded-deriv} $A'$ is the unique inverse in this range.
Also, in the range where we invert $X'$, the derivative of $f$ is
lower bounded by $1/9$, and so
$$|f^{-1}(X') - A| \le (9\cdot 4202)\eps K \le
10^7\eps A .$$
Thus, we can compute $A$ with relative error $10^7\eps A$, and so $2^r
A = (1\pm 50\eps)(1\pm 10^7\eps)F_0$.  We can thus obtain
$(1\pm\eps)F_0$ by running our algorithm with error parameter $\eps' =
c\eps$ for $c$ a sufficiently small constant. Thus, our algorithm is
correct as long as
$\mathcal{Q},\mathcal{Q}',\mathcal{Q}'',\mathcal{Q}''',\mathcal{Q}''''$
all occur.

Now we analyze the probability that all these events occur. We already
know $\Pr[\mathcal{Q}] \ge 99/100$ by
our
choice of failure probability when running the algorithm of
\cite{BJKST02}.
By Chebyshev's inequality,
$$ \Pr[\mathcal{Q'}|\mathcal{Q}] \ge 1 - 
\frac{2^r}{50^2\eps^2 F_0} \ge 1 - \frac{80}{50^2} \ge \frac
{19}{20}$$
and the exact same computation holds for lower bounding
$\Pr[\mathcal{Q}''|\mathcal{Q}]$. 

Now we bound $\Pr[\mathcal{Q}'''|\mathcal{Q}'\wedge
\mathcal{Q}'']$.  Arbitrarily label the $z = F_0'+F_0''$ balls as
$1,2,\ldots,z$ with $z\le K/5$.  Let $Z_{i,j}$ indicate that $h_2(i) =
h_2(j)$.  Then the expected number of collisions is at most
$((K/5)^2/2)\cdot (1/K^2) = 1/50$.  Thus, by Markov's inequality,
$\Pr[\mathcal{Q}'''|\mathcal{Q}'\wedge \mathcal{Q}''] \ge 49/50$.

By \Lemma{bigballexpectation},
\Lemma{badballsvar}, and \Lemma{ind},
$$\E[X'] \ge (1-\eps)K/500,\ \Var[X'] \le 7K + \eps^2$$
and thus by Chebyshev's inequality,
$$ \Pr[|X' - \E[X']| \le
4000\eps\E[X']|\mathcal{Q}'''] \ge 1 - 
\frac{8K}{4000^2\eps^2 (1-\eps)^2(K/500)^2} \ge \frac
{13}{16}$$
with the last inequality holding for $\eps$ sufficiently small. When
$|X' - \E[X']| \le 4000\eps\E[X']$ occurs,
then $X' = (1\pm 4002\eps)\mu$, implying $\mathcal{Q}''''$
occurs.  Thus, by the above and exploiting independence of some of the
events,
\begin{eqnarray*}
\Pr[\mathcal{Q}\wedge\mathcal{Q}'\wedge\mathcal{Q}''
\wedge\mathcal{Q}'''\wedge\mathcal{Q}'''']
&\ge& \Pr[\mathcal{Q}]\cdot(1 - \Pr[\bar{\mathcal{Q}'}|\mathcal{Q}] -
\Pr[\bar{\mathcal{Q}''}|\mathcal{Q}])\\
&& \cdot\
\Pr[\mathcal{Q}'''|\mathcal{Q} \wedge \mathcal{Q}' \wedge
\mathcal{Q}'']\\
&& \cdot\ \Pr[\mathcal{Q}''''|\mathcal{Q} \wedge
\mathcal{Q}' \wedge \mathcal{Q}'' \wedge \mathcal{Q}''']\\
&=& \Pr[\mathcal{Q}]\cdot(1 - \Pr[\bar{\mathcal{Q}'}|\mathcal{Q}] -
\Pr[\bar{\mathcal{Q}''}|\mathcal{Q}])\\
&& \cdot\
\Pr[\mathcal{Q}'''|\mathcal{Q}' \wedge
\mathcal{Q}'']\\
&& \cdot\ \Pr[\mathcal{Q}''''|\mathcal{Q}''']\\
&\ge&
\left(\frac {99}{100}\right)\cdot \left(1 - \frac{2}{20}\right)\cdot
\left(\frac{49}{50}\right) \cdot \left(\frac{13}{16}\right)\\
&>& 2/3
\end{eqnarray*}

Finally, we analyze the reporting time.  Recall we can query for $R$
in constant time.
In the case $R \le 100$, we output the number of non-\texttt{null}
bins, which we can maintain in constant time during updates using an
$O(\log(1/\eps))$-bit counter.
For $100< R\le K/40$, our
reporting time is $O(1)$ by using
\Lemma{fastlog}.  Otherwise, we need to find the smallest positive $A$
satisfying $K((1 - 1/K)^A - (1 - 1/K)^{2A}) = Y_r$. For this we can
discretize the interval $I = [f(K/1000),f(K/9)]$ into
$\Theta(1/\eps)$ evenly-spaced points $\mathcal{P}$ and precompute
$f^{-1}(p)$ for all $p\in\mathcal{P}$ during preprocessing.  We can
then compute $f^{-1}(x)$ for any $x\in I$ by table lookup, using the
nearest element of $\mathcal{P}$ to $x$, thus inverting $f$ with at
most an additive $\pm\eps K/160 = \pm\eps A$ error. Note we argued
above that $X'$ will be in $I$ conditioned on the good events.
Also, this upper bound
on the error suffices for our algorithm's correctness.
\end{proofof}

\section*{Acknowledgments}
We thank Nir Ailon, Erik Demaine, Avinatan Hassidim, Piotr Indyk,
T.S. Jayram, Swastik Kopparty, John Nolan, Mihai P\v{a}tra\c{s}cu, and
Victor Shoup for valuable discussions and references.  We also thank
Chris Umans and Salil Vadhan, both of whom shared insights that were
helpful in implementing the GUV extractor in linear space.

\bibliographystyle{plain}
\bibliography{./allpapers}
\newpage
\appendix

\section{Appendix}
\subsection{Small Universe Justification}\SectionName{small-universe}
If $n < m^2$, we can do nothing and already have a universe of size
$n$.  Otherwise, let
$\{i_1,\ldots,i_r\}$ be the set of indices appearing in the stream.
Picking a prime $q$ and treating all updates $(i,v)$ as $(i\mod q,
v)$, our estimate of $L_0$ will be unaffected as long as $i_{j_1} \neq
i_{j_2} \mod q$ for any $j_1\neq j_2$. There are at most $r^2/2$
differences $|i_{j_1}-i_{j_2}|$, and each difference is an integer
bounded by $n$, thus having at most $\log n$ prime factors.  There are
thus at most $r\log n$ prime factors dividing {\em some}
$|i_{j_1}-i_{j_2}|$.  If we pick a random prime
$q\in[r\log(n)\log(r\log(n)), c\cdot r\log(n)\log(r\log(n))]$ for a
sufficiently large constant $c$, we can ensure with constant
probability arbitrarily close to $1$ (by increasing $c$) that no
indices collide modulo $q$.  Since $r\le m$, we can pick $q =
O(\mathrm{poly}(m\log n))$.  We then pick a hash
function $h:\{0,\ldots,q-1\}\rightarrow [O(m^2)]$ at random from
pairwise independent family.  With constant probability which can be
made arbitrarily high, the mapping $i\mapsto h(i\mod q)$ perfectly
hashes the indices appearing in the stream.  Storing both $h$ and $q$
requires $O(\log q + \log m)= O(\log m + \log\log n)$ bits.  Since we
only apply this scheme when $m^2\le n$, the $O(\log m)$ term only
appears in our space bounds when $\log m = O(\log n)$.  Thus, the cost
of this scheme is $O(\log N + \log\log n)$, and the $\log N$ term is
dominated by other factors in all our space bounds.

\subsection{Notes on the Proof of \Lemma{median-lp}}\SectionName{magic-function}

In our proof of \Lemma{median-lp}, we needed a function
$f:\C\rightarrow \C$ such that the following properties hold when
restricting $f$ to $\R$:
\begin{itemize}
\item $f$ is an even function
\item $f$ decreases strictly monotonically to $0$ as $x$ tends away
  from $0$
\item $f$ is strictly positive
\end{itemize}

Also, to apply \Corollary{mcor}, we needed $f$ to be holomorphic on $\C$, and we
needed $|f(z)| = e^{O(1+\Im(z))}$ for all $z\in\C$.  We now justify
why 
$$f(x) = -\int_{-\infty}^{x}\frac{\sin^4(y)}{y^3}dy$$
has all these properties. First, note the integral exists
for all $x$ and thus $f$ is well-defined.  Now,
$f$ is even since it is
the integral of an odd function.  It decreases monotonically to $0$ as
$x$ tends away from $0$ since the sign of $f'(x)$ is the sign
of $-\sin^4(x)/x^3$, which is just the sign of $-x$. It is strictly
positive since on
the negative reals it is the integral of a strictly positive function,
also implying that $f$ is strictly positive on the positive reals 
since it is even. This also implies $f(0)>0$ since $f$ is maximized at
$0$.

Now, $f$ is holomorphic on $\C$ by construction: it is the integral of
a holomorphic function on $\C$.  To see that $f'$ is holomorphic, note $f'(z) =
\mathrm{sinc}^3(z)\sin(z)$ is the product of holomorphic functions.
Lastly, we need to show that $|f(z)| = e^{O(1+\Im(z))}$.  This can be
seen using Cauchy's integral theorem, which lets us choose a
convenient curve when
computing the line integral from $-\infty$ to $z$ of $f'$.  We choose
the curve which goes from $-\infty$ to $\Re(z)$, then goes from $\Re(z)$
to $\Re(z) + i\Im(z)$, thus integrating the real and imaginary axes
separately (here $\Re(z)$ denotes the real part of $z$).  The integral
on the real part of the curve is bounded by
a constant.  The integral on the imaginary part is bounded by
$e^{O(1+|\Im(z)|)}$ since $\sin(z) = (e^{-\Im(z) + i\Re{z}} -
e^{\Im(z) - i\Re(z)})/2$. Each term in the difference is bounded in
magnitude by $e^{|\Im(z)|}$.

We also comment on making the constants $C,C'$ explicit in the proof
of \Lemma{median-lp}.  Recall, for the function $g(c) = \E[f(cZ)]$
(where $Z\sim \mathcal{D}_p$), we picked positive constants
$C$ large enough and $C'$ small enough such that $g(C)$ and $g(C')$
landed in some desired
range. Knowing $C,C'$ is necessary to understand the quality of the
constant-factor approximation the median estimator gives.
These $C,C'$ depend on $p$, and can be found during
preprocessing in constant time
and space (as a function of constant $p$) as follows. First, note
$g(c)$ is strictly decreasing
on the positive reals with $g(0) = f(0)$ and $\lim_{c\rightarrow
  \infty} g(c) = 0$, and thus we can binary search, using the usual
trick of geometrically growing the interval size we search in since
we do not know it a priori.  The question then
becomes how to evaluate $g(c)$ at each iteration of the search.
$\E[f(cZ)]$ is defined as the integral
$$ \int_{-\infty}^{\infty} f(cx)q(x) dx$$
where $q$ is the probability density function of $\mathcal{D}_p$.
We only need to compute this integral to within constant accuracy, so
we can compute this integral numerically in constant time and
space.  We note a clumsy implementation would have to numerically
integrate in a $2$-level recursion, since $f$ and $q$ themselves are
defined as integrals for which we have no closed form.  A slicker
implementation can use Parseval's theorem, which tells us that
$$ \int_{-\infty}^{\infty} f(cx)q(x) dx = \int_{-\infty}^{\infty}
\frac 1c \hat{f}\left(\frac {\xi}{c}\right)\hat{q}(\xi) d\xi$$
We claim that the latter integral lets us avoid the recursive
integration step because {\em we do} have closed forms for
$\hat{f},\hat{p}$.  By definition of $p$-stability,
$\hat{q}(\xi)=e^{-|\xi|^P}$.  For $\hat{f}$, recall $f' =
\mathrm{sinc}^3(x)\sin(x)$.  The Fourier transform of $\mathrm{sinc}$
is the indicator function of an interval, and that of the $\sin$ function
is the difference of two shifted $\delta$ functions, scaled by an imaginary
component.  By convolution, the transform of $f'$ is thus a
piecewise-polynomial that
can be written explicitly, and thus we can compute $\hat{f}$
explicitly since integration corresponds to division by $i\xi$ in the
Fourier domain.

\subsection{Details of the Improvement to Armoni's PRG}\SectionName{appendix-armoni}
\subsubsection{GUV Extractor Preliminaries}
The following preliminary definitions and theorems will be needed
throughout \Section{appendix-armoni}.

\begin{theorem}\TheoremName{explicit-irr}
The following families of polynomials are irreducible over the given
rings:
\begin{enumerate}
\item[(1)] $x^{2\cdot 3^{\ell}} + x^{3^{\ell}} + 1 \in \mathbb{F}_2[x]$,
  $\ell\ge 0$
\item[(2)] $x^{2^{\ell}} + 2x^{2^{\ell - 1}} - 1 \in \mathbb{F}_7[x]$,
  $\ell\ge 1$
\item[(3)] $x^{3^{\ell}} + 3 \in \mathbb{F}_7[x]$, $\ell \ge 0$
\end{enumerate}
\end{theorem}
\begin{proof}
Polynomials in family (1) are shown irreducible in Theorem 1.1.28
of \cite{vanLint99}.  Polynomials in families (2) and (3) are
shown irreducible in Examples 3.1 and 3.2 of \cite{Menezes93}.
\end{proof}

\begin{theorem}[\cite{LN94}, Corollary 3.47]\TheoremName{lidl}
Let $p$ be an irreducible polynomial over $\mathbb{F}_q[x]$ of degree
$d$.  Then $p$ is irreducible over $\mathbb{F}_{q^m}[x]$ if and only
if $\hbox{gcd}(m, d) = 1$.\afterproof
\end{theorem}

The following fact is folklore.

\begin{fact}\FactName{obvious-efficiency}
Multiplication and division with remaindering of two polynomials of
degree at most $n$ in $\mathbb{F}_q[x]$ can be performed in
time $\hbox{poly}(n\log q)$ and space $O(n\log q)$.
\end{fact}

\begin{definition}
A $D$-regular bipartite graph $\Gamma : [N] \times [D] \rightarrow
[M]$ is a {\em $(\le K, A)$ expander} if $|\Gamma(S)| \ge A\cdot |S|$ for
all $S \subseteq [N]$ with $|S| \le K$.  $\Gamma(x,y)$ is 
the $y$th neighbor of the left vertex $x$.
\end{definition}

\begin{definition}
A probability distribution $\mathbf{X}$ on $\{0,1\}^n$ is called a
{\em $k$-source} if $\Pr[X = x] \le 2^{-k}$ for all $x\in\{0,1\}^n$.
We interchangeably use ``$\mathbf{X}$ is a
$k$-source'' and ``$\mathbf{X}$ has min-entropy $k$''.
\end{definition}

Henceforth we let $U_n$ denote the uniform distribution on
$\{0,1\}^n$.

\begin{definition}
A function $C:\{0,1\}^n\times\{0,1\}^d \rightarrow \{0,1\}^m$ is
called a $k\rightarrow_{\eps} k'$ {\em condenser} if $C(\mathbf{X},
\U_d)$ is $\eps$-close in statistical distance to some distribution of
min-entropy at least $k'$ whenever
$\mathbf{X}$ is a $k$-source.  A condenser is called {\em lossless} if
$k' = k + d$.  The statistical distance of two
probability distributions is defined to be half their $L_1$ distance.
\end{definition}

\begin{definition}
A function $E:\{0,1\}^n\times\{0,1\}^d \rightarrow \{0,1\}^m$ is
called a $(k,\eps)$ {\em extractor} if  $E(\mathbf{X}, \U_d)$ is
$\eps$-close in statistical distance to $\U_m$ whenever $\mathbf{X}$
is a $k$-source.
\end{definition}

In \Section{linear-guv} we will write write the expansion of graphs we
consider as $(1-\eps)D$, where $\eps>0$ is some parameter.
All logarithms below are base-$2$ unless otherwise stated.

\subsubsection{The GUV Extractor in Linear Space}\SectionName{linear-guv}
For a given positive integer $h$ and prime power $q$, and for a
degree-$n$ irreducible polynomial $E$ over $\mathbb{F}_q$ and positive
integer $m$, Guruswami et al.\ \cite{GUV07} consider the bipartite graph
with neighbor function
$\Gamma:\mathbb{F}_q^n\times\mathbb{F}_q\rightarrow
\mathbb{F}_q^{m+1}$ defined by
\begin{equation}\EquationName{pv-gamma}
\Gamma(f, y) = [y, f(y), (f^h\mod E)(y), (f^{h^2}\mod E)(y),\ldots,
(f^{h^{m-1}} \mod E)(y)]
\end{equation}
where $f\in\mathbb{F}_q^n$ is interpreted as a polynomial of
degree at
most $n-1$ over $\mathbb{F}_q$.  In particular, the $y$th neighbor of
$f$ in the expander is the $y$th symbol of the encoding of $f$ under
the Parvaresh-Vardy code \cite{PV05}.  The authors of \cite{GUV07}
then prove the following theorem.

\begin{theorem}[Theorem 3.3 of \cite{GUV07}]\TheoremName{gen-expander}
The bipartite graph
$\Gamma:\mathbb{F}_q^n\times\mathbb{F}_q\rightarrow
\mathbb{F}_q^{m+1}$ defined as in \Equation{pv-gamma} is a $(\le
K_{\hbox{max}}, A)$ expander with $K_{\hbox{max}} = h^m$ and $A = q -
(n-1)(h-1)m$.\afterproof
\end{theorem}

For positive integers $N$,
$K_{\hbox{max}} \le N$, and for any $\eps > 0$,
and all $\alpha \in (0,\log x/\log\log x)$ with $x = (\log N)(\log
K_{\hbox{max}})$, \cite{GUV07} then apply \Theorem{gen-expander} to analyze
the quality of the expander obtained using the setting of
parameters in \Figure{exp-params}.  For our purposes though, we are
only concerned with $0<\alpha\le 1/2$, $\alpha = \Omega(1)$, and will
thus present bounds assuming $\alpha$ in this range.  We also assume
$N,K_{\hbox{max}} \ge 2$.

\begin{center}
\begin{figure}[!!h]
\begin{itemize}
\item $n = \log N$
\item $k = \log K_{\hbox{max}}$
\item $h = \ceil{(2nk/\eps)^{1/\alpha}}$
\item $m = \ceil{(\log K_{\hbox{max}})/(\log h)}$
\item $q$ is the unique power of $2$ in $(h^{1+\alpha}/2,h^{1+\alpha}]$
\end{itemize}
\caption{Setting of parameters in the GUV expander (see proof of
  Theorem 3.5 in \cite{GUV07}).}\FigureName{exp-params}
\end{figure}
\end{center}

\begin{theorem}[Theorem 3.5 of \cite{GUV07}]\TheoremName{expander}
The graph with parameters as stated in \Figure{exp-params} yields a
$(\le K_{\hbox{max}}, (1-\eps)D)$
expander with $N$ left vertices, left-degree $D =
O(((\log N)(\log K_{\hbox{max}})/\eps)^{1+1/\alpha})$, and $M \le
D^2\cdot K_{\hbox{max}}^{1+\alpha}$ right vertices.  Furthermore, the
neighbor function $\Gamma(f, y)$ can be computed in time
$\log^{O(1)}(ND)$, and $D$ and $M$ are each powers of
$2$.\afterproof
\end{theorem}

While the setting of parameters in \Figure{exp-params} yields an
expander whose neighbor function is time-efficient, for our purposes
we need a neighbor function that is both time-efficient and
space-efficient.  To accomplish this goal, we use the following
setting of parameters instead. Throughout this section, we borrow much
of the notation of \cite{GUV07} for ease of noting differences in the
two implementations.

\begin{center}
\begin{figure}[!!h]
\begin{itemize}
\item $n$ chosen in $(\log(N)/\log(q),3\log(N)/\log(q)]$ so that
  $n=3^{\ell}$ for some $\ell\in\mathbb{N}$
\item $k = \log K_{\hbox{max}}$
\item $z = 3\log(N)k/\eps$
\item $\alpha'\le \alpha$ is chosen as large as possible so that
  $(z^{1+1/\alpha'})/2$ is of the form $7^{2^{\ell}}$ for
  some $\ell \in\mathbb{N}$
\item $h_0 = z^{1/\alpha'}$
\item $h = \ceil{h_0}$
\item $q = (h_0^{1+\alpha'})/2$
\item $m = \ceil{(\log K_{\hbox{max}})/(\log h)}$
\end{itemize}
\caption{New setting of parameters for the GUV
  expander}\FigureName{exp-params-opt}
\end{figure}
\end{center}

\begin{theorem}\TheoremName{expander-space}
The graph with parameters as stated in \Figure{exp-params-opt} yields
a $(\le K_{\hbox{max}}, (1-\eps)D)$
expander with $N$ left vertices, left-degree $D =
O(((\log N)(\log K_{\hbox{max}})/\eps)^{1+3/\alpha})$, and $M \le
D^2\cdot K_{\hbox{max}}^{1+\alpha}$ right vertices.  Furthermore, the
neighbor function $\Gamma(f, y)$ can be computed in time
$\log^{O(1)}(ND)$ and space $O(\log(ND))$.
\end{theorem}
\begin{proof}
The proof is very similar to that of Theorem 3.5 of \cite{GUV07},
but taking the new parameters into account.  First we show
$\alpha' \ge \alpha/3$.  Note there is always an integer of the form
$7^{2^{\ell}}$ in $[t,t^2]$ whenever $t\ge 7$.  Since $\alpha\le 1/2$
and $z\ge 3$, we have
$$(z^{1+3/\alpha})/2 \ge (z^{2(1+1/\alpha)+1})/2
\ge z^{2(1+1/\alpha)}$$
Setting $t = z^{(1+1/\alpha)}$, we have $t\ge 3^3 > 7$, implying the
existence of an integer of the form $7^{2^{\ell}}$ in
$[(z^{1+1/\alpha})/2, (z^{1+1/(\alpha/3)})/2]$ so that $\alpha' \ge
\alpha/3$.

The number of left vertices of $\Gamma$ is $q^n \ge N$.  By choice of
$m$, $h^{m-1}\le K_{\hbox{max}} \le h^m$.
Thus, the number of right vertices $M$ satisfies
$$ M = q^{m+1} \le q^2h^{(1+\alpha')(m-1)} \le q^2h^{(1+\alpha)(m-1)}
\le q^2K^{1+\alpha}$$

The left-degree is
$$D = q < h^{1+\alpha'} \le (h_0+1)^{1+\alpha'} =
O((3\log(N)k/\eps)^{1 + 1/\alpha'}) = O(((\log N)(\log
K_{\hbox{max}})/\eps)^{1+3/\alpha})$$
with the penultimate equality following since $\alpha = O(1)$.

The expansion is $A = q - (n-1)(h-1)m \ge q-nhk$.  As in \cite{GUV07},
we now show $nhk
\le \eps q$ so that $q-nhk\ge q-\eps q = (1-\eps)D$.  Since
$h^{\alpha'} \ge 3\log(N)k/\eps \ge 3nk/\eps$, we have $nhk \le
(\eps/3) h^{1+\alpha'}\le \eps q$.  The final inequality holds since,
by the fact that $\alpha'\le \alpha \le 1/2$ and $h_0\ge z^2\ge 9$,
$$q=\frac{h_0^{1+\alpha'}}{2} \ge \frac{2((h_0+1)^{1+\alpha'})/3}{2} \ge
\frac{\ceil{h_0}^{1+\alpha'}}{3} = \frac{h^{1+\alpha'}}{3}$$

Calculating $\Gamma(f,y)$ requires performing arithmetic
over the finite field $\mathbb{F}_q$, which can be done by multiplying
polynomials in $\mathbb{F}_7[x]$ of degree at most $(\log_7 q) - 1$
modulo an irreducible polynomial $E'$ of degree $\log_7 q$.
By choice of $q$, $E'$ can be taken from family (2) of
\Theorem{explicit-irr}.  Also, as stated in \Equation{pv-gamma}, we
must take powers of $f$ modulo an irreducible $E$ of degree $n$.  By
choice of $n$, the polynomial $E$ can be taken from
family (3) of \Theorem{explicit-irr}.  The irreducibility of $E$ over
$\mathbb{F}_q[x]$ follows from \Theorem{lidl} since
$\hbox{gcd}(2^{\ell},3^{\ell'}) = 1$ for any $\ell,\ell'$.

The time complexity is immediate.  For space, in calculating
$\Gamma(f,y)$ for $k=0,\ldots,m-1$ we must calculate $f_k =
f^{h^{k}}\mod E$ then evaluate $f_k(y) = q$.  We have $f_k =
f_{k-1}^h\mod E$, which we can calculate time-efficiently in
$O(\ceil{n}\log(q)) = O(\log N + \log q)$ space by iterative
sucessive squaring.  Evaluating $f_k(y)$ takes an additional $O(\log
q)$ space.  In the end, we must perform $m+1$ such evaluations, taking
a total of $O(m\log q) = O(\log
M)$ space.  The total space is thus $O(\log N + \log D +
(1+\alpha)\log K_{\hbox{max}}) = O(\log(DN))$ since $K_{\hbox{max}}
\le N$ and $\alpha = O(1)$.
\end{proof}

Given their expander construction, the authors of \cite{GUV07} then
use an argument of Ta-Shma et al. \cite{TUZ07} that for positive
integers $n,m,d$ and
for $\eps\in (0, 1)$ and $k\in [0, n]$, a $(\le \ceil{2^k},
(1-\eps)\cdot 2^d)$ expander yields a $k\rightarrow_{\eps} k + d$
condenser.  Specifically, as argued in \cite{TUZ07}, the
constructed expander {\em is} a
condenser, where the input string is treated as a left vertex of
an expander with left-degree $2^d$, and the output string is the
index of the right-hand side vertex obtained by following the random
edge corresponding to the seed.  GUV could immediately apply this
connection to obtain a condenser since their $M,D$ of
\Theorem{expander} were powers of $2$. In \Theorem{expander-space}
however, $D$,$M$ are not powers of $2$ (they are powers of $7$).
Dealing with $M$ not being a
power of $2$ is simple: one can add dummy vertices to the right hand
side of the expander to make $M$ a power of $2$, at most doubling $M$
in the process.  The problem with $D$ not being a power of $2$ though
is that
a seed $s$ of length $d = \ceil{\log D}$ does not yield a uniformly random
neighbor if one interprets $s$ modulo $D$.  To deal with this issue,
if we desire a condenser whose output has statistical distance $\eps$
from a $k'$-source, we increase the seed length to $d = \ceil{\log D} +
\ceil{\log(1/\eps)} + 1$.  Now, interpreting the seed as a number in a
range of size at least $(2/\eps)D$, the seed modulo $D$ does yield a
random neighbor conditioned on the good event that the seed is not
larger than $2^{\floor{\log(2D/\eps)}}$, which happens with
probability at least $1-\eps/2$.  Statistical distance $\eps$ can thus
be achieved as long as the expander has expansion at least
$(1-\eps/2)D$.  This gives the following theorem.

\begin{theorem}[Based on Theorem 4.3 of \cite{GUV07}]\TheoremName{condenser}
For every positive integer $n$, and every $k_{\hbox{max}} \le n$,
$\eps > 0$, and $0\le \alpha\le 1/2$, $\alpha = \Omega(1)$, there is a
function
$C : \{0, 1\}^n\times \{0,1\}^d \rightarrow \{0,1\}^m$ with
$d = (1 + 1/\alpha)\cdot (\log n + \log k_{\hbox{max}})
+ O(\log(1/\eps))$ and $m \le 2d + (1 + \alpha)k_{\hbox{max}} + 1$ such that for all
$k \le k_{\hbox{max}}$, $C$ is a $k\rightarrow_{\eps} k + d$ lossless
condenser.  Furthermore, for any $x\in \{0,1\}^n$ and $s\in\{0,1\}^d$,
$C(x,s)$ can be computed in $O(n + \log(1/\eps))$ space and
$\hbox{poly}(n\log(1/\eps))$ time.
\end{theorem}
\begin{proof}
The proof of the theorem, except for the space upper bound, can be
found as Theorem 4.3 of \cite{GUV07}.  The space requirement follows
follows from \Theorem{expander-space}.
\end{proof}

In the construction of one of their extractors (the extractor we will
be concerned with), \cite{GUV07} uses the following extractor of
Impagliazzo, Levin, and Luby \cite{ILL89} as a subroutine, based on
the leftover hash lemma.

\begin{theorem}[Based on \cite{ILL89}]\TheoremName{hash-extractor}
For all integers $n = 2\cdot 3^{\ell}$, $k\le n$, with $\ell\ge 0$ an
integer,
and for all $\eps > 0$, there is a
$(k,\eps)$ extractor $E : \{0,1\}^n\times\{0,1\}^d \rightarrow
\{0,1\}^m$ with $d = n$ and $m \ge k + d - 2\log(1/\eps)$ such that
for all inputs $x,y$, $E(x,y)$ can be computed simultaneously in space
$O(n)$ and time $n^{O(1)}$.
\end{theorem}
\begin{proof}
The proof is sketched in \cite{GUV07} (and given fully in
\cite{ILL89}) except for the analysis of space complexity.  We review
the scheme so that we may prove the space bound.  Elements of
$\{0,1\}^n$ are treated as elements of $\mathbb{F}_{2^n}$, and $E(x,y)
= (y,xy|_m)$, where $xy|_m$ is the first $\ceil{k + d -
  2\log(1/\eps)}$ bits of the product $xy$ over $\mathbb{F}_{2^n}$.
The time and space complexity are thus dictated by the complexity of
multiplying two elements of $\mathbb{F}_{2^n}$ and remaindering modulo
a reducible $E$ of degree polynomial of degree $n$.  By the form of
$n$, we can take $E$ to be from family (1) in \Theorem{explicit-irr}.
The claim then follows by \Fact{obvious-efficiency}.
\end{proof}

The authors of \cite{GUV07} then give the following extractor
construction.

\begin{lemma}[based on Lemma 4.11 of \cite{GUV07}]\LemmaName{kt-extractor}
For every integer $t \ge 1$ and all positive integers $n\ge k$ and all
$\eps > 0$, there is a $(k,\eps)$ extractor $E:\{0,1\}^n\times
\{0,1\}^d\rightarrow \{0,1\}^m$ that can be computed in
$\hbox{poly}(n\log(1/\eps))$ time and $O(n + \log(1/\eps))$ space with $m =
\ceil{k/2}$ and $d\le k/t + O(\log(n/\eps))$.
\end{lemma}
\begin{proof}
The analysis of running time and proofs of correctness and output
length are identical to \cite{GUV07}, so we focus on space
analysis.  We
now review the algorithm of \cite{GUV07} for computing $E(x,y)$.
\begin{enumerate}
\item Round $t$ to a positive integer and set $\eps_0 = \eps/(4t+1)$.
\item Apply the condenser of \Theorem{condenser} with error $\eps_0$,
  $\alpha = 1/(6t)$, min-entropy $k$, and seed length $d' =
  O(\log(n/\eps))$ to $x$, using the first $d'$ bits of $y$.  The
  output $x'$ of the condenser will be of length at most $n' = (1+\alpha)k
  + O(\log(n/\eps))$.
\item Partition $x'$ into $2t$ blocks $x_1',\ldots,x_{2t}'$ of size
  $n'' =
  \floor{n'/(2t)}$ or $n'' + 1$ and set $k'' = k/(3t) -
  O(\log(n/\eps))$.
\item Let $E''$ be the extractor of \Theorem{hash-extractor} for
  min-entropy $k''$ with input length $n'' + 1$, seed length $d'' =
  k/t + O(\log(n/\eps))$, and error parameter $\eps_0$.  For this
  setting of parameters, the output length of $E''$ will be $m'' \ge
  \max\{d'',
  k'' + d'' - 2\log(1/\eps_0)\}$.  Now output $(z_1,\ldots,z_{2t})$
  where $y_{2t}'$ is the last $d - d' = d''$ bits of $y$, and for $i =
  2t,\ldots,1$, $(y_{i-1}', z_i)$
  is defined inductively to be a partition of $E''(x_i', y_i')$ into a
  $d''$-bit prefix and $(m'' - d'')$-bit suffix.
\end{enumerate}

We now analyze the space complexity of computing this extractor.
First we note $d = k/t + O(\log(n/\eps)) = O(n + \log(1/\eps))$.  Step
2 requires $O(d' + (1+\alpha) n)$ space,
which is $O(n + \log(1/\eps))$.
To apply $E''$, by \Theorem{hash-extractor} we need $n''+1$ to be of
the form $2\cdot 3^{\ell}$;
for now assume this, and we will fix this later.
Each evaluation of $E''$ in Step 4 takes space $O(n'') = O(n'/t) =
O((1+\alpha) k + \log(n/\eps)) = O(n + \log(1/\eps))$.  We also have
to maintain the
$z_i$ as we generate them, but we can stop
the recursive applications of $E''$ in Step 4 once we have extracted
$\ceil{k/2}$ bits.  Also, there are only $2t = O(1)$
levels of recursion in Step 4, so an implementation can keep track of
the current level of recursion with only $O(1)$ bits of bookkeeping.
The total seed length is $d'' + d' = k/t + O(\log(n/\eps))$.

Now, to fix the fact that $n''+1$ might not be of the form $2\cdot
3^{\ell}$, we increase $n''$ so that this does hold.  Doing so
increases $n''$ by at most a factor of $3$.  Since $d'' = n''$, we
increase the seed to be of length $3k/t + O(\log(n/\eps))$, but this
can be remedied by applying the above construction for $t' = 3t$.
\end{proof}

We now come to the final theorem we will use from \cite{GUV07}.

\begin{theorem}[Theorem 4.17 of \cite{GUV07}]\TheoremName{final-extractor}
For all positive integers $n>0, k\le n$ and for all $\eps > 0$, there
is a $(k,\eps)$ extractor $E:\{0,1\}^n\times
\{0,1\}^d\rightarrow \{0,1\}^m$ with $d = O(\log n + \log(1/\eps))$
and $m\ge k/2$ where $E(x,y)$ can be computed in time
$\hbox{poly}(n\log(1/\eps))$ and space $O(n + \log(1/\eps))$ for all
$x,y$.
\end{theorem}
\begin{proof}
The construction here is completely unchanged from \cite{GUV07}.
We only analyze the space complexity, as the time complexity and
extractor parameters were analyzed in \cite{GUV07}.  To perform this
analysis, we present the details of the construction.

Define $\eps_0 = \eps/\hbox{poly}(n)$.  For any integer $k$, define
$i(k)$ to be the smallest integer $i$ such that $k\le 2^i\cdot 8d$
with $d = c\log(n/\eps_0)$ for some large constant $c$.  For every
$k\in [0, n]$, GUV define an extractor $E_k$ recursively.  In their
base case, $i(k) = 0$ so that $k \le 8d$.  Here they apply
\Lemma{kt-extractor} with $t = 9$.

For $i(k) > 0$, $E_k(x,y)$ is evaluated as follows.
\begin{enumerate}
\item Apply \Theorem{condenser} to $x$ with a seed length of
  $O(\log(n/\eps_0))$ to obtain a string $x'$ of
  length $(9/8)k + O(\log(n/\eps_0))$.
\item Divide $x'$ into two equal-sized halves $x'_1$, $x'_2$. Set $k'
  = k/2 - k/8 - O(\log(n/\eps_0))$, for which $k' \ge 2d$ by setting
  $c$ sufficienty large.  Set $E' = E_{k'}$, which has seed
  length $d_1 = d$, and obtain a
  $(2d,\eps_0)$ extractor
  $E''$ from \Lemma{kt-extractor} with $t=16$, output length $d$,
  and seed length $d_2 = d/8 + O(\log(n/\eps_0))$.
\item Apply $E''$ to $x'_2$ to to yield an output $y$.  Output
  $E'(x'_1, y)$, which has length at least $k/6$.
\end{enumerate}
The total seed length is $d/8 + O(\log(n/\eps_0))$.  To yield
$k/2$ bits of output and not just $k/6$, repeat Steps 1
through 3 above but with $k$ replaced by $k_2 = 5k/6 - 1$.  Then
repeat again with $k_3 = 5k_2/6 - 1$ then $k_4 = 5k_3/6 - 1$.  The
total number of output bits is then $(1 - (5/6)^4)k - O(1) \ge
k/2$, and the seed length has increased by a factor of $4$, but is
still at most $d$.

Now we analyze space complexity.  Steps 1 through 3 above are
performed four times at any recursive level, so we have a recursion
tree with branching factor four and height $O(\log k)$.  At a level
of recursion where we handle some min-entropy $k''$, the input at that
level is some $x''$ of length $\Theta(k'')$ along with a seed of
length $d$ (except for the topmost level which has input $x$ of length
$n$).  For all levels but the bottommost in the recursion tree, we
have
$k'' > d$ so that the total space needed to store all inputs at all
levels when performing the computation depth-first on the recursion
tree is bounded by a geometric series with largest value $\Theta(k +
d)$.  At each non-leaf node of recursion we run the extractor from
\Lemma{kt-extractor} four times, each with input length $\Theta(k'')$
and seed length $d < k''$, using space $O(k'')$.  We therefore have
that no level uses space more than
$O(k'')$ to perform its computations.  The total space to calculate
$E_k$ is thus $O(n + k + d) = O(n + \log(1/\eps))$.
\end{proof}

\subsubsection{Applying GUV to Armoni's PRG}
We begin with a formal definition of a pseudorandom generator.

\begin{definition}\DefinitionName{prg}
A function $G:\{0,1\}^l\rightarrow \{0,1\}^R$ is a {\em
  $\gamma$-pseudo-random generator} ($\gamma$-PRG) for space $S$ with
$R$ random bits if any space-$S$
machine $M$ with one-way access to $R$ random bits is {\em $\gamma$-fooled} by
$G(\U_l)$.  That is, if we let $M(x,y)$ denote the final state of the
machine $M$ on an input $x$ and $R$-bit string $y$, $||M(x,\U_R) -
M(x,G(\U_l))|| \le \gamma$ for all inputs $x$, where
$||\mathbf{A}-\mathbf{B}||$ denotes the statistical difference between
two distributions $\mathbf{A},\mathbf{B}$.
\end{definition}

Armoni defines a $\gamma$-PRG slightly differently.  Namely, in his
definition the machine $M$ outputs a binary answer (``accept'' or
``reject''), and he only requires that the distribution of the
decision made by $M$ changes by at most $\gamma$ in statistical
distance for any input.  However, in fact the PRG construction he
gives actually satisfies \Definition{prg}.  This is because he models
the machine's execution on an input $x$ by a branching
program with $R$ layers and width $2^S$ in each layer. One should interpret
nodes in the branching program as states of the algorithm, where each
node in the $i$th layer, $i<R$, has out-degree $2$ into the $(i+1)$st
layer with the edges labeled $0$ and $1$ corresponding to the $i$th
bit of randomness. Armoni then actually provides a PRG
which $\gamma$-fools branching programs with respect to the
distribution of the final ending node, i.e., the final state
of the algorithm.

Henceforth we summarize the PRG construction of Armoni \cite{Armoni98}
to illustrate that the space-efficient implementation of the GUV
extractor described in \Theorem{final-extractor} gives a PRG using
seed length $O((S/(\log S - \log\log R + O(1)))\log R)$ to produce $R$
pseudorandom bits for any $R = 2^{O(S)}$ which fool space-$S$ machines
with one-way access to their randomness.  We assume $R\ge S$, since
otherwise the machine could afford to store all random bits it uses.

To use notation similar to that of Armoni, for an extractor
$E:\{0,1\}^k\times \{0,1\}^t\rightarrow \{0,1\}^r$, define
$\hat{G}_{E,n}:\{0,1\}^{k+nt}\rightarrow \{0,1\}^{nr}$ by
$$\hat{G}_{E,n}(x,y_1,y_2,\ldots,y_n) = E(x,y_1)\cdots E(x,y_n)$$
where $x\in\{0,1\}^k$ and $y_i\in\{0,1\}^t$.  To obtain a
$\gamma$-PRG, Armoni recursively defines functions $G_i
\{0,1\}^k\times \{0,1\}^{(i-1)k'}\times\{0,1\}^{n_{i-1}t}\rightarrow
\{0,1\}^R$ as follows.
\begin{enumerate}
\item $G_1(x_1,y_1,\ldots,y_n) = \hat{G}_{E,n}(x_1,y_1,\ldots,y_n)$
\item $G_i(x_1,\ldots,x_i,y_1,\ldots,y_n) =
  G_{i-1}(x_1,\ldots,x_{i-1},\hat{G}_{E',n}(x_i,y_1,\ldots,y_{n_{i-1}}))$
\end{enumerate}
where $k = O(S)$, $k' = O(S + \log(1/\gamma))$, $t =
O(\log(R/\gamma))$, and $n_i = n_{i-1}/\Theta(S/(\log R +
\log(1/\gamma)))$ for $i>0$ with $n_0 = R$.  The extractor $E$ has
input length $k$ and seed length $t$, while $E'$ has input length $k'$
and seed length $t$.  The string $x_1$ is in $\{0,1\}^k$, while
$x_2,\ldots,x_i\in \{0,1\}^{k'}$ and $y_i\in\{0,1\}^t$.  
The final PRG is defined as $G = G_h$, with $h =
\Theta(\log(R)/(\max\{1,\log(S) - \log\log(R/\gamma))\}))$. For each
$i<h$, the output of $\hat{G}_{E',n}$ is split into equal-size blocks
of size $t$ to obtain the $y_1,\ldots,y_{n_i}$ for $G_{i+1}$.

In Armoni's proof of correctness of his PRG, he needs the following
type of extractor.  For every integer $\ell$ and every $\eps>0$, he
requires a $(\ell/2,\eps)$ extractor
$E:\{0,1\}^{\ell}\times\{0,1\}^d\rightarrow\{0,1\}^{\ell/4}$ with $d =
\Theta(\log(\ell/\eps))$.  
The extractors $E$, $E'$ above must be taken to
have these parameters with $\ell=k$ and $\ell=k'$.  By
\Theorem{final-extractor}, we know such $E$, $E'$ can be chosen that
can be evaluated in space $O(k + t)$ and $O(k' + t)$,
respectively.\footnote{We note Armoni defines extractors to
  be ``strong'', i.e. the seed appears at the end of the output. It is
  known that the GUV extractor can be easily made strong with no
  increase in complexity (see Remark 4.22 of \cite{GUV07}).}
We now analyze the space-complexity of computing any single bit in the
output of $G$.  We must store
a seed of length $O(k + k'(h-1) +t)$, which is
$O(((S+\log(1/\gamma))/\max\{1,\log S -
\log\log(R/\gamma)\})\log R)$ (see Theorem 2 of \cite{Armoni98} for a
detailed calculation).  To calculate a single output bit, in a
recursive implementation there are $h = O(\log R) = O(S)$ levels of
recursion, and in each we must evaluate either $E$ or $E'$ on some
$y_i$, split the output of that evaluation into blocks, then recurse.
At a level $i$ of recursion we need to know the seed $y_i$ we have
recursed on, as well as which output bit $b_i$ we will want in
$G_i(x_1,x_2,\ldots,x_i,y)$.  The value $b_i$ fits
into at most $\log R$ bits, and the length of $y_i$ is $t =
O(\log(R/\gamma))$.  Note though that once we have calculated
$y_{i-1}$ and $b_{i-1}$ for our recursive step to the $(i-1)$st level,
we no longer need to know $y_i$ and $b_i$.  Thus, the $y_i$ and $b_i$
can be kept in a global register, taking a total of $t =
O(\log(R/\gamma)) = O(S + \log(1/\gamma))$ bits throughout the entire
recursion.  At each level of recursion we must perform one evaluation
of an extractor, which takes space $O(k' + t) = O(S +
\log(1/\gamma))$.  We thus have the following theorem, which extends
Corollary 1 of \cite{Armoni98} by working for the full range of $R$,
as opposed to just $R < 2^{S^{1-\delta}}$ for some $\delta>0$.

\begin{theorem}\TheoremName{final-armoni}
For any $\gamma>0$ and integers $S\ge 1, R = 2^{O(S)}$, there is a
$\gamma$-PRG
stretching $O(\frac{S+\log(1/\gamma)}{\max\{1,\log S -
\log\log(R/\gamma)\}}\log R)$ bits of seed to $R$ pseudorandom
bits $\gamma$-fooling space-$S$ machines such that any of the $R$
output bits
can be computed in space $O(S + \log(1/\gamma))$ and time
$\hbox{poly}(S\log(1/\gamma))$.\afterproof
\end{theorem}

We note that Indyk's algorithm is designed to succeed with constant
probability (say, $2/3$), so in the application of
\Theorem{final-armoni} to his algorithm, $\gamma$ is a constant.


\subsection{A balls and bins process}\SectionName{balls-and-bins}
Consider the following random process which arises in the analysis
of both our $F_0$ and $L_0$ algorithms.  We throw a set of $A$ ``good''
balls and $B$ ``bad'' balls into $K$ bins at random.  In the analysis
of our $L_0$ algorithm, we will be concerned with the special case
$B=0$, whereas the $F_0$ algorithm analysis requires understanding the
more general random process.
We let $X_i$ denote
the random variable indicating that at least one good ball, and no bad
balls, landed in bin $i$, and we let $X = \sum_{i=1}^K X_i$.
We now prove a few lemmas.  

\begin{lemma}\LemmaName{exactgoodbadballs}
$$\E[X] = K\left(1 - \left(1 - \frac{1}{K}\right)^A\right) \left(1 -
  \frac{1}{K}\right)^B$$
\begin{center}
and
\end{center}
\begin{align*}
\Var[X] = K\left(1 - \left(1 - \frac 1K\right)^A\right) \left(1 -
  \frac 1K\right)^B & + K(K-1)\left(1 - \frac 2K\right)^B \left(1 -
  2\left(1 - \frac 1K\right)^A + \left(1 - \frac 2K\right)^A\right) \\
& - K^2\left(1 - 2\left(1 - \frac 1K\right)^A + \left(1 - \frac
    1K\right)^{2A}\right) \left(1 - \frac 1K\right)^{2B}
\end{align*}
\end{lemma}
\begin{proof}
The computation for $\E[X]$ follows by linearity of expectation.

For $\Var[X]$, we have
$$\Var[X] = \E[X^2] - \E^2[X] = \sum_i \E[X_i^2] +
2\sum_{i<j}\E[X_iX_j] - \E^2[X]$$
We have $\E[X_i^2] = \E[X_i]$, so the first sum is simply $\E[X]$.  We
now calculate $\E[X_iX_j]$ for $i\neq j$.  Let $Y_i$ indicate that
at least one good ball landed in bin $i$, and let $Z_i$ indicate that
at least one bad ball landed in bin $i$.  Then,
\begin{eqnarray*}
\E[X_iX_j] &=& \Pr[Y_i\wedge Y_j\wedge \bar{Z_i} \wedge \bar{Z_j}] \\
&=& \Pr[\bar{Z_i}\wedge \bar{Z_j}]\cdot \Pr[Y_i\wedge
Y_j|\bar{Z_i}\wedge \bar{Z_j}]\\
&=& \Pr[\bar{Z_i}\wedge \bar{Z_j}]\cdot \Pr[Y_i\wedge Y_j]\\
&=& \left(1 - \frac 2K\right)^B\cdot \left(1 - \Pr[\bar{Y_i}\wedge
  \bar{Y_j}] - \Pr[Y_i \wedge \bar{Y_j}] - \Pr[\bar{Y_i} \wedge Y_j]
\right)\\
&=& \left(1 - \frac 2K\right)^B\cdot \left(1 - \Pr[\bar{Y_i}\wedge
  \bar{Y_j}] - 2\cdot \Pr[Y_i \wedge \bar{Y_j}] \right)\\
&=& \left(1 - \frac 2K\right)^B\cdot \left(1 - \Pr[\bar{Y_i}\wedge
  \bar{Y_j}] - 2\cdot \Pr[\bar{Y_j}]\cdot\Pr[Y_i|\bar{Y_j}] \right)\\
&=& \left(1 - \frac 2K\right)^B\cdot \left(1 - \left(1 - \frac
    2K\right)^A - 2\left(1 - \frac 1K\right)^A\left(1 - \left(1 -
      \frac 1{K-1}\right)^A\right) \right)\\
&=& \left(1 - \frac 2K\right)^B\cdot \left(1 - \left(1 - \frac
    2K\right)^A - 2\left(1 - \frac 1K\right)^A + 2\left(1 -
      \frac 2K\right)^A\right) \\
&=& \left(1 - \frac 2K\right)^B\cdot \left(1 - 2\left(1 - \frac
    1K\right)^A + \left(1 - \frac 2K\right)^A \right)\\
\end{eqnarray*}

The variance calculation then follows by noting $2\sum_{i<j}\E[X_iX_j]
= K(K-1)\E[X_1X_2]$ then expanding out $\E[X] + K(K-1)\E[X_1X_2] -
\E^2[X]$.
\end{proof}

\begin{lemma}\LemmaName{bigballexpectation}
If $A\ge K/160$ and $A,B\le K/2$, 
then $\E[X] \ge K/500$.
\end{lemma}
\begin{proof}
Applying \Lemma{exactgoodbadballs},
$$\E[X] \ge K\left(1 - \frac BK\right)\left(1 -
  \left( 1 - \frac AK + \frac{A^2}{2K^2}\right)\right)\ge \frac
K2\cdot\frac AK\left(1 - \frac A{2K}\right)\ge \frac K{320}\cdot \frac
34 \ge \frac K{500}$$
\end{proof}

In the next lemma, we use the following inequalities.

\begin{lemma}[Motwani and Raghavan {\cite[Proposition
    B.3]{MR95}}]\LemmaName{mr95}
For all $t,n\in\mathbb{R}$ with $n\ge 1$ and $|t|\le n$,
$$ e^t\left(1 - \frac{t^2}{n}\right) \le \left(1 + \frac tn\right)^n
\le e^t  $$
\end{lemma}

\begin{lemma}\LemmaName{badballsvar}
If $A,B \le K/4$ then $\Var[X] \le 7K$.
\end{lemma}
\begin{proof}
Applying \Lemma{exactgoodbadballs} and \Lemma{mr95},
\begin{eqnarray*}
\Var[X] &\le& Ke^{-B/K} - Ke^{-(A+B)/K}\left(1-\frac
  1K\right)^{(A+B)/K} \\
&+& K(K-1)e^{-2B/K}\left(1 - 2e^{-A/K}\left(1-\frac
  1K \right)^{(A+2B)/K}  + e^{-2A/K} \right)\\
&-& K^2e^{-2B/K}\left( \left(1 - \frac 1K \right)^{2B/K} - 2e^{-A/K} +
  e^{-2A/K}\left(1-\frac
  1K \right)^{2(A+B)/K} \right)
\end{eqnarray*}
Now using the fact that $A,B\le K/4$ and combining like terms,
\begin{eqnarray*}
\Var[X] &\le& K\left(e^{-B/K} - e^{-(A+B)/K}\left(1-\frac
  1K\right) - e^{-2B/K} + 2e^{-(A+2B)/K} - e^{-2A/K}\right)\\
&&+\ K^2\Bigg( e^{-2B/K} - 2e^{-(A+2B)/K}\left(1 - \frac 1K\right) +
  e^{-2(A+B)/K} - e^{-2B/K}\left(1 - \frac 1K\right)\\ 
&&+\ 2e^{-(A+2B)/K}
  - e^{-2(A+B)/K}\left(1 - \frac 1K\right) \Bigg)\\
&=& K\left(e^{-B/K} - e^{-(A+B)/K}\left(1-\frac
  1K\right) - e^{-2B/K} + 2e^{-(A+2B)/K} - e^{-2A/K}\right)\\
&&+\ K\left( e^{-2B/K} +
  e^{-2(A+B)/K} + 2e^{-(A+2B)/K} \right)
\end{eqnarray*}

Each of the positive terms multiplying $K$ above is upper bounded by
either $1$ or $2$, and we have $\Var[X] \le 7K$.
\end{proof}

\begin{lemma}\LemmaName{yuck}
If $B=0$ and $100 \leq A \leq K/20$, then $\Var[X] <
4A^2/K$.
\end{lemma}
\begin{proof}
By \Lemma{exactgoodbadballs},
\begin{eqnarray*}
\Var[X] & = & K\left (K-1 \right
)\left(1-\frac 2K\right)^{A} + K\left(1-\frac 1K\right)^{A} -
K^2 \left(1-\frac 1K\right)^{2A}\\
& = & K^2 \left [\left(1-\frac 2K\right)^{A} - \left(1-\frac
    1K\right)^{2A}
\right ] + K\left [\left(1-\frac 1K\right)^{A} -
  \left(1- \frac 2K\right)^{A} \right ]\\
& = & K^2\left(1- \frac 2K\right)^{A} \left [1 - \left
    (\frac{1- \frac 2K + \frac 1{K^2}}{1- \frac 2K} \right )^{A}\right ] +
K \left [\left(1- \frac 1K\right)^{A} - \left (1 - \frac 2K\right)^{A}
\right ]\\
& = & K^2\left(1- \frac 2K\right)^{A} \left [1 - \left (1 +
    \frac{1}{K^2\left(1- \frac 2K\right)} \right )^{A} \right ] +
K\left [\left(1- \frac 1K\right)^{A} - \left(1- \frac 2K\right)^{A}
\right ]\\
& = & K^2\left(1- \frac 2K\right)^{A} \left [1 - \left (1 +
    \frac{A}{K^2\left(1- \frac 2K\right)} + E_1 \right ) \right ]+
K\left [\left (1 - \frac AK + E_2 \right ) - \left (1- \frac{2A}{K} +
    E_3 \right ) \right ],
\end{eqnarray*}
where $E_1, E_2,$ and $E_3$ are the sum of quadratic and higher terms
of the binomial expansions for $(1+/(K^2(1-2/K)))^{A}$,
$(1-1/K)^{A}$, and $(1-2/K)^{A}$, respectively. Continuing
the expansion, 
\begin{eqnarray*}
\Var[X] & = & -K^2\left(1-\frac 2K\right)^{A} \left
  (\frac{A}{K^2\left(1-\frac 2K\right)} +E_1 \right ) + A +
K(E_2-E_3)\\
& = & -A\left(1-\frac 2K\right)^{A-1} - K^2E_1\left(1-\frac
  2K\right)^{A} +
A + K(E_2-E_3)\\
& = & -A\left(1-\frac{2(A-1)}{K} + E_4\right) -
K^2E_1\left(1-\frac 2K\right)^{A} + A + K(E_2-E_3)\\
& = & -A + \frac{2A(A-1)}{K} - AE_4 -
K^2E_1\left(1-\frac 2K\right)^{A} + A + K(E_2-E_3)\\
& = & \frac{2A(A-1)}{K} - AE_4 - K^2E_1\left(1-\frac 2K\right)^{A}
+ K(E_2-E_3),\\
\end{eqnarray*}
where $E_4$ is the sum of quadratic and higher terms of the binomial
expansion of $(1-2/K)^{A-1}$. Since
$10 \leq A \leq K/20$, we have that $E_4$ is bounded by a
geometric series with starting value $(2/K)^2(A-1)^2/2 \leq
2A(A-1)/K^2 \leq (A-1)/(5K)$ and common ratio at most
$2(A-1)/K \leq 2A/K \leq 1/10$, and so $E_4 \leq
((A-1)/(5K))/(1-1/10) = 2(A-1)/(9K)$. Thus,
$-A E_4 \leq 2A(A-1)/(9K)$. 

Arguing similarly, we see that $E_1$ is at most
$(A^2)/(K^4(1-A/K^2)) \leq 2A^2/K^4$ for sufficiently
large $K$. It follows that 
$$ K^2E_1\left(1-\frac 2K\right)^{A} \leq K^2E_1 \leq
2\frac{A^2}{K^2} \leq \frac{A(A-1)}{9K},$$
for sufficiently large $K$. 

Finally, we look at $E_2-E_3$,
\begin{eqnarray*}
E_2 - E_3 & = & \left (\frac{\binom{A}{2}}{K^2} -
  \frac{\binom{A}{3}}{K^3} +
  \cdots \right ) - \left (\frac{4\binom{A}{2}}{K^2} -
  \frac{8\binom{A}{3}}{K^3}  + \cdots \right )\\
& = & - \frac{3\binom{A}{2}}{K^2} + \frac{7\binom{A}{3}}{K^3} - \cdots
\end{eqnarray*}
This series can be upper bounded by the series
$\sum_{i=2}^{\infty}\frac{(2^i-1)(A/K)^i}{i!}$, and  lower
bounded by the series
$-\sum_{i=2}^{\infty}\frac{(2^i-1)(A/K)^i}{i!}$. This series,
in absolute value, is just a geometric series with starting term
$3A^2/(2K^2)$ and common ratio at most $A/K \leq
1/20$. Thus, $|E_2-E_3| \leq \frac{20}{19} \cdot \frac{3A^2}{2K^2} =
\frac{30}{19} \cdot (A/K)^2$. It follows that
$|K(E_2-E_3)| \leq \frac{30}{19} \cdot A^2/K \leq
\frac{30}{19} \cdot \frac{100}{99} \cdot A(A-1)/K =
\frac{3000}{1881} \cdot A(A-1)/K$, since $A \geq 100$. 

Hence, 
$$|A E_4| + \left |K^2E_1\left(1-\frac 2K\right)^{A} \right | +
\left |K(E_2 - E_3)\right | \leq \left (\frac{2}{9} +
  \frac{1}{9} + \frac{3000}{1881} \right )A(A-1)/K <
1.93A(A-1)/K.$$
and thus $\Var[X] \le 3.93A^2/K$.
\end{proof}

\begin{lemma}\LemmaName{ind}
There exists some constant $\eps_0$ such that the following holds for
$\eps\le \eps_0$.
Let $\mathcal{H}$ be a family of $c\log(K/\eps)/\log\log(K/\eps)$-wise
independent hash functions mapping the $A+B$ good and bad balls into
$K$ bins for
some sufficiently large constant $c > 0$.  Suppose $A,B\le K/e$ and
$A \ge 1$, and
we choose a
random $h \in \mathcal{H}$ mapping balls to bins. For $i \in
[K]$, let $X_i'$ be an indicator variable which is $1$ if and
only if there
exists at least one good ball, and no bad balls, mapped to bin $i$ by
$h$. Let $X'=
\sum_{i=1}^{K}X_i'$.  Then for a sufficiently large constant
$c$, the following holds:

\begin{enumerate}
\item  $|\E[X']-\E[X]| \le \eps\E[X]$
\item $\Var[X']-\Var[X] \le \eps^2$
\end{enumerate}
\end{lemma}
\begin{proof}
Let $A_i$ be the random variable number counting the number of good
balls in bin $i$ when
picking $h$ from $\mathcal{H}$.  Let
$B_i$ be the
number of bad balls in bin $i$.  Define the function:

$$f_k(n) = \sum_{i=0}^k (-1)^i \binom{n}{i}$$

We note that $f_k(0)=1$, $f_k(n)=0$ for $1\le n \le k$ and
$|f_k(n)| \le
\binom{n}{k+1}$ otherwise.  Let $f(n) = 1$ if $n=0$ and $0$ otherwise.
We now approximate $X_i$ as $f_k(B_i)(1-f_k(A_i))$.  We note that this
value
is determined entirely by $2k$-independence of the bins the balls are
put into.  We note that this is also

\begin{eqnarray*}
&&\left(f(B_i) \pm O\left(\binom{B_i}{k+1}\right)\right)\left
  (1-f(A_i)\pm O\left(\binom{A_i}{k+1}\right)\right)\\
&&=\ X_i \pm
O\left(\binom{B_i}{k+1} + \binom{A_i}{k+1} +
\binom{A_i}{k+1}\binom{B_i}{k+1}\right)
\end{eqnarray*}

The same expression holds for the $X_i'$, and thus both $\E[X_i']$ and
$\E[X_i]$ are sandwiched inside an interval of size bounded by twice the
expected error.
To bound the expected error we can use $2(k+1)$-independence.  We have
that the
expected value of, say, $\binom{A_i}{k+1}$ is $\binom{A}{k+1}$ ways of
choosing $k+1$ of the good balls times the product of the
probabilities that each ball is in bin $i$.  This is

$$\binom{A}{k+1}K^{-(k+1)} \le \left(\frac{eA}{K(k+1)}\right)^{k+1}$$
and similarly for $\E[\binom{B_i}{k+1}]$.  Assuming that $A,B\le
K/e$, $|\E[X_i] -
\E[X_i']| \le \eps^2/K$ as long as $6(2(k+1))^{-(k+1)} \le \eps^2$, which
occurs for $k = c\log(K/\eps)/\log\log(K/\eps)$ for sufficiently large
constant $c$.
In this case $|\E[X] - \E[X']| \le \eps^2 \le
 \eps\E[X]$ for sufficiently small $\eps$ since $\E[X] = \Omega(1)$
 when $B\le K$ and $A\ge 1$.

We now analyze $\Var[X']$.  
We approximate $X_i X_j$ as
$f_k(B_i)f_k(B_j)(1-f_k(A_i))(1-f_k(A_j))$.  This is determined by
$4k$-independence of the balls and is equal to
\begin{eqnarray*}
&&\left(f(B_i) \pm O\left(\binom{B_i}{k+1}\right)\right)
\left(f(B_j) \pm O\left(\binom{B_j}{k+1}\right)\right)
\left(1-f(A_i)\pm O\left(\binom{A_i}{k+1}\right)\right)\\
&&\times\
\left(1-f(A_j)\pm O\left(\binom{A_j}{k+1}\right)\right) \\
&=& X_iX_j \pm
O\Bigg(\binom{A_i}{k+1} + \binom{A_j}{k+1} + \binom{B_i}{k+1} +
  \binom{B_j}{k+1} + \binom{A_i}{k+1}\binom{A_j}{k+1} +
  \binom{B_i}{k+1}\binom{B_j}{k+1}\\
&&+\ \binom{A_i}{k+1}\binom{B_i}{k+1} +
\binom{A_i}{k+1}\binom{B_j}{k+1}
  + \binom{A_j}{k+1}\binom{B_i}{k+1} +
  \binom{A_j}{k+1}\binom{B_j}{k+1}\\
&&+\
\binom{A_i}{k+1}\binom{A_j}{k+1}\binom{B_i}{k+1} +
\binom{A_i}{k+1}\binom{A_j}{k+1}\binom{B_j}{k+1} +
\binom{A_i}{k+1}\binom{B_i}{k+1}\binom{B_j}{k+1}\\
&&\ + \binom{A_j}{k+1}\binom{B_i}{k+1}\binom{B_j}{k+1} +
\binom{A_i}{k+1}\binom{A_j}{k+1}\binom{B_i}{k+1}\binom{B_j}{k+1}
\Bigg)
\end{eqnarray*}

We can now analyze the error using $4(k+1)$-wise independence.  The
expectation of each term in the error is calculated as before, except
for products of the form

$$\binom{A_i}{k+1}\binom{A_j}{k+1},$$
and similarly for $B_i,B_j$.
The expected value of this is

$$\binom{A}{k+1,k+1}K^{-2(k+1)} \le
\binom{A}{k+1}^2 K^{-2(k+1)} \le
\left(\frac{eA}{K(k+1)}\right)^{2(k+1)} .$$

Thus, again, if $A,B\le K/e$ and $k =
c'\log(K/\eps)/\log\log(K/\eps)$ for $c'$ sufficiently large, each
summand in the error above is bounded by $\eps^3/(32K^2)$, in which
case
$|\E[X_iX_j] - \E[X_iX_j]| \le \eps^3/K^2$. We can also make $c'$
sufficiently large so that $|\E[X] - \E[X']| \le \eps^3/K^2$.
Now, we have 
\begin{eqnarray*}
\Var[X'] - \Var[X] &\le& |(\E[X] - \E[X']) + 2\sum_{i<j}(\E[X_iX_j] -
\E[X_i'X_j']) - (\E^2[X] - \E^2[X'])|\\
&\le& |\E[X] - \E[X']| + K(K-1)\max_{i<j}|\E[X_iX_j] -
\E[X_i'X_j']| + |\E^2[X] - \E^2[X']|\\
&\le& \eps^3/K^2 + \eps^3 + \E^2[X](2\eps^3/K^2 + (\eps^3/K^2)^2)\\
&\le& 5\eps^3
\end{eqnarray*}
which is at most $\eps^2$ for $\eps$ sufficiently small.
\end{proof}

\begin{lemma}\LemmaName{ind-consequences}
There exists a constant $\eps_0$ such that the following holds.
Let $\mathcal{H}$, $X'$ be as in \Lemma{ind}, and also
assume $B=0$ and $100 \leq A \leq K/20$ with $K = 1/\eps^2$ and
$\eps\le \eps_0$.  Then
$$\Pr_{h\leftarrow\mathcal{H}}[|X' - \E[X]| \le 8\eps \E[X]] \ge
3/4$$.
\end{lemma}
\begin{proof}
Observe that 
\begin{eqnarray*}
\E[X] &\ge& (1/\eps^2)\left(1-\left(1-A\eps^2 +
\binom{A}{2}\eps^4\right)\right)\\
&=& (1/\eps^2)\left(A \eps^2 - \binom{A}{2}\eps^4\right)\\
&\ge& (39/40) A,
\end{eqnarray*}
since $A \le 1/(20\eps^2)$.

By \Lemma{ind} we have $\E[X'] \ge (1-\eps)\E[X] > (9/10)A$, and
additionally using \Lemma{yuck} we have that
$\Var[X'] \le \Var[X] + \eps^2 \le 5\eps^2A^2$.
Set $\eps' = 7\eps$. Applying Chebyshev's inequality,
\begin{eqnarray*}
\Pr[|X' - \E[X']| \ge (10/11)\eps' \E[X']] &\le& \Var[X']/((10/11)^2(\eps')^2
\E^2[X'])\\
&\le& 5\cdot A^2 \eps^2/((10/11)^2(\eps')^2 (9/10)^2 A^2)\\
&<& (13/2)\eps^2/(10\eps'/11)^2\\
&<& 1/4
\end{eqnarray*}

Thus, with probability at least $1/4$, by the triangle inequality and
\Lemma{ind} we have
$|X'-\E[X]| \le |X' - \E[X']| + |\E[X'] - \E[X]|\le
8\eps\E[X]$.
\end{proof}

\subsubsection{Proofs from \Section{f0}}\SectionName{f0-proofs}

Here we provide the proofs of two lemmas used in the analysis of our
$F_0$ algorithm in \Section{f0}.

\begin{proofof}{\Lemma{bounded-deriv}}
We calculate
\begin{eqnarray*}
f'(y) &=& x\ln\left(1 - \frac 1x\right)\left(1 - \frac 1x\right)^y -
2x\ln\left(1 - \frac 1x\right)\left(1 - \frac 1x\right)^{2y}\\
&=& x\ln\left(1 + \frac 1{x-1}\right)\left(1 - \frac 1x\right)^y
\left[ 2\left(1 - \frac 1x\right)^y - 1\right ]\\
 &\ge& \frac{x}{2(x-1)}\left(1 - \frac 13\right)\left[ 2\left( 1 -
     \frac 13\right) - 1\right]\\
&\ge& \frac 12\cdot \frac 23 \cdot \frac 13
\end{eqnarray*}
\end{proofof}

\begin{proofof}{\Lemma{AequalsB}}
We use $B \pm \eps B$ to denote a value in $[(1-\eps)B, (1+\eps)B]$.
Then,
\begin{eqnarray*}
\left(1 - \frac 1K\right)^{B'} &=& \left(1 - \frac
  1K\right)^{B}\left(1 - \frac
  1K\right)^{\pm\eps B}\\
&\le& \left(1 - \frac 1K\right)^{B}\cdot\frac{1}{1 - \frac {\eps
    B}{K}}\\
&\le& (1 + 2\eps)\left(1 - \frac 1K\right)^{B}
\end{eqnarray*}
Also,
\begin{eqnarray*}
\left(1 - \frac 1K\right)^{B'} &=& \left(1 - \frac
  1K\right)^{B}\left(1 - \frac
  1K\right)^{\pm\eps B}\\
&\ge& \left(1 - \frac 1K\right)^{B}\cdot\frac{1}{1 - \frac {\eps
    B}{K}}\\
&\ge& \left(1 - \frac BK\right)\left(1 - \frac 1K\right)^{B}\\
&\ge& (1 - \eps)\left(1 - \frac 1K\right)^{B}
\end{eqnarray*}
\end{proofof}

\end{document}